\documentclass{article}[11pt, margin=1in]

\usepackage[english]{babel}

\usepackage[dvipsnames]{xcolor}

\def\mB{\mathcal{B}}
\def\mC{\mathcal{C}}
\def\mD{\mathcal{D}}
\def\mE{\mathcal{E}}
\def\mF{\mathcal{F}}

\def\mH{\mathcal{H}}
\def\mI{\mathcal{I}}

\def\mK{\mathcal{K}}

\def\mN{\mathcal{N}}

\def\mS{\mathcal{S}}
\def\mT{\mathcal{T}}

\def\mY{\mathcal{Y}}

\usepackage{a4wide}
\usepackage{color}
\usepackage{graphicx}
\usepackage{xspace}
\usepackage{url}
\usepackage{amsmath,amssymb, amsthm}
\usepackage{microtype,mathtools,mathrsfs}
\usepackage{hyperref}
\usepackage{authblk}
\usepackage{tikz}
\usetikzlibrary{fit}
\usepackage[noabbrev,capitalise,nameinlink]{cleveref}
\usepackage{thm-restate}
\usepackage{complexity}
\usepackage{dsfont}
\usepackage{array}

\newtheorem{theorem}{Theorem}

\newtheorem{lemma}[theorem]{Lemma}

\newtheorem{corollary}[theorem]{Corollary}
\newtheorem{proposition}[theorem]{Proposition}
\newtheorem{claim}[theorem]{Claim}

\theoremstyle{remark}
\newtheorem{remark}[theorem]{Remark}
\def\Box{\hbox{\hskip 1pt \vrule width 4pt height 8pt depth 1.5pt \hskip 1pt}}
\renewenvironment{proof}[1][Proof. ]{\medskip\noindent\textbf{#1}}{{}\hfill$\Box$\\}

\newcommand{\norm}[1]{\left\lVert#1\right\rVert}

\title{A Dense Neighborhood Lemma: Applications of Partial Concept Classes to Domination and Chromatic Number}

\author[1]{Romain Bourneuf}
\affil[1]{Univ. Bordeaux, CNRS, Bordeaux INP, LaBRI, UMR 5800, F-33400 Talence, France.}
\author[2]{Pierre Charbit}
\affil[2]{Université de Paris, CNRS, IRIF, F-75006, Paris, France.}
\author[3]{Stéphan Thomassé}
\affil[3]{Univ. Lyon, ENS de Lyon, UCBL, CNRS, LIP, France.}

\date{}

\begin{document}

\maketitle
\begin{abstract}
In its Euclidean form, the Dense Neighborhood Lemma (DNL) asserts that if $V$ is a finite set of points of $\mathbb{R}^N$ such that for each $v \in V$ the ball $B(v,1)$ intersects $V$ on at least $\delta |V|$ points, then for every $\varepsilon >0$, the points of $V$ can be covered with $f(\delta,\varepsilon)$ balls $B(v,1+\varepsilon)$ with $v \in V$.
DNL also applies to other metric spaces and to abstract set systems, where elements are compared pairwise with respect to (near) disjointness. 
In its strongest form, DNL provides an $\varepsilon$-clustering with size exponential in $\varepsilon^{-1}$, which amounts to a Regularity Lemma with 0/1 densities of some \emph{trigraph}.

Trigraphs are graphs with additional \emph{red} edges.  
They are natural instances of \emph{partial concept classes}, introduced by Alon, Hanneke, Holzman and Moran [FOCS 2021]. This paper is mainly a combinatorial study of the generalization of Vapnik-Cervonenkis dimension to partial concept classes. The main point is to show how trigraphs can sometimes explain the success of random sampling even though the VC-dimension of the underlying graph is unbounded. 
All the results presented here are effective in the sense of computation: they primarily rely on uniform sampling with the same success rate as in classical VC-dimension theory.

Among some applications of DNL, we show that $\left(\frac{3t-8}{3t-5}+\varepsilon\right)\cdot n$-regular $K_t$-free graphs have bounded chromatic number. Similarly, triangle-free graphs with minimum degree $n/3-n^{1-\varepsilon}$ have bounded chromatic number (this does not hold with $n/3-n^{1-o(1)}$). For tournaments, DNL implies that the domination number is bounded in terms of the fractional chromatic number. Also, $(1/2-\varepsilon)$-majority digraphs have bounded domination, independently of the number of voters.
\end{abstract}

\section{Introduction}

A classical way to form a graph from a metric space $(X,d)$ is to fix a threshold $\tau$ (usually 1) and connect two vertices by an edge if they are at distance at most $\tau$. 
To get a finite graph, we consider the induced subgraph on some finite subset $V$ of $X$. 
The study of the properties of such (\emph{threshold}) graphs is generally delicate as some edges or non-edges may be very close to the threshold and thus prone to being flipped by a small perturbation. 
A convenient way to overcome this is to consider \emph{trigraphs} instead of graphs, in which a new edge status, \emph{red edges}, is introduced. 
In general, a trigraph $T=(V,E,R)$ consists of a set $V$ of vertices, a set $E$ of edges, and a set $R$ of \emph{red} edges disjoint from $E$. 
Allowing a second status for edges offers a buffer zone which is very convenient to study graphs both from a structural and from an algorithmic perspective.

Trigraphs were first introduced in \cite{Chud06} to study Berge Graphs, and red edges represented undecided edges. Trigraphs also appear in the definition of twin-width~\cite{BKTW21}, and red edges represent errors occurring when both an edge and a non-edge connect two contracted sets of vertices. 
In both cases the red edges in some sense represent a frontier between the edges and the non-edges. This is exactly the nature of pairs of vertices at distance close to the threshold in a metric space. 
However, whereas red edges are neutral objects for Berge trigraphs and bad features for twin-width, in this paper we view them as assets to bridge the integrality gap. 
Indeed, we allow the red edges of a trigraph to be used to form a minimum dominating set (making it smaller) while not considering them when computing the fractional domination number. 
With this convention, the gap between fractional domination and domination is easier to bridge in trigraphs than in graphs.

In this paper, we will focus on two main types of trigraphs. The most intuitive ones, \emph{metric-trigraphs}, are obtained by selecting a metric space $(X,d)$, a threshold $\tau$ and a sensitivity $\varepsilon>0$, and connecting two points $x,y\in X$ by an edge if $d(x,y)\leq \tau$ and by a red edge if $\tau<d(x,y)\leq \tau+\varepsilon$. A less intuitive class of trigraphs, \emph{disjointness-trigraphs}, are more adapted to the study of graphs. 
They are based on set systems. 

Let us now illustrate the main point of the paper on the first result of the abstract, which, in the Euclidean setting, explores the structure of sets of points in which every neighborhood is dense.
Here, when $V$ is a set, a \emph{$V$-ball} is a ball with center in $V$. The next result is proved in \cref{sec:euclidean}. 

\begin{restatable}{lemma}{DNLeuclidean}\label{lem:realdensecor}
Let $V$ be a finite subset of $\mathbb{R}^N$. If every $V$-ball of radius 1 intersects $V$ on at least $\delta |V|$ points, then for all $\varepsilon >0$, the set $V$ can be covered by $\textup{poly}(\delta^{-1},\varepsilon^{-1})$ $V$-balls with radius $1+\varepsilon$.
\end{restatable}  

This Dense Neighborhood Lemma is a well-known result when the dimension $N$ is bounded, where it even holds with $\varepsilon = 0$.
Indeed, in that case the VC-dimension is bounded, hence any $\delta$-net provides such a covering. 
This is no longer true when $N$ is unbounded, hence the relaxation on the radius of the balls of the covering. 
The reason why it is possible to trade dimension for larger radius is easily interpreted in the associated metric-trigraph $T$ with threshold $1$ and sensitivity $\varepsilon$. 
The key observation is that red edges are needed to shatter a large subset of $T$ using neighborhoods.
A more intuitive point of view is that no large random-like bipartite subgraph $G$ of $T$ can be constructed without using red edges. 
The reason for this is that the gap of $\varepsilon$ between the length of edges and non edges would produce an excessively sparse cut in $G$. The key point then is that VC-dimension theory applies to trigraphs with no large shattered set, where we do not allow red edges when shattering a set. 
Hence we can again find $\delta$-nets, just like when the dimension $N$ is bounded.

We could have chosen another metric space. Natural examples include the \emph{spherical setting} where vertices belong to the unit sphere $\mathbb S^N$ and the distance between $x$ and $y$ is the angle $xOy$, or the \emph{Hamming setting} in which vertices are 0/1 valued words and the distance is the Hamming distance. All these cases admit equivalent statements of \cref{lem:realdensecor}. Hamming distance is well-suited for combinatorial applications (we can also show the main results of this paper via Hamming distance), but using disjointness-trigraphs gives better bounds and smoother proofs. 

Let $\mF$ be a set system on ground set $V$. Given a vertex $x$, we denote by $\mF_{x}$ the set of all sets in $\mF$ which contain $x$, and write $\mF_{xy}=\mF_x\cap \mF_y$.
We denote by $D_\varepsilon(x)$ the set $\{y\in V : |\mF_{xy}|\leq \varepsilon |\mF|\}$. When $\varepsilon =0$, we simply write $D(x)\coloneqq D_0(x)$. By extension, when $X \subseteq V$ we write  $D_\varepsilon(X)=\bigcup_{x\in X} D_\varepsilon (x)$. The notation $D$ stands for disjoint.
An \emph{$\varepsilon$-covering} of $\mF$ is a subset $X$ of vertices such that $D_\varepsilon(X)=V$.
The \emph{disjointness-ratio} of $\mF$ is the minimum of $|D(x)|/|V|$ over all $x$ in $V$.

\begin{restatable}{lemma}{DNLSanti}\label{lem:DNL-set-anti}
Every set system with disjointness-ratio $\delta$ has an $\varepsilon$-covering of size $O\left(\frac{1}{\varepsilon \delta}\cdot\log \frac{1}{\delta}\right)$.
\end{restatable} 

We can apply \cref{lem:DNL-set-anti} in the particular case where the set system is the set of neighborhoods of a graph $G$. For any $\varepsilon\geq 0$ and any vertex $v$, define $D_\varepsilon(v)=\{u\in V(G):\ |N(v)\cap N(u)|\leq \varepsilon |V(G)|\}$.

\begin{restatable}{lemma}{DNLGanti}\label{lem:DNL-graph-anti}
Let $G$ be a graph such that $|D_0(v)|\geq \delta |V(G)|$ for any $v\in V(G)$.  Then for any $\varepsilon>0$ there exists a set $X$ of size $O\left(\frac{1}{\varepsilon \delta}\cdot\log \frac{1}{\delta}\right)$ such that $\bigcup_{x\in X} D_\varepsilon(x)=V$.
\end{restatable}

Note that \cref{lem:DNL-graph-anti} does not only apply in the linear-density setting, but also provides meaningful information to all graphs in which the minimum degree of disjointness is at least $\log^2 n$.
An application of \cref{lem:DNL-graph-anti} to graph theory is exemplified in the following result, first proved by Thomassen~\cite{CT02} for $\varepsilon>0$. The very short proof shows how DNL naturally applies to combinatorics.

\begin{restatable}{theorem}{thomassen}\label{thm:thomassen}
Every $n$-vertex triangle-free graph $G=(V,E)$ with minimum degree $(1/3+\varepsilon)n$ has chromatic number $O\left(\frac{1}{\varepsilon}\right)$.
\end{restatable} 

\begin{proof}
Note that $N(v)\subseteq D_0(v)$ for all $v\in V$ since $G$ is triangle-free, so $|D_0(v)|\geq n/3$. 
Let $X$ be the set obtained from \cref{lem:DNL-graph-anti} with sensitivity $\varepsilon$. 
For every vertex $x$, there is no edge $uv$ inside $D_\varepsilon (x)$ since we would then have $|N(u)\cup N(v)\cup N(x)|\geq 3(1/3+\varepsilon) n-2\varepsilon n=n+\varepsilon n$. 
So $V$ is covered by $|X|=O\left(\frac{1}{\varepsilon}\right)$ independent sets.
\end{proof}

We can alternatively see this proof in the following way: form the trigraph $T=(V,E,R)$ where $E=\{xy: N(x)\cap N(y)=\emptyset\}$  and $R=\{xy: 1\leq |N(x)\cap N(y)|\leq \varepsilon n\}$.
The key is the following \emph{randomness transversal argument}. 
In any large random-like bipartite subgraph of $G$ with parts $A$ and $B$, there must exist some non-edge $ab$ with $a\in A$ and $b\in B$ such that $a$ and $b$ have at most $\varepsilon n$ common neighbors (otherwise $G[A,B]$ would contain the complement of an excessively large biclique).
Therefore $ab$ is either an edge or a red edge of $T$.
In particular, $T$ is a trigraph in which there is no large random-like bipartite subgraph avoiding red edges (in other words, the red edges form an edge-transversal of large random-like bipartite subgraphs). Since $G$ has minimum degree more than $n/3$, its fractional domination is less than 3 and thus by an adaptation to trigraphs of a classical result in VC-dimension, $G$ has a small dominating set $X$ (possibly using red edges). Thus, every vertex shares at most $\varepsilon n$ common neighbors with some vertex of $X$, and we conclude as previously. 
This strategy can be thought of as starting from a graph with potentially large VC-dimension and turning it into a trigraph with bounded VC-dimension. We will see several other applications of this strategy.

The study of the VC-dimension of trigraphs (under the name of \emph{partial concept classes}) was launched by Alon, Hanneke, Holzman and Moran \cite{AHHM21}. In particular, they proved that PAC learnable partial concept classes are exactly those with bounded VC dimension. In this paper, we give several examples of partial concept classes with bounded VC-dimension, thereby extending the scope of \cite{AHHM21}. The authors of \cite{AHHM21} also emphasized the additional expressivity of partial concept classes when compared to concept classes. We observe a similar phenomenon with the trigraph point of view of metric spaces, which reduces distances to only three values (close, intermediate, far). 
This offers much more flexibility than merely threshold graphs, and opens the way to the more general translation of graph parameters to trigraphs. 
We study here VC-dimension, but the general idea of using red edges to concentrate the difficulty of the problem (here the integrality gap of domination) could be applied to other optimization questions.

DNL also admits a much stronger clustering form, in which the set $V$ is partitioned into clusters consisting of similar elements.
This clustering version directly implies a result of O'Rourke~\cite{O'R14} asserting that every $(1/4+\varepsilon)n$-regular triangle-free graph has bounded chromatic number. Our bound is exponential in 
$1/\varepsilon$ rather than tower-type in O'Rourke's proof, and the argument is considerably simpler. 
This ease of use of DNL allows us to extend this result: every $\left(\frac{3t-8}{3t-5}+\varepsilon\right)\cdot n$-regular $K_t$-free graph has bounded chromatic number (and the constant is optimal). 
Another application of the clustering form of DNL allows us to strengthen the $n/3$ threshold of triangle-free graphs: minimum degree $n/3-n^{1-\varepsilon}$ is enough to obtain bounded chromatic number, and this is essentially best possible. 
We also apply DNL to homomorphism thresholds of $K_t$-free graphs.

The second area in which DNL naturally applies is tournaments, which are dense structures per se. Using again a randomness transversal argument, we show that in tournaments, domination is bounded in terms of fractional chromatic number. 
These two very basic parameters were not known to be related. 
This implies by Farkas Lemma that if a tournament $T$ satisfies that all its out-neighborhoods have bounded chromatic number, then $T$ has bounded domination, hence bounded chromatic number. 
This was already known with a tower-type bound~\cite{HLTW19}, the use of DNL reduces it to factorial-type. 
Finally, going back to one of the very first application of VC-dimension to graphs, we extend a result of~\cite{ALON2006374} on majority digraphs.

We feel that these applications are just a glimpse of the spectrum of use of Dense Neighborhood type Lemmas, and more generally trigraphs for optimization problems. Here is a detailed overview of the paper.

\subsection{Examples of Dense Neighborhood Lemma}

In all the various forms of DNL, the central definition is always the definition of the neighborhood of a vertex. It can simply be based on the distance in the metric case, or on some counting as in the case of disjointness-trigraphs. 
In all cases, two different sizes of neighborhoods are considered, with a \emph{gap} depending of some $\varepsilon$, and DNL asserts that if the small neighborhoods are dense (i.e. contain a linear fraction of the vertex set), then a bounded union of large neighborhoods covers the whole vertex set. 
In the language of optimization this can be reformulated as: the minimum cover by large neighborhoods is bounded by a function of the minimum fractional cover by small neighborhoods. A series of examples of such small/large neighborhoods is listed in \cref{table:exmples-nbhds}.

\newcolumntype{M}[1]{>{\centering\arraybackslash}m{#1}}
\newcolumntype{N}{@{}m{0pt}@{}}

\begin{table}[h]
\centering
\begin{tabular}{M{4cm}|M{5cm}|M{5cm} N}
\textbf{Structure} & \textbf{Small neighborhood of $v$} & \textbf{Large neighborhood of $v$} \\[1pt]
\hline
&&&\\[1pt]
$V \subseteq \mathbb{R}^N$ & $B(v, 1) = \{u : d(u, v) \leq 1\}$ & $B(v, 1+\varepsilon)$\\[10pt]
\hline
&&&\\[1pt]
$V \subseteq \mathbb{S}^N$ & $\Bigl\{ u : \langle u, v \rangle \geq \tau \Bigr\}$ & $\Bigl\{ u : \langle u, v\rangle \geq \tau - \varepsilon \Bigr\}$\\[10pt]
\hline
&&&\\[1pt]
$V \subseteq \{0,1\}^N$ & $\Bigl\{ u : d_H(u,v) \geq \tau \cdot N \Bigr\}$ & $\Bigl\{ u : d_H(u, v) \geq (\tau - \varepsilon) \cdot N \Bigr\}$\\[10pt]
\hline
&&&\\[1pt]
Set system $\mF = (V, E)$ & $\{u : \mF_u \cap \mF_v = \emptyset \}$ & $\{u : |\mF_u \cap \mF_v| \leq \varepsilon|\mF|\}$\\[10pt]
\hline
&&&\\[1pt]
Graph $G = (V,E)$ & $\{ u : N(u) \cap N(v) = \emptyset\}$ & $\{ u : |N(u) \cap N(v)| \leq \varepsilon n\}$\\[10pt]
\hline
&&&\\[1pt]
Digraph $D = (V, A)$ & $\{u : N^-(u) \cap N^+(v) = \emptyset\}$ & $\{ u : |N^-(u) \cap N^+(v)| \leq \varepsilon n\}$\\[10pt]
\hline
&&&\\[1pt]
0-1 random variables $X_v$ & $\{X_u : \mathbb{P}[X_u = X_v = 1] \geq \tau \}$ & $\{X_u : \mathbb{P}[X_u = X_v = 1] \geq \tau - \varepsilon \}$\\[10pt]
\hline
&&&\\[1pt]
Majority voting on $V$ & $\{ u : v \text{ is } 1/2 \text{-preferred to } u \}$ & $\{ u : v \text{ is } (1/2 - \varepsilon) \text{-preferred to } u \}$\\[5pt]
\end{tabular}
\caption{Typical examples of small and large neighborhoods in various structures. 
}
\label{table:exmples-nbhds}

\end{table}

DNL also holds when considering complements of neighborhoods, as illustrated in the following \emph{antipodal} Hamming version, where the Hamming distance is denoted by $d_H$. 

\begin{restatable}{lemma}{DNLHamanti}\label{lem:DNL-Ham-antipodal}
Let $V \subseteq \{0,1\}^N$ be such that for every $v \in V$, the set $\{u\in V:d_H(v,u)\geq c\cdot N\}$ has size at least $\delta |V|$. Then, there is a set $X \subseteq V$ of size $\textup{poly}(\varepsilon^{-1}, \delta^{-1})$ such that for every $v\in V$, there exists $x\in X$ such that $d_H(x,v)\geq (c-\varepsilon)N$.
\end{restatable} 

For set systems, we saw how to apply \cref{lem:DNL-set-anti} when the sets are the neighborhoods of a graph $G$ on $n$ vertices. 
In that case, the small neighborhood $D(v)$ is the set of vertices $u$ such that $N(u)\cap N(v)=\emptyset$, in other words there is no path of length 2 from $u$ to $v$. The large neighborhood $D_\varepsilon(v)$ consists of vertices $u$ admitting at most $\varepsilon n$ such paths. This surprising choice of neighborhoods becomes clear in the context of triangle-free graphs: the set $D(v)$ contains $N(v)$. DNL then asserts that every triangle-free graph with linear degree contains a bounded size set $X$ such that every vertex shares very few common neighbors with some vertex in $X$.

Interestingly, DNL also holds for the oriented analogue of the previous result. In a directed graph we can define $D^+(v)$ as the set of vertices non reachable by a directed path of length 2 from $v$ (and define $D_\epsilon^+(v)$ analogously). Again, the minimum cover by sets of the form $D_\epsilon^+(v)$ is bounded in terms of the minimum fractional cover by sets of the form $D^+(v)$. 

For another example, consider $n$ boolean random variables $(X_v)_{v\in V}$ and say that $X_u,X_v$ are \emph{$\tau$-close} if $\mathbb{P}[X_u=X_v=1]\geq \tau$. Then if every $X_v$ is $\tau$-close to $\delta n$ variables, there is a set $Y$ of bounded size such that every $X_v$ is $(\tau-\varepsilon)$-close to some $X_y$ with $y\in Y$.

Let us finish with a voting process where an odd set of voters ranks a set of applicants. For every applicant $x$, its small neighborhood $N_{1/2}(x)$ is the set of applicants $y$ such that at least half (hence the majority) of the voters prefer $x$ to $y$. Its large neighborhood $N_{1/2-\varepsilon}(x)$ is the set of applicants $y$ such that at least a $(1/2-\varepsilon)$-fraction of the voters prefer $x$ to $y$. An application of Farkas Lemma shows that the minimum fractional cover by sets $N_{1/2}(x)$ is at most 2. DNL then implies that a bounded number of large neighborhoods $N_{1/2-\varepsilon}(x)$ covers all applicants.

\subsection{VC-dimension of Trigraphs}

We consider \emph{trigraphs} $T=(V,E,R)$ where $V$ is the set of vertices, $E$ is the set of (\emph{plain}) edges, and $R$ is the set of \emph{red} edges, with $R\cap E=\emptyset$. We will always denote by $n$ the number of vertices. In graphs, we usually speak of ``non-edges'' to indicate pairs of nonadjacent vertices, we do the same for trigraphs, and denote by $W$ the set of non-edges. This stands for ``without edge'' or more visually the fact that two nonadjacent vertices drawn on a whiteboard are connected by a ``white edge''. Therefore $(E,R,W)$ is a partition of the set of pairs of vertices. Given a vertex $v$, we denote by $R(v)$ its set of red neighbors and by $W(v)$ its set of non-neighbors. We still denote the (plain) neighborhood of $v$ by $N(v)$, or by $N[v]$ for the closed neighborhood, including $v$. 

Let us reinterpret some classical graphs parameters. The \emph{minimum degree} $\delta(T)$ of $T$ is the minimum size of $N(v)$ over all $v\in V$. A \emph{dominating set} of $T$ is a set $X$ of vertices such that for every $y \in V$, there exists $x \in X$ such that $y\in N[x]\cup R(x)$.
We say that a subset $X$ of $V$ is \emph{shattered} if there exists a set $S$ of $2^{|X|}$ vertices with no red neighbor in $X$ such that for every $Y\subseteq X$, there exists $v_Y\in S$ such that $N[v_Y]\cap X=Y$. The \emph{VC-dimension} of a trigraph is the largest size of a shattered set.

The Sauer-Shelah Lemma, see \cite{Sauer72,Shelah72}, asserts that when a (usual) graph has  VC-dimension $d$, then for every subset $X$ of vertices, the number of distinct neighborhoods of the vertices of $V$ on $X$ is at most $|X|^d$. 
Sadly, we were not able to adapt this result for trigraphs of bounded VC-dimension. 
Consequently, we could not lift some classical results from VC-theory including notably approximations. However, the existence of $\delta$-nets, asserting that every trigraph with minimum degree $\delta n$ and bounded VC-dimension has bounded domination, is still valid, see \cref{thm:triHW}.

Let us now argue why trigraphs with bounded VC-dimension are natural combinatorial objects. The following proof is based on the random hyperplane argument of Goemans and Williamson~\cite{Goe95}. Suppose that $V$ is a set of points in the $N$-dimensional sphere $\mathbb{S}^N$. We measure the distance $d(x,y)$ between two vertices $x,y$ of $V$ by the angle $x0y$, or equivalently by the length of the circular arc $xy$ (e.g. $\pi$ for antipodal points). Fix some $\varepsilon >0$ and some angle $c \in [\varepsilon ,\pi]$ and form the trigraph $T=(V,E,R)$ where $E$ consists of all pairs of vertices $x,y$ such that $d(x,y)\geq c$, and $R$ consists of all pairs such that $c-\varepsilon \leq d(x,y)<c$. We claim that $T$ has bounded VC-dimension.

To see it, consider a shattered set $X$ of $T$ and its corresponding shattering set $S$ of size $2^{|X|}$. Recall that $S$ contains vertices realizing all possible neighborhoods on $X$, without any red edge. Draw from $S$ a random subset $Z$ of size $|X|$.
Observe that (with high probability) $Z$ is disjoint from $X$ and that the set of edges between $X$ and $Z$ forms a random (semi induced) bipartite graph $G$ on $X \cup Z$ (say with $|X| \cdot|Z|/2$ edges). 
Consider a random linear hyperplane $P$ in $\mathbb{R}^{N+1}$. Every pair $(x, z)$ is separated by $P$ with probability $d(x,z)/\pi$. 
This means that in expectation, the number of edges $xz$ which are cut by $P$ is at least $c \cdot|X|\cdot|Z|/(2\pi)$, and the number of non-edges cut by $P$ is at most $(c-\varepsilon) \cdot|X|\cdot|Z|/(2\pi)$. 
This cannot happen in a large random bipartite graph, hence showing that the size of $X$, and hence the VC-dimension of $T$ is bounded. A more careful analysis  gives a $O\left(1/\varepsilon ^2\right)$ bound, see \cref{thm:VCdim-Hamming,,thm:VCdim-spherical} in \cref{sec:vcdimtri}.

In the case of set systems on a ground set $V$, we can also form a (disjointness) trigraph $T = (V, E, R)$ by letting $uv\in E$ if $v\in D(u)$ and $uv\in R$ if $v\in D_\varepsilon(u)\setminus D(u)$. The bound on the VC-dimension, proved in \cref{lem:VCdim-disj-trigraph}, is much cleaner: it is at most $\frac{1}{\varepsilon}$. Because of that, disjointness-trigraphs are often better choices than the more intuitive metric-trigraphs.

We generically call a trigraph \emph{gap-representable} if it is obtained from the above constructions based on metric spaces and set systems. 

\subsection{Tri-hypergraphs}

Tri-hypergraphs offer an alternative point of view to trigraphs, which is sometimes more adapted to some specific proofs. 
The idea is very simple: instead of considering that a hyperedge is a subset of vertices (hence implicitely describing a bi-partition), a hyperedge in a tri-hypergraph is a tripartition $(B,R,W)$, standing for black, red, white, of the vertex set $V$. Therefore, a \emph{tri-hypergraph} $H$ is a pair $(V,\mathcal{E})$ where $V$ is a finite set and $\mathcal{E}$ is a collection of partitions of $V$ into three (possibly empty) parts $(B, R, W)$. Each such triple is called a \emph{tri-edge} of $H$. A standard hypergraph can be seen as a tri-hypergraph in which $R = \emptyset$ for every $(B, R, W)\in \mE$. When defining a tri-edge, we will often only define $B$ and $R$ explicitly, $W$ will then implicitly be defined as the complement of $B \cup R$.

Tri-hypergraphs are a combinatorial way of defining partial concept classes~\cite{AHHM21}. Indeed, each tri-hyperedge $(B, R, W)$ can be interpreted as a partial concept $c: V \to \{0, 1, \star\}$ with $c^{-1}(0) = W$, $c^{-1}(1) = B$ and $c^{-1}(\star) = R$. The definition of VC-dimension for tri-hypergraphs is then exactly the same as the corresponding definition for partial concept classes, see \cite{AHHM21}.

A set of vertices $X \subseteq V$ is a \emph{transversal} (or \emph{hitting set}) of $H$ if  $X \cap (B \cup R) \neq \emptyset$ for every $(B, R, W)\in \mE$. One of the most central optimization problems on hypergraphs is to compute a minimum (or at least a small) transversal. 
This is a considerably easier task when the VC-dimension is bounded, as uniform sampling then gives very good results. The same holds for tri-hypergraphs, for a suitable definition of VC-dimension. 
A notion closely related to tri-hypergraphs, \emph{paired hypergraphs}, was introduced in \cite{Luczak2010coloringdensegraphsvcdimension} to study chromatic thresholds, see also \cite{ABGKM13}. 
However, the transversality of paired hypergraphs was not investigated. In a sense, this paper continues this study with a slightly more adapted structure.

Given a tri-hypergraph $H=(V,\mE)$ and a set $X \subseteq V$ we define the \emph{trace} of $H$ on $X$ as 
$$tr_H(X)=\{ Y\subseteq X \ |\ \exists (B, R, W)\in \mE \text{  such that }  Y=X\cap B=X\cap (B \cup R) \}.$$
A set $X$ is \emph{shattered} by $H$ if $tr_H(X)=2^{X}$. The \emph{VC-dimension} of $H$ is the largest size of a shattered set. 
Note that when $R = \emptyset$ for every tri-edge, we recover the usual definition of VC-dimension and transversal for hypergraphs. 

Given a trigraph $T=(V,E,R)$ we can form a tri-hypergraph $H_T$ on $V$ with the tri-edges $(N[v],R(v),W(v))$ for all vertices $v\in V$. Note that $T$ and $H_T$ have the same VC-dimension. Also, the dominating sets of $T$ are exactly the transversals of $H_T$. This allows us to conveniently switch between the different points of view. 

\subsection{Domination in Trigraphs}

Given a set $X$ of vertices in a trigraph $T=(V,E,R)$, a partition $(X_0, X_1)$ of $X$ is a \emph{clean separation} if there exists a vertex $v$ such that $W(v) \cap X=X_0$ and $N[v] \cap X=X_1$. 
The Sauer-Shelah Lemma can be restated as: if $T$ has VC-dimension $d$, the number of clean separations of a subset $X$ of vertices is at most $|X|^d$. 
It is crucial here that no red edge is involved in the separations, i.e. that $T$ induces a proper graph between $X$ and the vertices forming the clean separations.
 
We now show that every trigraph with bounded VC-dimension and linear minimum degree has a small dominating set. This is a generalization of a classical result of Haussler and Welzl~\cite{HW87}, but the usual proof involves picking two samples of equal size, which is not suited for trigraphs as we want to avoid having to consider red edges. We provide here an alternative proof of independent interest. 
 
Haussler-Welzl type results are based on sampling, and it is simpler to think of samples as sets rather than as multisets. To this end, the following trick is very convenient: Substitute all vertices by very large independent sets of equal sizes. 
Both the degree ratios (i.e. $\deg(x)/n$) and the VC-dimension are unchanged, and the domination does not decrease. Samples can now be almost surely considered as sets.
Another convenient trick to manipulate inequalities involving logarithms is the following basic calculus fact, where $a\geq 0$ and $x\geq 1$  : if $x\leq a\ln(x)$, then
$x\leq 2a\ln(a)$, see \cref{lem:tech} for a more general form.
Ceilings and floors in the following proof are hidden in the $O(\cdot)$ notation. The short proof we present here gives a slightly weaker bound $O(\frac{d}{\delta}\cdot\log \frac{d}{\delta})$. The $O(\frac{d}{\delta}\cdot\log \frac{1}{\delta})$ bound is proved in \cref{sec:trans}.

\begin{restatable}{theorem}{trigraphHW}\label{thm:triHW}
Every $n$-vertex trigraph $T$ with minimum degree $\delta n$ and VC-dimension $d$ has domination number  $\gamma (T)=O(\frac{d}{\delta}\cdot\log \frac{1}{\delta})$
\end{restatable} 

\begin{proof}
Let $k > \frac{4d}{\delta}\ln \frac{2d}{\delta}$ and assume for contradiction that $\gamma (T)> k$. Consider a random subset $X$ of vertices, of size $i\leq k$. There exists a vertex $v$ such that $N[v]\cup R(v)$ is disjoint from $X$. Since $|N(v)|\geq \delta n$, if we draw another random subset $X'$ of size $k-i$, we have $\mathbb{P}[X'\subseteq N(v)]\geq \delta ^{k-i}$, in which case $(X, X')$ is a clean separation of $X\cup X'$.
By double counting, the average number of clean separations of a set $Z$ of size $k$ is at least $\sum_{i=0}^k \binom{k}{i} \delta ^{k-i}=(1+\delta)^k$. 
By the Sauer-Shelah Lemma,  $(1+\delta)^k \leq k^d$ or equivalently $k\ln(1+\delta)\leq d\ln k$. Since $2\ln(1+\delta)\geq \delta$, we have $k\delta\leq 2d\ln k$. By the calculus trick $k\leq \frac{4d}{\delta}\ln \frac{2d}{\delta}$, a contradiction.
\end{proof}

A similar argument shows that the transversal number $\tau$ (with respect to $B\cup R$) of a tri-hypergraph is bounded in terms of the linear programming relaxation $\tau^*$ (with respect to $B$), and its VC-dimension $d$.

\begin{restatable}{theorem}{boundedintegralitygap}\label{thm:bounded-integrality-gap}
Every tri-hypergraph of VC-dimension $d$ satisfies $\tau = O(d\tau^*\log \tau^*)$.
\end{restatable} 

The previous results focus on dominating all vertices in trigraphs or intersecting all hyperedges in tri-hypergraphs, but we can ask to only dominate large degree vertices, or intersect large hyperedges. A set $X$ of vertices of a tri-hypergraph $H=(V,\mE)$ is a \emph{$\delta$-net} if $X \cap (B \cup R) \neq \emptyset$ for every $(B, R, W) \in \mE$ such that $|B| \geq \delta|V|$. We use here the term $\delta$-net instead of $\varepsilon$-net since we already use $\varepsilon$ for the distance slack. In the next result about $\delta$-nets, we explicitly give the probability that a uniform sample forms a $\delta$-net for a tri-hypergraph of bounded VC-dimension.

\begin{restatable}{theorem}{randomsample}\label{thm:randomsample}
Let $H=(V,\mE)$ be a tri-hypergraph of VC-dimension $d$ and let $\delta > 0$. 
With probability $1-p$, a random sample of $\tau_p$ elements from $V$ is a $\delta$-net of $H$, where
$$
\tau_p = 2 \cdot \frac{d}{\ln(1+\delta)} \cdot \ln\left(\frac{e}{\ln(1+\delta) \cdot p^{1/d}}\right).
$$
\end{restatable} 

\subsection{Clustering}

In VC-theory, the natural step after transversals is clustering. In the case of usual graphs with bounded VC-dimension, the key-observation is that the hypergraph of differences of neighborhoods, that is all $N(x)\setminus N(y)$, still has bounded VC-dimension.
Hence taking an $\varepsilon$-net $X$ of this hypergraph, and partitioning the vertices according to their neighborhoods on $X$, leads to a partition in which clusters consist of vertices which are near twins. This corresponds to an ultra-strong regularity partition which can be a good partition to consider when VC-dimension is not bounded.

We tried in vain to obtain a clustering version of DNL for general trigraphs with bounded VC-dimension. We show that VC-dimension remains bounded under intersections and differences. The difficulty then resides mainly in the classification step since red edges in the difference are not refined enough to distinguish vertices. 
Fortunately, a clustering result holds in the case of gap-representable trigraphs. The reason is that they offer the opportunity to partition further using the midpoint threshold $\varepsilon/2$. 
This extra precision enables to classify the vertices and then get a clustering. Here are three examples of clustering versions of DNL:

Let $\mF$ be a set system on ground set $V$. 
An \emph{$\varepsilon$-cluster} is a subset $X$ of $V$ such that $|D(x)\setminus D_\varepsilon (y)|\leq \varepsilon |V|$ for all $x,y\in X$. An \emph{$\varepsilon$-clustering} of $\mF$ is a partition of $V$ into $\varepsilon$-clusters. The \emph{size} of a clustering is its number of clusters. 

\begin{restatable}{lemma}{DNLScluster}\label{lem:DNLS-cluster-intro}
Every set system has an $\varepsilon$-clustering of size $2^{\textup{poly}(\varepsilon^{-1})}$.
\end{restatable} 

Here is the Hamming clustering: 

\begin{lemma}\label{lem:DNLH-cluster-intro}
Let $V$ be a subset of $\{0,1\}^N$. There is a partition of $V$ into $2^{\textup{poly}(\varepsilon^{-1})}$ clusters such that if $u, v$ are in the same cluster, there are at most $\varepsilon \cdot |V|$ points $w$ of $V$ such that $d_H(u,w)\geq c\cdot N$ and $d_H(v,w)\leq (c-\varepsilon)\cdot N$.
\end{lemma} 

A matching exponential lower bound on Hamming clustering is provided in \cref{subsec:clustering}. The Euclidean version of clustering is:

\begin{restatable}{lemma}{DNLEclusterintro}\label{lem:DNLE-cluster-intro}
Let $V$ be a finite subset of $\mathbb{R}^N$. There is a partition of $V$ into $2^{\textup{poly}(\varepsilon^{-1})}$ clusters such that if $u, v$ are in the same cluster, there are at most $\varepsilon \cdot |V|$ points $w$ of $V$ such that $w \in B(u, 1) \setminus B(v, 1+ \varepsilon)$.
\end{restatable}

We can further detail the structure between the clusters. This is done in the form of a Strong Regularity Lemma discussed in \cref{subsec:ultrastrong}. Roughly speaking, all pairs $(C,C')$ of clusters apart from an $\eta$-fraction are homogeneous, in the sense that apart from an $\eta$-fraction of pairs $(c,c')$ of elements of $(C,C')$, all are at distance at most $\tau N$ or all are at distance at least $(\tau -\varepsilon) N$.

\subsection{Chromatic thresholds of \texorpdfstring{$K_t$}{Kt}-free graphs}

Triangle-free graphs are the perfect playground for DNL. Indeed, when each vertex $x$ of a triangle-free graph $G$ is represented by its corresponding 0/1 column vector $V_x$ in the adjacency matrix $A_G$ of $G$, then every edge $xy$ of $G$ corresponds to vectors $V_x$ and $V_y$ which never agree on a 1 coordinate.
In other words, edges in triangle-free graphs correspond to vertices which are far apart with respect to the Hamming distance. 
This basic remark gives another short proof of \cref{thm:thomassen} using the Hamming version of DNL.

The existence of $\varepsilon$-clusterings, given by \cref{lem:DNLS-cluster-intro,lem:DNLH-cluster-intro,lem:DNLE-cluster-intro}, also offers surprisingly short arguments. The following result was proved in the Master's Thesis of O'Rourke \cite{O'R14}. He used the fact that $n$-vertex regular graphs have no independent set of size more than $n/2$.
We revisit his proof using DNL instead of the Regularity Lemma.

    \begin{restatable}{theorem}{thresholdregular}\label{thm:regulartrianglefree}
    Every $(1/4 + \varepsilon)n$-regular triangle-free graph has  chromatic number at most $2^{\textup{poly}(\varepsilon^{-1})}$.
    \end{restatable} 
    
    \begin{proof}
    Let $\mF = \{N(v) : v \in V\}$ and let $\mathcal{P}$ be an $\varepsilon/2$-clustering of $\mF$ as in \cref{lem:DNLS-cluster-intro}. Assume for contradiction that some cluster contains an edge $xy$. Since $N(x)\subseteq D(x)$, the set  $X'=N(x)\setminus D_{\varepsilon/2}(y)$ has size at most $\varepsilon n/2$ (analogous definition for $Y'$). 
    Note that $(N(x)\setminus X')\cup (N(y)\setminus Y')$ has size at least $2(1/4 + \varepsilon)n -2\varepsilon n/2>n/2$, so contains an edge $zt$.
    As $N(x), N(y), N(z), N(t)$  all pairwise intersect on at most $\varepsilon n/2$ elements, their union has size at least $4(1/4+\varepsilon)n-\binom{4}{2}\varepsilon n/2>n$, a contradiction. Thus clusters are independent sets, and $\chi(G)\leq |\mathcal{P}|$. 
    \end{proof}

Extremal constructions matching this $1/4$ threshold are provided in \cref{sec:TFG}.
We do not know if the bound on the chromatic number can be reduced. This threshold could even be a hard one: every $n/4$-regular triangle-free graph might have bounded chromatic number (this holds with $n/3$ in the non-regular setting).
\cref{thm:regulartrianglefree} extends to regular $K_t$-free graphs, where the (sharp) threshold is $\frac{3t-8}{3t-5}$.
The first difficulty (which we believe is the main reason why the regular threshold was not more investigated in the past) is that the proofs in this area rely on edge-maximality, which is not suited for regular graphs.
To this end, adding red edges is very convenient to materialize would-be edges. The second difficulty is that in the $K_t$-free case, the set system to consider is no longer the set of neighborhoods of vertices, but the set of common neighborhoods of cliques of size $t-2$, making computations more tedious. Our proof is rather technical and can be found in \cref{subsec:regthreshold}. We strongly believe that a much shorter argument exists.

Applications to homomorphism thresholds are discussed in \cref{subsec:homthreshold}. We obtain a bound singly exponential in $1/\varepsilon^2$ with DNL, which is slightly worse than the singly exponential in $1/\varepsilon$ dependency obtained in~\cite{LSSZ24}.

Finally, we unearth a surprising link between odd girth and chromatic thresholds. Let $og_c(n)$ be the minimum integer $k$ such that every $n$-vertex graph with odd girth at least $k$ has chromatic number at most $c$. We show that every triangle-free graph with minimum degree $n/3-n/8og_c(n)$ has chromatic number bounded in terms of $c$. By a result of Kierstead, Szemerédi and Trotter~\cite{DBLP:journals/combinatorica/KiersteadST84}, $og_c(n)=O\left(n^{\frac{1}{c-1}}\right)$. This implies that triangle-free graphs with minimum degree  $n/3-n^{1-\varepsilon}$ have bounded chromatic number, when $\varepsilon>0$ is fixed. This fully characterizes how far from $n/3$ we can go since the so-called \emph{Schrijver-Hajnal} graphs show that triangle-free graphs with minimum degree $n/3-n^{1-o(1)}$ do not have bounded chromatic number. This is discussed in \cref{subsec:digging}.
\subsection{Domination versus fractional chromatic number in tournaments}

Tournaments (orientations of complete graphs) are genuinely simpler than graphs as they only admit one type of simple substructure, transitive tournaments, whereas graphs admit both stable sets and cliques.
By analogy with graphs, the \emph{chromatic number} of a tournament is the minimum number of transitive tournaments needed to partition its vertex set. This parameter was investigated for instance in \cite{Abo24,Chu18,Gir24,Kli24,NSS24}. In the paper initiating this study \cite{BCCFLSST13}, the following question was raised: If the out-neighborhood of every vertex of a tournament $T$ has chromatic number at most $k$ (say that $T$ is \emph{locally $k$-bounded}), does $T$ itself have chromatic number at most $f(k)$?

This ``local to global'' property does not hold for graphs, since triangle-free graphs with large chromatic number have the property that the neighborhood of every vertex has chromatic number 1. The outcome is completely different for tournaments and the existence of such a function $f(k)$ was proved in \cite{HLTW19}. The proof is constructive, but not completely satisfactory as the function $f$ is tower-type. We propose here a completely different argument based on a very simple relation between tournament parameters: 

\begin{restatable}{theorem}{domfracchi}\label{thm:domfracchi}
Every tournament $T$ has domination number bounded in terms of its fractional chromatic number.
\end{restatable} 

Here, the fractional chromatic number of a tournament $T$ is the minimum fractional cover by transitive subtournaments. This is the exact analogue of the fractional chromatic number of a graph, based on independent sets. Measuring how large and well-distributed the transitive subtournaments are is particularly interesting, as the notorious Erd\H{o}s-Hajnal conjecture for graphs is equivalent to its tournament counterpart, see~\cite{Alo01}. A study of a link between the number of transitive subtournaments and the global structure can be found in~\cite{Fox21}. Let us sketch now why \cref{thm:domfracchi} holds. The first ingredient was proved in~\cite{Fish92}. 

\begin{theorem}\label{thm:fisherryan}
Every tournament $T$ has fractional domination number at most 2.
\end{theorem} 

Indeed, a direct application of Farkas Lemma provides a probability distribution $p$ on $V$ such that $p(N^+(v))\leq p(N^-(v))$ for all $v\in V$ (i.e. the weight of the in-neighborhood of every vertex is at least the weight of the out-neighborhood). 
Call this $p$ a \emph{winning strategy} and observe that if we consider the weight function $\omega \coloneqq 2p$, then $\omega(N^-[v])\geq 1$ for all vertices $v$, hence $\omega$ is a fractional domination function with total weight 2.

The second ingredient is an adaptation of the notion of trigraphs to tournaments. 
A \emph{tri-tournament} $T=(V,A,R)$ consists of a set $V$ of vertices, a set $A$ of arcs and a set $R$ of red arcs which are not in $A$. 
Moreover, for every pair of vertices $x,y$ in $V$, exactly one of the arcs $xy$ or $yx$ belongs to $A$. 
We can then see $T$ as a usual tournament with some additional red arcs forming circuits of length two. 
A set $X\subseteq V$ is a \emph{dominating set} of $T$ if for every vertex $y\in V\setminus X$, there exists $x\in X$ such that $xy\in A$ or $xy\in R$. 
In a tri-tournament $T$, a set of vertices $X\subseteq V$ is \emph{shattered} if there exists a set $S\subseteq V$ of size $2^{|X|}$ such that there is no red arc (in any direction) between $X$ and $S$, and all the closed in-neighborhoods in $X$ of the vertices of $S$ are pairwise distinct. 
The \emph{VC-dimension} of $T$ is the largest size of a shattered set. 
A direct application of \cref{thm:bounded-integrality-gap,thm:fisherryan} gives:

\begin{restatable}{theorem}{HWtournaments}\label{thm:HWtournaments}
Every tri-tournament with VC-dimension $d$ has domination number $O(d)$.
\end{restatable} 

The proof of \cref{thm:domfracchi} follows from a randomness transversal argument followed by a density increase argument. We consider a tournament $T$ with fractional chromatic number $1/c$, i.e. there exists a family ${\mathcal F}$ of  transitive subtournaments $T_1,\dots ,T_t$ such that every vertex belongs to $ct$ of them. 
Consider two large (but bounded) equal size disjoint subsets of vertices $A,B$ of $T$ which induce a random-like bipartite subtournament in the following sense: all subsets $A',B'$ of $A,B$ of size $\varepsilon |A|$ induce a (directed) cycle of length 4. 
This implies that no $T_i$ intersects both $A$ and $B$ on $\varepsilon |A|$ vertices, hence that there is an arc $ab$ with $a\in A$ and $b\in B$ such that $a$ and $b$ appear together in at most $\varepsilon t$ subtournaments $T_i$. 
For all such arcs $ab$, add the red arc $ba$, and call the resulting tri-tournament $T'$. By construction, every bipartite random-like subtournament of $T'$ has a red arc, hence $T'$ has bounded VC-dimension, therefore bounded domination by \cref{thm:HWtournaments}. Let $X$ be a dominating set of $T'$.
The vertices which are not dominated by $X$ in the original tournament $T$ are dominated in red in $T'$. The key-fact is that the red out-neighborhood $R^+(x)$ of a vertex $x$ in $T'$ consists of vertices which do not share more than $\varepsilon t$ subtournaments $T_i$'s with $x$. This means that they appear more on the $T_i$'s not containing $x$ (density increase argument). The fractional chromatic number of $T[R^+(x)]$ is then at most $(1-c+\varepsilon)/c$, which is close to $1/c-1$, and we conclude by induction.
The full proof is in \cref{sec:domtour}.

This result directly implies the local to global property of the chromatic number. Indeed, if $T$ is locally-$k$-bounded, we can fractionally distribute, via \cref{thm:fisherryan}, the acyclic partitions induced by the out-neighborhoods over the vertices in order to achieve fractional chromatic number $2k$. 
Hence $T$ has domination bounded by some function $g(k)$ by \cref{thm:domfracchi}, and finally chromatic number at most $kg(k)$. The proof gives a factorial-type bound instead of a tower function and the algorithmic construction of the dominating set is easy to describe: 
At each step we compute the winning strategy, and draw some (large enough) subset $X$ of vertices according to this distribution. 
We then delete the vertices dominated by $X$, partition the other ones into at most $|X|$ subtournaments with smaller fractional chromatic number, and iterate the process in each of these subtournaments.

\subsection{Domination in majority directed graphs}

This section revisits one of the earliest applications of VC-dimension to graph theory, due to Alon, Brightwell, Kierstead, Kostochka and Winkler in~\cite{ALON2006374}.

A committee of $2k+1$ referees is in charge of awarding some grants. Each referee proposes a total order on the set of applicants, and, as a result, an applicant $x$ is \emph{preferred} to $y$ if it is ranked before $y$ by at least $k+1$ voters. The key-result of~\cite{ALON2006374} is that the committee can select a bounded (in terms of $k$) set $X$ of applicants such that for every applicant $y\notin X$, some $x \in X$ is preferred to $y$. This indeed certifies (in some way) why $y$ was not chosen among the happy few selected for a grant. Unfortunately, the bound on the size of $X$ (and thus the number of grants a fair committee needs to award) grows with $k$, a fact popularized by the authors as: ``Large committees need more money''.

To generalize the voting process, we consider $n$ applicants, a committee of $m$ voters, and some $c\in [0,1]$. Say that the applicant $x$ is \emph{$c$-preferred} to $y$ if $x$ is preferred to $y$ by at least $cm$ voters. We obtain the \emph{$c$-majority} directed graph $D_c$ whose vertex set is the set of all applicants, with an arc $xy$ if $x$ is $c$-preferred to $y$. The original result of \cite{ALON2006374} asserts that when $m=2k+1$ and $c=1/2$, the digraph $D_{1/2}$ is a tournament with bounded domination. However, when the number $m$ of voters is not bounded (but still odd), the tournament $D_{1/2}$ no longer has bounded domination number. A key observation is that when the number of voters is $2k+1$, the candidate $x$ is preferred to $y$ if and only if more than $k$ voters prefer $x$ to $y$. This corresponds to a preference ratio greater than $k/(2k+1)<1/2$. This means that $D_{1/2}$ is in fact a $(1/2-\varepsilon)$-majority digraph, for any  $0\leq \varepsilon<\frac{1}{4k+2}$. One can wonder if, rather than from the fixed number of voters, bounded domination follows from the fact that $D$ is a $c$-majority directed graph with $c<1/2$. This is indeed the case:

\begin{restatable}{theorem}{majosloppy}\label{thm:majosloppy}
Every $\left(\frac{1}{2}-\varepsilon\right)$-majority digraph has domination number at most $O\left(\frac{1}{\varepsilon^2}\right)$.
\end{restatable} 

\cref{thm:majosloppy} was proved independently by Charikar, Ramakrishnan and Wang \cite{CRW25}.
This result can also be seen as an application of the Dense Neighborhood Lemma to some oriented metric-trigraph on the $n$-torus. However, we prove it without this geometric point of view. It again follows from a randomness transversal argument. Observe that when $D$ is a $(1/2-\varepsilon)$-majority digraph, it is a supergraph of the $1/2$-majority tournament $T$ on the same set of voters. 
If the arcs of $D$ which are not in $T$ are colored in red, every large random-like bipartite subtournament of $T$ contains a red edge. Thus $D$ can be seen as a tri-tournament with bounded VC-dimension, hence it has bounded domination.

\cref{thm:majosloppy} can be interpreted in a very concrete way: assume that the committee allows some slack in the decision process, for instance selecting $X$ such that for every applicant $y$ not in $X$, there some $x\in X$ is $0.49$-preferred to $y$. In such a case, the set $X$, and thus the number of grants, can be chosen of bounded size, hinting that ``Sloppy committees need less money''. Unfortunately, our method so far only provides an upper bound ($f(0.49)$) of a few million grants. Improving this value would be necessary before any application in policy making.

\subsection*{Organisation of the paper}

All relevant definitions are collected in \cref{sec:preliminaries}. 
We start by proving the analogue of the result of Haussler and Welzl in \cref{sec:trans}.
In \cref{sec:vcdimtri}, we show how to bound the VC-dimension of various classes of trigraphs and tri-hypergraphs. 
Combining the results of these two sections, we give all statements of DNL and its variants in \cref{sec:DNL}. 
We then move to the applications. 
In \cref{sec:TFG}, we investigate the chromatic and homomorphism thresholds of $K_t$-free graphs.
We then turn to tournaments in \cref{sec:domtour}, and explore the connection between the fractional chromatic number and domination. 
Finally, in \cref{sec:majo} we prove that $(1/2-\varepsilon)$-majority digraphs have bounded domination.
\section{Preliminaries}\label{sec:preliminaries}

\paragraph{Set systems.}
A set system $\mF$ is a family of sets on some ground set $V$. 
Given a vertex $x$, we denote by $\mF_{x}$ the set of all sets in $\mF$ which contain $x$, and write $\mF_{xy}=\mF_x\cap \mF_y$.
We denote by $D_\varepsilon(x)$ the set $\{y\in V : |\mF_{xy}|\leq \varepsilon  |\mF|\}$. When $\varepsilon =0$, we simply write $D(x) \coloneqq D_0(x)$. By extension, when $X \subseteq V$ we write  $D_\varepsilon(X)=\bigcup_{x\in X} D_\varepsilon (x)$. The notation $D$ stands for disjoint.
An \emph{$\varepsilon$-covering} of $\mF$ is a subset $X$ of vertices such that $D_\varepsilon(X)=V$.
The \emph{disjointness-ratio} of $\mF$ is the minimum of $|D(x)|/|V|$ over all $x\in V$. We usually denote it by $\delta$.
An $(\varepsilon, \eta)$-cluster is a subset $X$ of $V$ such that $|D(x)\setminus D_\varepsilon (y)|\leq \eta |V|$ for all $x, y\in X$, and an $(\varepsilon, \eta)$-clustering is a partition of $V$ into $(\varepsilon, \eta)$-clusters. 
The \emph{size} of a clustering is its number of clusters.
We will often refer to an $(\varepsilon, \varepsilon)$-clustering as an $\varepsilon$-clustering.

\paragraph{Graphs.} If $G = (V, E)$ is a graph and $V_1, V_2$ are subsets of $V$, we denote by $E(V_1, V_2)$ the set of edges with one endpoint in $V_1$ and the other in $V_2$, and we set $e(V_1, V_2) = |E(V_1, V_2)|$.  We also set $\overline{e}(V_1, V_2) = |V_1| \cdot |V_2| - e(V_1, V_2)$. Observe that $\overline{e}(V_1, V_2)$ is the number of non-edges between $V_1$ and $V_2$. We denote by $e(G)$ the number of edges of $G$.
A \emph{dominating set} of $G$ is a set $X \subseteq V$ of vertices such that every vertex not in $X$ has a neighbor in $X$. The size of a smallest dominating set of $G$ is denoted by $\gamma(G)$.

\paragraph{Trigraphs.}
A \emph{trigraph} is a triple $T=(V,E,R)$ where $V$ is the set of vertices, $E$ is the set of (\emph{plain}) edges, and $R$ is the set of \emph{red} edges, with $R\cap E=\emptyset$. We will always denote by $n$ the number of vertices. In graphs, we usually speak of ``non-edges'' to indicate pairs of nonadjacent vertices, we do the same for trigraphs, and denote by $W$ the set of non-edges. This stands for ``without edge'' or more visually the fact that two nonadjacent vertices drawn on a whiteboard are connected by a ``white edge''. Therefore $(E,R,W)$ is a partition of the set of pairs of vertices. Given a vertex $v$, we denote by $R(v)$ its set of red neighbors and by $W(v)$ its set of non-neighbors. We still denote the (plain) neighborhood of $v$ by $N(v)$, and by $N[v]$ for the closed neighborhood, including $v$. 

\paragraph{Trigraph parameters.}
The \emph{minimum degree} $\delta(T)$ of a trigraph $T$ is the minimum size of $N(v)$ over all $v\in V$. A \emph{dominating set} of $T$ is a set $X$ of vertices such that for every $y \in V$, there exists $x\in X$ such that $y\in N[x]\cup R(x)$. The size of a smallest dominating set of $T$ is denoted by $\gamma(T)$.
We say that a subset $X$ of $V$ is \emph{shattered} if there exists a set $S$ of $2^{|X|}$ vertices with no red neighbor in $X$ and such that for every $Y\subseteq X$, there exists $v_Y\in S$ such that $N[v_Y]\cap X=Y$. The \emph{VC-dimension} of a trigraph is the largest size of a shattered set.

\paragraph{Classes of trigraphs.}\label{par:classes-trigraphs}
A trigraph $T = (V, E, R)$ is a \emph{metric-trigraph} if there exists a metric space $(X, d)$ such that $V \subseteq X$, a threshold $\tau$ and a sensitivity $\varepsilon$ such that $E = \{uv : d(u, v) \leq \tau\}$ and $R = \{uv : \tau < d(u, v) \leq \tau + \varepsilon\}$. Equivalently, for every $u_1, u_2, v_1, v_2$ such that $u_1v_1 \in E$ and $u_2v_2 \notin E \cup R$, we have $d(u_2, v_2) - d(u_1, v_1) \geq \varepsilon$.
In particular, $T = (V, E, R)$ is a \emph{Hamming-trigraph} if the metric space is $(\{0, 1\}^N, d_H)$, where $d_H$ is the Hamming distance. 
Similarly, $T = (V, E, R)$ is a \emph{spherical-trigraph} if the metric space is $(\mathbb{S}^{N}, d_S)$, where $\mathbb{S}^{N} \subseteq \mathbb{R}^{N+1}$ is the unit sphere, and $d_S$ is the spherical distance, i.e. the angle $x0y$.
Finally, $T$ is a \emph{disjointness-trigraph} if there exists a set system $\mF$ on ground set $V$ and a sensitivity $\varepsilon > 0$ such that $E = \{xy : \mF_{xy} = \emptyset\}$ and $R = \{xy : 0 < |\mF_{xy}| \leq \varepsilon |\mF|\}$.

\paragraph{Tri-hypergraphs.}
A \emph{tri-hypergraph} $H$ is a pair $(V,\mathcal{E})$ where $V$ is a finite set and $\mathcal{E}$ is a collection of partitions of $V$ into three (possibly empty) parts $(B, R, W)$. Each such triple is called a \emph{tri-edge} of $H$. A standard hypergraph can be seen as a tri-hypergraph in which $R = \emptyset$ for every $(B, R, W)\in \mE$. When defining a tri-edge, we will often only define $B$ and $R$ explicitly, $W$ will then implicitly be defined as the complement of $B \cup R$.

\paragraph{Parameters on tri-hypergraphs.}
Given a tri-hypergraph $H=(V,\mE)$ and a set $X \subseteq V$ we define the \emph{trace} of $H$ on $X$ as 
$$tr_H(X)=\{ Y\subseteq X \ |\ \exists (B, R, W)\in \mE \text{  such that }  Y=X\cap B=X\cap (B \cup R) \}.$$
A set $X \subseteq V$ is \emph{shattered} by $H$ if $tr_H(X)=2^{X}$. The \emph{VC-dimension} of $H$ is the largest size of a shattered set. 
A partition $(X_0, X_1)$ of a set $X \subseteq V$ is a \emph{clean separation} if there exists an edge $(B, R, W) \in \mE$ such that $W \cap X=X_0$ and $B \cap X=X_1$. Then, $X$ is shattered if every partition $(X_0, X_1)$ of $X$ is a clean separation.
A set of vertices $X \subseteq V$ is a \emph{transversal} (or \emph{hitting set}) of $H$ if  $X \cap (B \cup R) \neq \emptyset$ for every $(B, R, W)\in \mE$. A set of vertices $X \subseteq V$ is a \emph{$\delta$-net} of $H$ if $X \cap (B \cup R) \neq \emptyset$ for every $(B, R, W) \in \mE$ such that $|B| \geq \delta|V|$.
Note that when $R = \emptyset$ for every tri-edge, we recover the definition of VC-dimension and transversal for hypergraphs. 

Given a tri-hypergraph $H = (V, \mE)$, a weight function $\omega : V \to \mathbb{R}^+$ is a \emph{fractional transversal} if for every $(B, R, W) \in \mE$, it holds that $\omega(B) \coloneqq \sum_{b \in B} \omega(b) \geq 1$. Its \emph{weight} is $\omega(V) \coloneqq \sum_{v \in V} \omega(v)$. The minimum weight of a fractional transversal of $H$ is denoted by $\tau^*(H)$ (or simply $\tau^*$ if $H$ is clear from the context), and the minimum size of a transversal of $H$ is denoted by $\tau(H)$.
Note that if $R = \emptyset$ for every $(B, R, W) \in \mE$, this coincides with the standard definitions of $\tau$ and $\tau^*$ in hypergraphs.
However, there are some key differences when compared to hypergraphs, one of them being that the inequality $\tau^* \leq \tau$ does not always hold, as $\tau^*$ is not the fractional relaxation of $\tau$.

\paragraph{Operations on tri-hypergraphs.}
If $H = (V, \mE)$ is a tri-hypergraph, let $\overline{\mE} = \{(W, R, B) : (B, R, W) \in \mE\}$ and $\overline{H} = (V, \overline{\mE})$.
We say that $\overline{H}$ is the \emph{complement} of $H$. Observe that $H$ and $\overline{H}$ shatter the same sets, and thus have the same VC-dimension.
If $e_1 = (B_1, R_1, W_1)$ and $e_2 = (B_2, R_2, W_2)$ are two tri-edges on the same ground set, the tri-edge $e_1 \cap e_2$ is the tri-edge $(B, R, W)$ where $B = B_1 \cap B_2$ and $R = ((B_1 \cup R_1) \cap (B_2 \cup R_2)) \setminus B$. Similarly, the tri-edge $e_1 \setminus e_2$ is the the tri-edge $(B, R, W)$ where $B = B_1 \cap W_2$ and $R = ((B_1 \cup R_1) \cap (W_2 \cup R_2)) \setminus B$.
Then, if $H_1 = (V, \mE_1)$ and $H_2 = (V, \mE_2)$ are two tri-hypergraphs on the same ground set, the tri-hypergraph $H_1 \cap H_2$ is the tri-hypergraph $(V, \{e_1 \cap e_2 : e_1 \in \mE_1, e_2 \in \mE_2\})$, and the tri-hypergraph $H_1 \setminus H_2$ is the tri-hypergraph $(V, \{e_1 \setminus e_2 : e_1 \in \mE_1, e_2 \in \mE_2\})$.
Given two tri-hypergraphs $H = (V, \{(B_e, R_e, W_e): e \in \mE\})$ and $H' = (V, \{(B'_e, R'_e, W'_e): e \in \mE\})$, we say that $H$ \emph{refines} $H'$ if for every $e \in \mE$ we have $B'_e \subseteq B_e$ and $W'_e \subseteq W_e$. Intuitively, if $H$ refines $H'$ then each tri-edge of $H$ can be obtained from the corresponding tri-edge of $H'$ by moving some vertices from the red part to either the black or the white part. Observe that if $H$ refines $H'$ then the VC-dimension of $H'$ is upper bounded by the one of $H$.
If $T = (V, E, R)$ is a trigraph, its \emph{corresponding tri-hypergraph} is $H_T = (V, \mE)$ with $\mE = \{(N[v], R(v), W(v)) : v \in V\}$. Observe that $T$ and $H_T$ have the same VC-dimension, and that the transversals of $H_T$ are exactly the dominating sets of $T$.

\paragraph{Digraphs.}
In a directed graph (or digraph) $D=(V,A)$, $A$ is the set of \emph{arcs} (oriented edges). We denote by $N^-(v)$ the set of in-neighbors of a vertex $v$ and by $N^+(v)$ its set of out-neighbors. We adopt the notation $N^-[v]$ and $N^+[v]$ for the closed neighborhoods, that is, including $v$. The \emph{domination number} $\gamma ^+(D)$ of $D$ is the minimum size of a (\emph{dominating}) set $X$ of vertices satisfying that $N^-[v]\cap X\neq \emptyset$ for every vertex $v$.
The \emph{fractional dominating number} $\gamma ^+_f(D)$ is the minimum of $\omega(V)$ taken over all weight functions $\omega : V \to \mathbb{R}^+$ such that $\omega (N^-[v])\geq 1$ for all $v\in V$. 
The \emph{absorption number} $\gamma ^-(D)$ of $D$ is the minimum size of a (\emph{absorbing}) set $X$ of vertices satisfying that $N^+[v]\cap X\neq \emptyset$ for every vertex $v$.
The \emph{fractional absorption number} $\gamma ^-_f(D)$ is the minimum of $\omega(V)$ taken over all weight functions $\omega : V \to \mathbb{R}^+$ such that $\omega (N^+[v])\geq 1$ for all $v\in V$.
The \emph{acyclic chromatic number} $\chi ^a(D)$ of a directed graph $D=(V,A)$ is the minimum number of acyclic induced subgraphs of $D$ whose union covers $V$ (or equivalently partitions $V$). 
The \emph{fractional acyclic chromatic number} $\chi ^a_f(D)$ is the minimum total sum of a weight function $\omega$ on acyclic induced subgraphs, satisfying that for every vertex $v$, the sum of the weights of the acyclic subgraphs containing $v$ is at least 1.
A \emph{tournament} is a digraph where any two vertices are connected by exactly one arc.

\paragraph{Tri-tournaments.}
A \emph{tri-tournament} $T=(V,A,R)$ consists of a set of vertices $V$, a set $A$ of arcs and a set $R$ of red arcs which are not in $A$. Moreover for every pair of vertices $x,y$ in $V$, exactly one of the arcs $xy$ or $yx$ belongs to $A$. We can then see $T$ as a usual tournament with some additional red arcs forming circuits of length two. A set $X\subseteq V$ is a \emph{dominating set} of $T$ if for every vertex $y\in V\setminus X$, there exists $x\in X$ such that $xy\in A$ or $xy\in R$. In a tri-tournament $T$, a set of vertices $X\subseteq V$ is \emph{shattered} if there exists a set $S\subseteq V$ of size $2^{|X|}$ such that there is no red arc (in any direction) between $X$ and $S$, and all closed in-neighborhoods in $X$ of vertices of $S$ are pairwise distinct. The \emph{VC-dimension} of $T$ is the largest size of a shattered set.

\paragraph{Miscellaneous.}
We conclude this section with a technical Lemma, which we often use.
\begin{lemma}\label{lem:tech}
If $a, b, x>0$ satisfy $ax\geq 1$ and $x\leq b \ln(ax)$, then $x\leq 2b \ln(ab)$
\end{lemma}

\begin{proof}
Increase the value of $x$ until reaching $x= b \ln(ax)$. We then want to show that $b \ln(ax)\leq 2b \ln(ab)$, which is simply $\ln(ax)\leq 2 \ln(ab)$, hence $ax\leq a^2b^2$. Substituting $b$ by $x/\ln(ax)$, our target inequality is equivalent to $\ln(ax)^2\leq ax$, which holds since $ax \geq 1$.
\end{proof}
\section{Transversals in bounded VC-dimension}\label{sec:trans}

The main result of this section is a generalization to tri-hypergraphs of the celebrated $\delta$-net theorem of Haussler and Welzl on hypergraphs \cite{HW87}, stating that in an $n$-vertex hypergraph of VC-dimension $d$, there is a set of size $O\left(\frac{d}{\delta}\log\left(\frac{1}{\delta}\right)\right)$ which intersects all hyperedges of size at least $\delta n$.

\begin{theorem}\label{thm:HW-trihyper}
Let $H=(V,\mE)$ be a tri-hypergraph of VC-dimension $d$ and let $\delta > 0$. 
There is a $\delta$-net $X$ such that
$$
|X|< 2 \cdot \frac{d}{\ln(1+\delta)} \cdot \ln\left(\frac{e}{\ln(1+\delta)}\right).
$$
\end{theorem}
Using that $\ln(1+x) \geq x/2$ for any $0\leq x \leq 1$, we recover the more classical upper bound $O\left(\frac{d}{\delta}\ln\left(\frac{1}{\delta}\right)\right)$.

Observe that \cref{thm:triHW} follows immediately from \cref{thm:HW-trihyper} by considering the corresponding tri-hypergraph $H_T$, using that $T$ and $H_T$ have the same VC-dimension, and that the $\delta$-nets of $H_T$ are the dominating sets of $T$.

To simplify some computations in the proof of \cref{thm:HW-trihyper}, we will consider multisets of vertices, which we will view as words on the alphabet $V$. For a word $Z$, we denote its length by $|Z|$, and for any $I\subseteq \{1,\ldots,|Z|\}$ we denote by $Z_I$ (resp. $\overline{Z_I}$) the subword of $Z$ obtained by keeping (resp. removing) all letters with indices in $I$. 

Given a tri-hypergraph $H=(V,\mE)$, we say that a word $X$ is \emph{separated} from a word $Y$ if there exists an edge $(B, R, W)$ of $H$ such that all letters of $X$ belong to $B$ and all letters of $Y$ belong to $W$. For a word $Z$ we define the \emph{trace} of $H$ on $Z$, denoted $tr_H(Z)$ by 
$$
tr_H(Z)=\left\{\text{$I\subseteq \{1,\ldots,|Z|\}$ such that  $Z_I$ is separated from  $\overline{Z_I}$}\right\}.
$$

The following result, proved independently by Sauer \cite{Sauer72} and Shelah \cite{Shelah72} is the central ingredient of the theory of VC-dimension.

\begin{lemma}\label{lem:SS}
If $\mH = (V, E)$ is a hypergraph of VC-dimension $d$ then $|E| \leq \sum_{i = 0}^d\binom{|V|}{i}$.
\end{lemma}

We extend the result of Sauer and Shelah to words.

\begin{lemma}\label{lem:gen-sauer-shelah}
If $H=(V,\mE)$ be a tri-hypergraph of VC-dimension $d$ then every word $Z$ satisfies $|tr_H(Z)|\leq \sum_{i=0}^{d}\binom{|Z|}{i}$.
\end{lemma}
\begin{proof}
Assume by contradiction that $|Z|=t$ and  $|tr_H(Z)|> \sum_{i=0}^{d}\binom{t}{i}$.  Consider the hypergraph on ground set $[t]$ with edge set $tr_H(Z)$. By \cref{lem:SS}, there exists a shattered set $I\subseteq [t]$ of size at least $d+1$. The corresponding subword $Z_I$ must contain pairwise distinct letters, and thus corresponds to a subset of size $d+1$ of $V$ that is shattered by $H$, contradicting the definition of $d$.
\end{proof}

\begin{lemma}
Let $H=(V,\mE)$ be an $n$-vertex tri-hypergraph and $\delta > 0$. Let $\tau$ denote the minimum size of a $\delta$-net of $H$. Then for every $t<\tau$,
$$(1+\delta)^{t} \leq \frac{1}{n^t}\sum_{|Z|=t} |tr_H(Z)|.$$
\end{lemma}

\begin{proof}
Any word $X$ of length less than $\tau$ contains less than $\tau$ distinct letters, hence there exists an edge $(B, R, W) \in \mE$ such that $X \cap (B \cup R) = \emptyset$ and $|B| \geq \delta n$. For any integer $y$, there exist at least $(\delta n)^y$ words of length $y$ using only letters from $B$. Hence, for any word $X$ with $|X|<\tau$:
$$
\text{$|\{Y\in V^{y}$, such that $X$ is separated from $Y$}\} | \geq (\delta n)^y.
$$
Now for $t<\tau$ we have:
\belowdisplayskip=-12pt\begin{align*}
\frac{1}{n^t}\sum_{|Z|=t} |tr_H(Z)|
&=\frac{1}{n^t}\sum_{Z\in V^t} \left|\left\{\text{$I\subseteq [t]$ such that   $Z_I$ is separated from $\overline{Z_I}$}\right\}\right|\\
&= \frac{1}{n^t} \sum_{I\subseteq [t]} \left|\left\{\text{$Z\in V^t$ such that   $Z_I$ is separated from $\overline{Z_I}$}\right\}\right|\\
&\geq \frac{1}{n^t} \sum_{i=0}^t  \binom{t}{i} n^i (\delta n)^{t-i}= (1+\delta)^t.
\end{align*}
\end{proof}

\begin{proof}[Proof of \cref{thm:HW-trihyper}.]
We combine the two previous Lemmas. For every $t<\tau$ we have $(1+\delta)^{t} \leq \sum_{i=0}^{d}\binom{t}{i}\leq \left(\frac{et}{d}\right)^d$.
By denoting $a=\frac{e}{d}$ and $b=\frac{d}{\ln(1+\delta)}$, we get $t\leq b \ln(at)$, which yields the desired bound by \cref{lem:tech}.

\end{proof}

A similar proof actually yields that a random sample of $O\left(\frac{d}{\log(1+\delta)} \cdot \log\left(\frac{e}{\log(1 + \delta)}\right)\right)$ elements is a $\delta$-net with constant probability. This gives a unified explanation for why various proofs based on random sampling work so well. We state the result for completeness.  

\randomsample*

Using the generalization of $\delta$-nets to tri-hypergraphs, we can easily prove that in tri-hypergraphs of VC-dimension $d$, the inequality $\tau \leq f(d, \tau^*)$ also holds.

\boundedintegralitygap*

\begin{proof}
Since $\tau^*$ is the solution of a linear program with integer constraints, there exists a fractional transversal $\omega : V \to \mathbb{Q}$ of weight $\tau^*$.
Let $N$ be an integer such that $N \cdot \omega(v) \in \mathbb{N}$ for every $v \in V$, and set $\omega'(v) = N \cdot \omega(v)$ for every $v \in V$.
Let $H' = (V', \mE')$ be the tri-hypergraph obtained from $H$ by creating $\omega'(v)$ copies of every vertex $v \in V$ (if $\omega'(v) = 0$ we simply delete $v$).
Observe that the VC-dimension of $H'$ is at most $d$ since there cannot be two copies of the same vertex in a shattered set.
Note also that $|V'| = \sum_{v \in V}\omega'(v) = N \cdot \tau^*$, and for every $(B', R', W') \in \mE'$ corresponding to $(B, R, W) \in \mE$, we have $|B'| = \sum_{b \in B}\omega'(b) \geq N$ since $\omega$ is a fractional transversal.
Thus, $H'$ is a tri-hypergraph of VC-dimension $d$ where every $(B', R', W') \in \mE'$ satisfies $|B'| \geq \frac{1}{\tau^*} \cdot |V'|$.
By \cref{thm:HW-trihyper}, $H'$ has a transversal of size $O(d\cdot \tau^* \log(\tau^*))$. 
This transversal is also a transversal of $H$.
\end{proof}
\section{Tri-hypergraphs with bounded VC-dimension}\label{sec:vcdimtri}

The aim of this section is to give bounds on the VC-dimension of several natural tri-hypergraphs, and to provide various methods to prove such bounds.
In order to give some geometric intuition, we start with the most visual proofs, by considering metric-trigraphs.

We first show that metric-trigraphs have bounded VC-dimension. The proof is based on a random cut argument to argue that some random-like graph has a separation which is too unbalanced. 
We then prove that having bounded VC-dimension is preserved by taking intersections or differences. The proof of this result uses double counting coupled with the Sauer-Shelah Lemma.
With that result in hand, we easily deduce a more subtle result about the intersection of metric-trigraphs.

We then consider disjointness-trigraphs, for which we provide similar results.
To show that they have bounded VC-dimension, we argue that some random-like graph has a complete bipartite subgraph which is too large.
We finally show that a refined variant of the difference tri-hypergraph of disjointness-trigraphs also has bounded VC-dimension. The proof of this results relies on the existence of small transversals of dense set systems to obtain adjacency labellings which are too short for random-like graphs.

\subsection{Metric-trigraphs}\label{sec:VCdim-trigraphs}

We start by showing that Hamming-trigraphs have bounded VC-dimension. Importantly, the bound on their VC-dimension depends on the sensitivity parameter $\varepsilon$ but not on the dimension of the ambient space.

\begin{theorem}\label{thm:VCdim-Hamming}
Every Hamming-trigraph with sensitivity $\varepsilon$ has VC-dimension $O(\varepsilon ^{-2})$.
\end{theorem}

\begin{proof}
Let $T = (V, E, R)$ be a Hamming-trigraph with sensitivity $\varepsilon$, threshold $\tau$ and $V \subseteq \{0, 1\}^N$, for some integer $N$. If $xy \in E$ then $d_H(x, y) \leq \tau \cdot N$, and if $xy \notin E \cup R$ then $d_H(x, y) > (\tau + \varepsilon) \cdot N$.
Let $H_T = (V, \mE)$ be the tri-hypergraph corresponding to $T$, with $\mE = \{(B(v), R(v), W(v)) : v \in V\}$.
Let $X \subseteq V$ be a shattered set: for every $Y \subseteq X$, there exists $v_Y \in V$ such that $B(v_Y) \cap X = (B(v_Y) \cup R(v_Y)) \cap X = Y$.

Let $Z$ be a random multi-subset of $\{v_Y : Y \subseteq X\}$ of size $|X|$ obtained by picking with repetition uniformly random elements of $\{v_Y : Y \subseteq X\}$.
Consider the bipartite graph $\Gamma$ on vertex set $X \cup Z$ where $xz$ is an edge if $x\in X$, $z\in Z$, and $d_H(x, z) \leq \tau \cdot N$ (or equivalently if $d_H(x, z) \leq (\tau + \varepsilon) \cdot N$).
Then, $\Gamma$ is a random bipartite graph where each edge between $X$ and $Z$ is present independently of all others with probability $1/2$.
Let $E_0$ be the event that there exist partitions $(X_1, X_2)$ of $X$ and $(Z_1, Z_2)$ of $Z$ such that $\overline{e}(X_1, Z_2) + \overline{e}(X_2, Z_1) - e(X_1, Z_2) - e(X_2, Z_1) \geq \frac{\varepsilon}{2} \cdot |X|\cdot|Z|$. Informally, $E_0$ is the event that there exists a cut of $\Gamma$ such that among all pairs $(x, z) \in X \times Z$ which are separated by the cut, there are much more non-edges than edges.

First, by symmetry between the edges and the non-edges in $\Gamma$, we have $\mathbb{P}[e(\Gamma) \leq |X|\cdot|Z|/2] \geq 1/2$. Suppose from now on that this inequality holds.
Pick a random coordinate $i \in [N]$, and let $X_1 = \{x \in X : x_i = 0\}$ and $X_2 = \{x \in X : x_i = 1\}$, and define $Z_1$ and $Z_2$ similarly.
If $xz$ is an edge of $\Gamma$ then $x$ and $z$ differ on at most $\tau \cdot N$ coordinates so the probability that $x$ and $z$ are separated (meaning $x \in X_a$ and $z \in Z_{3-a}$) is at most $\tau$.
On the other hand, if $xz$ is not an edge of $\Gamma$ then $x$ and $z$ differ on at least $(\tau + \varepsilon) \cdot N$ coordinates so the probability that $x$ and $z$ are separated is at least $\tau + \varepsilon$.
Thus, \begin{align*}
    \mathbb{E}[e(X_1, Z_2) + e(X_2, Z_1)] &\leq \frac{|X|\cdot|Z|}{2} \cdot \tau \\
    \mathbb{E}[\overline{e}(X_1, Z_2) + \overline{e}(X_2, Z_1)] &\geq \frac{|X|\cdot|Z|}{2} \cdot (\tau + \varepsilon).
\end{align*}
Putting together the previous two inequalities, we get $$
    \mathbb{E}[\overline{e}(X_1, Z_2) + \overline{e}(X_2, Z_1) - e(X_1, Z_2) - e(X_2, Z_1)]
    \geq \frac{\varepsilon}{2} \cdot |X|\cdot|Z|.
$$
Therefore, there exist partitions $(X_1, X_2)$ of $X$ and $(Z_1, Z_2)$ of $Z$ such that $\overline{e}(X_1, Z_2) + \overline{e}(X_2, Z_1) - e(X_1, Z_2) - e(X_2, Z_1) \geq \frac{\varepsilon}{2} \cdot |X| \cdot |Z|$.
Overall, we get $\mathbb{P}[E_0] \geq 1/2$.

We now give an upper bound on $\mathbb{P}[E_0]$.
Fix a partition $(X_1, X_2)$ of $X$ and a partition $(Z_1, Z_2)$ of $Z$, and consider a random graph $\Gamma$ on vertex set $X \cup Z$.
Let $E_1$ be the event that $\overline{e}(X_1, Z_2) + \overline{e}(X_2, Z_1) - e(X_1, Z_2) - e(X_2, Z_1) \geq \frac{\varepsilon}{2} \cdot |X|\cdot|Z|$.
Consider the random variable $K = \overline{e}(X_1, Z_2) + \overline{e}(X_2, Z_1)$.
Then, $K$ is the sum of $|X_1| \cdot |Z_2| + |X_2| \cdot |Z_1|$ independent Bernoulli variables with success probability $1/2$.
Observe that $E_1$ is the event that $K \geq \frac{|X_1| \cdot |Z_2| + |X_2| \cdot |Z_1|}{2} + \frac{\varepsilon}{4} \cdot |X|\cdot|Z| = \mathbb{E}[K]  + \frac{\varepsilon}{4} \cdot |X|\cdot|Z|$.
By Hoeffding's inequality, we have \begin{align*}
    \mathbb{P}[E_1] &\leq \exp\left(-\frac{2 \cdot \varepsilon^2 \cdot |X|^2 \cdot |Z|^2}{16 \cdot (|X_1| \cdot |Z_2| + |X_2| \cdot |Z_1|)}\right) \\
        &\leq \exp\left(-\frac{\varepsilon^2 \cdot |X| \cdot |Z|}{8}\right).
\end{align*}
There are $2^{|X| + |Z|}$ possible choices for the partitions $(X_1, X_2)$ and $(Z_1, Z_2)$ so by the union bound we have $\mathbb{P}[E_0] \leq 2^{|X| + |Z|} \cdot \exp\left(-\frac{\varepsilon^2 \cdot |X| \cdot |Z|}{8}\right)$.

Putting together the upper bound and the lower bound, and using that $|X| = |Z|$, we get $1/2 \leq 2^{2|X|} \cdot \exp\left(-\frac{\varepsilon^2 \cdot |X|^2}{8}\right)$.
This inequality implies $|X| \leq \frac{32}{\varepsilon^2}$, which concludes the proof.
\end{proof}

A similar result holds for spherical-trigraphs. One way to see it is to repeat the above proof, except that to show that $\mathbb{P}[E_0] \geq 1/2$, instead of defining $X_1, X_2, Z_1, Z_2$ according to a random coordinate, we pick a random hyperplane and define $X_1$ and $Z_1$ as the vertices lying on one side of this hyperplane, and $X_2$ and $Z_2$ as the vertices lying on the other side of this hyperplane. We give an alternative proof by showing that every spherical-trigraph is a Hamming-trigraph. The embedding is obtained by taking a large number of random hyperplanes and classifying the vertices based on which side of the hyperplanes they lie. The details are given in \cref{sec:proofs-VCdim-trigraphs}.

\begin{restatable}{theorem}{VCdimspherical}\label{thm:VCdim-spherical}
Every spherical-trigraph with sensitivity $\varepsilon$ is a Hamming-trigraph with sensitivity $\frac{\varepsilon}{2\pi}$.
In particular, every spherical-trigraph with sensitivity $\varepsilon$ has VC-dimension $O(\varepsilon^{-2})$.
\end{restatable}

The converse also holds: every Hamming-trigraph is a spherical-trigraph. This time the embedding is even simpler: map every $0$ coordinate to $-1/\sqrt{N}$ and every $1$ coordinate to $1/\sqrt{N}$. The proof is deferred to \cref{sec:proofs-VCdim-trigraphs}.

\begin{restatable}{lemma}{Hamissph}
Every Hamming-trigraph with sensitivity $\varepsilon$ is a spherical-trigraph with sensitivity $2\varepsilon$.
\end{restatable}

We now show that the bound in \cref{thm:VCdim-spherical} is essentially tight, up to constant factors. This implies that the bound of \cref{thm:VCdim-Hamming} is also essentially tight. 

\begin{proposition}
For every $\varepsilon > 0$, there exists a spherical-trigraph with sensitivity $\varepsilon$ with VC-dimension $\Omega(\varepsilon^{-2})$.
\end{proposition}

\begin{proof}
Recall that $\arccos{x} = \pi/2 - x + o(x)$ when $x \to 0$. Let $N$ be large enough so that $\arccos{(-1/\sqrt{N})} > \pi/2 + 1/(2\sqrt{N})$. Let $\mathcal{B} = \{b_1, \ldots, b_N\}$ be an orthonormal basis of $\mathbb{R}^N$. Then, $b_1, \ldots, b_N \in \mathbb{S}^{N-1}$.
For every subset $S \subseteq [N]$, let $v_S = \frac{1}{\sqrt{N}} \left( \sum_{s \in S}b_s - \sum_{t \notin S}b_t\right) \in \mathbb{S}^{N-1}$.
Consider the $1/(2\sqrt{N})$-trigraph $T= (V,E,R)$ with $V = \mathcal{B} \cup \{v_S : S \subseteq [N]\}$, $E = \{uv: d_S(u, v) \leq \pi/2\}$ and $R = \{uv: \pi/2 < d_S(u, v) \leq \pi/2 + 1/(2\sqrt{N})\}$.

We show that $\mathcal{B}$ is shattered in $T$, which will conclude the proof.
Let $\mB' \subseteq \mathcal{B}$ and let $S = \{i \in [N] : b_i \in \mB'\}$.
Then, for every $b \in \mB'$, we have $\langle b, v_S\rangle = 1/\sqrt{N}$ and for every $b \notin \mB'$, we have $\langle b, v_S\rangle = -1/\sqrt{N}$.
Thus, for every $b \in \mB'$, we have $d_S(b, v_S) = \arccos{(1/\sqrt{N})} \leq \pi/2$ so $bv_S \in E$ and for every $b \notin \mB'$, we have $d_S(b, v_S) = \arccos{(-1/\sqrt{N})} > \pi/2 + 1/(2\sqrt{N})$ so $bv_S \notin E \cup R$. Therefore $\mathcal{B} \cap N[v_S] = \mathcal{B} \cap (N[v_S] \cup R(v_S)) = \mB'$.
\end{proof}

\begin{remark}
Using similar ideas, we can show that the bound of \cref{thm:HW-trihyper} is almost tight.
Consider the trigraph $T = (V, E, R)$ with $V = \mathbb{S}^{n-1}$, $E = \{uv: d_S(u, v) \leq \pi/2-1/2\sqrt{n}\}$ and $R = \{uv: \pi/2-1/2\sqrt{n} < d_S(u, v) \leq \pi/2-1/4\sqrt{n}\}$.
Then, the (plain) neighborhood of each vertex represents at least a quarter of the ground set, and $T$ has VC-dimension $\Theta(n)$.
We show that $\gamma(T) = \Omega(n)$.
Note that for any point $v \in \mathbb{S}^{n-1}$, its neighborhood $N[v] \cup R(v)$ only consists of vertices whose scalar product with $v$ is $> 0$.
Therefore, for any set $S \subseteq \mathbb{S}^{n-1}$ of size at most $n-1$, there exists a point $p \in \mathbb{S}^{n-1}$ which is orthogonal to all points in $S$, hence is not dominated by $S$.
This proves that $\gamma(T) \geq n$.
\end{remark}

\subsection{Operations on tri-hypergraphs}\label{subsec:inter-trihyper}

A standard result in VC-dimension theory states that if a hypergraph $\mH$ has bounded VC-dimension then the hypergraph whose edges are all pairwise intersections (resp. differences) of edges of $\mH$ also has bounded VC-dimension.
We show that these results also hold for tri-hypergraphs. Our proofs are somewhat different from the classical ones.

\begin{lemma}\label{lem:VCdim-intersection-tri-hyper}
Let $H_1=(V,\mathcal{E}_1)$ and $H_2 = (V, \mE_2)$ be two tri-hypergraphs of VC-dimension $d$.
Let $\mathcal{I} = H_1 \cap H_2$ be their intersection tri-hypergraph.
Then, $\mathcal{I}$ has VC-dimension $O(d)$.
\end{lemma}

\newcommand{\fillcarre}[3]{ \fill[#3!80] (#1*\cellsize,#2*\cellsize) rectangle (#1*\cellsize+\cellsize,#2*\cellsize+\cellsize); } 
\newcommand{\fillrect}[4]{
  \fill[#4!80] (#1,#2) rectangle (#3,#2+1);
}

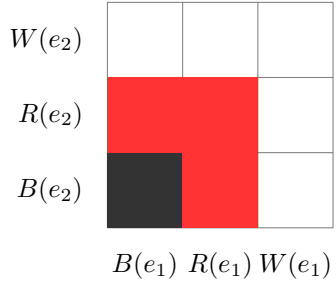
\begin{figure}[ht]
    \begin{center}
    \hspace{-1cm}
        \begin{tikzpicture} 
            \draw[step=1,gray] (0,0) grid (3,3);
        
            \fill[black!80] (0,0) rectangle (1,1);
            \fill[red!80] (1,1) rectangle (2,2);
            \fill[red!80] (0,1) rectangle (1,2);
            \fill[red!80] (1,0) rectangle (2,1);
        
            \node[anchor=center] at (0.5,-0.5) {$B(e_1)$};
            \node[anchor=center] at (1.5,-0.5) {$R(e_1)$};
            \node[anchor=center] at (2.5,-0.5) {$W(e_1)$};
        
            \node[anchor=east] at (-0.2,0.5) {$B(e_2)$};
            \node[anchor=east] at (-0.2,1.5) {$R(e_2)$};
            \node[anchor=east] at (-0.2,2.5) {$W(e_2)$};
        \end{tikzpicture}
    \caption{Edges of $\mathcal{I}$}
    \label{fig:intersection-tri-hyper}
    \end{center}
\end{figure}

\vspace{-10pt}

\begin{proof}
Let $X \subseteq V$ be a set of size $s$ which is shattered by $\mathcal{I}$. For every subset $Y \subseteq X$, there exists a pair $(e_1, e_2) \in \mE_1 \times \mE_2$ such that for every $y\in Y$, $y \in B(e_1)\cap B(e_2)$ and for every  $y\not\in Y$, $y \in W(e_1)\cup W(e_2)$ (see \cref{fig:intersection-tri-hyper}). 
So, one of $W(e_1), W(e_2)$ contains at least half of the elements of $|X\setminus Y|$. 
For every $Y\subseteq X$, fix a set $\phi(Y)$ of size at least $|X \setminus Y|/2$ such that there exists an edge $e \in \mE_1 \cup \mE_2$ with $Y \subseteq B(e)$ and $\phi(Y)\subseteq W(e)$. 

We construct a bipartite graph $\Gamma$ whose two parts, $A$ and $B$, are copies of $2^X$ and add an edge from $Y\in A$ to every set $Y \cup Z$ for $Z \subseteq \phi(Y)$ in $B$. 
Observe first that the degree of every vertex of $A$ corresponding to a subset of size $k$ is at least $2^{(s-k)/2}$. 
Consider now a vertex $z$ of $B$ corresponding to a subset $Z \subseteq X$. 
By definition of $\Gamma$, for every neighbor of $z$ in $A$ corresponding to a set $Y \subseteq X$, we have $Z \setminus Y \subseteq \phi(Y)$ so there exists an edge $e \in \mE_1 \cup \mE_2$ such that $Y \subseteq B(e)$ and $Z \setminus Y \subseteq \phi(Y) \subseteq W(e)$. 
Therefore, every neighbor of $z$ in $\Gamma$ corresponds to a clean separation on $Z$ in $H_1$ or in $H_2$, and these traces are pairwise distinct. 
Since $H_1$ and $H_2$ have VC-dimension $d$, the Sauer-Shelah Lemma implies that $z$ has at most $2\sum_{i = 0}^d\binom{|Z|}{i} \leq 2\left(\frac{e|Z|}{d}\right)^d$ neighbors in $A$. A double-counting of the edges of $\Gamma$ then yields $$\sum_{k=0}^{s} \binom{s}{k} 2^{(s-k)/2} \leq e(\Gamma) \leq \sum_{k=0}^s \binom{s}{k} \left(\frac{2ek}{d}\right)^d.$$
This implies that $(1+\sqrt{2})^s \leq  2^s \cdot (2es/d)^d$ so  $\left(\frac{1 + \sqrt{2}}{2}\right)^s \leq (2es/d)^d$.
Thus, $s \leq \frac{d}{\alpha} \cdot \ln(2es/d)$, for $\alpha = \ln\left(\frac{1 + \sqrt{2}}{2}\right)$, which implies by \cref{lem:tech} that $s \leq \frac{2d}{\alpha}\ln\left(\frac{2e}{\alpha}\right) = O(d)$.
\end{proof}

The analogue result for the difference of tri-hypergraphs follows immediately since $H_1 \setminus H_2 = H_1 \cap \overline{H_2}$ and since the VC-dimension is invariant under complementation.

\begin{lemma}\label{lem:VCdim-difference-tri-hyper}
Let $H_1=(V,\mathcal{E}_1)$ and $H_2 = (V, \mE_2)$ be two tri-hypergraphs of VC-dimension $d$.
Let $\mathcal{D} = H_1 \setminus H_2$ be their difference tri-hypergraph.
Then, $\mathcal{D}$ has VC-dimension $O(d)$.
\end{lemma}

\subsection{Refined difference of metric-trigraphs}\label{subsec:VCdim-diff-metric}

For technical reasons which will become clear when trying to prove clustering results (namely \cref{lem:DNLH-cluster-intro}), we need a stronger statement about the difference of tri-hypergraphs. We could not prove this result in full generality for tri-hypergraphs of bounded VC-dimension, but the result holds for metric-trigraphs. We first need some definitions.

Let $T = (V, E, R)$ be a Hamming-trigraph with sensitivity $\varepsilon$ and threshold $\tau$. Let $N$ be such that $V \subseteq \{0, 1\}^N$.
We refine the trigraph $T$ by classifying the red edges into two types.
Formally, let $R_1 = \{vw \in R : d_H(v, w) \leq (\tau + \varepsilon/2) \cdot N\}$ and $R_2 = \{vw \in R : (\tau + \varepsilon/2) \cdot N < d_H(v, w)\}$.
We consider two refinements of $T$, which are $T_1 = (V, E, R_1)$ and $T_2 = (V, E \cup R_1, R_2)$.
The tri-hypergraph $\mD = T_1 \setminus T_2$ is the \emph{refined tri-hypergraph of differences} of $T$, see \cref{fig:diff}.
Observe that both $T_1$ and $T_2$ are Hamming-trigraphs with sensitivity $\varepsilon/2$, so they have VC-dimension $O\left(1/\varepsilon^2\right)$ by \cref{thm:VCdim-Hamming}.
It then follows from \cref{lem:VCdim-difference-tri-hyper} that $\mD$ has VC-dimension $O\left(1/\varepsilon^2\right)$.

\begin{restatable}{lemma}{VCdimdiff}\label{lem:VCdim-diff}
Let $T$ be a Hamming-trigraph with sensitivity $\varepsilon$ and let $\mD$ be its refined tri-hypergraph of differences.
Then, $\mD$ has VC-dimension $O\left(\frac{1}{\varepsilon^2}\right)$.
\end{restatable}

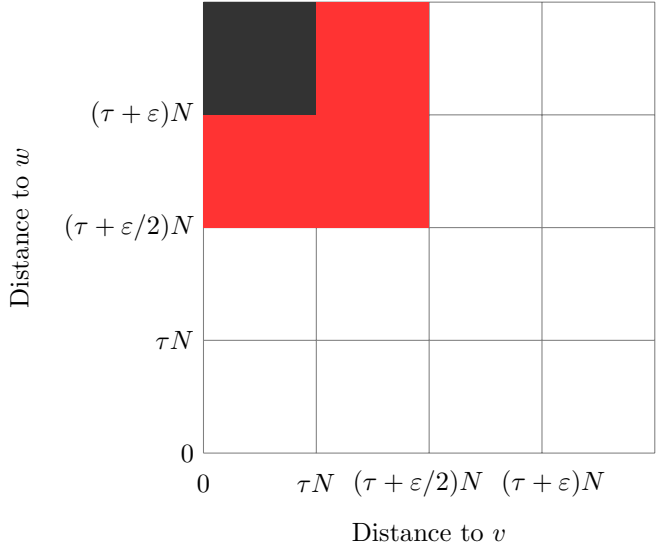
\begin{figure}[ht]
    \centering
    \hspace{-3cm}
    \begin{tikzpicture}
    \def\cellsize{1.5cm}
  
  \draw[step=\cellsize,gray] (0,0) grid (4*\cellsize,4*\cellsize);
  
  \fill[black!80] (\cellsize,4*\cellsize) rectangle (0,3*\cellsize);  
  
  \fill[red!80] (\cellsize,3*\cellsize) rectangle (0,2*\cellsize);   
  \fill[red!80] (\cellsize,3*\cellsize) rectangle (2*\cellsize,2*\cellsize);
  \fill[red!80] (\cellsize,3*\cellsize) rectangle (2*\cellsize,4*\cellsize);
  
  \node[anchor=center] at (0, -.4cm) {$0$};
  \node[anchor=center, align=center] at (1*\cellsize, -.4cm) {$\tau N$};
  \node[anchor=center, align=center] at (1.9*\cellsize, -.4cm) {$(\tau+\varepsilon/2)N$};
  \node[anchor=center, align=center] at (3.1*\cellsize, -.4cm) {$(\tau+\varepsilon)N$};
  
  \node[anchor=east, align=center] at (-0cm, 0) {$0$};
  \node[anchor=east, align=center] at (-0cm, 1*\cellsize) {$\tau N$};
  \node[anchor=east, align=center] at (-0cm, 2*\cellsize) {$(\tau+\varepsilon/2)N$};
  \node[anchor=east, align=center] at (-0cm, 3*\cellsize) {$(\tau+\varepsilon)N$};
  
  \node[anchor=north] at (2*\cellsize, -0.8) {Distance to $v$};
  \node[rotate=90, anchor=south] at (-2.2cm,2*\cellsize) {Distance to $w$};
    \end{tikzpicture}
    \caption{Edge of $\mD$ which is the difference between the tri-edge corresponding to the the neighborhood of $v$ in $T_1$ and the tri-edge corresponding to the neighborhood of $w$ in $T_2$. Observe that every vertex in the black zone is at distance at most $\tau N$ of $v$ and at least $(\tau + \varepsilon)N$ of $w$, and every vertex in the white zone is either at distance at least $(\tau + \varepsilon/2) N$ of $v$ or at distance at most $(\tau + \varepsilon/2) N$ of $w$.}
    \label{fig:diff}
\end{figure}

\subsection{Disjointness-trigraphs} \label{sec:VCdim-sets}

We now move to disjointness-trigraphs, for which we prove results analogous to Hamming-trigraphs. The proofs are not more complicated, but are less visual than for Hamming-trigraphs. Recall that a trigraph $T = (V, E, R)$ is a disjointness-trigraph with sensitivity $\varepsilon$ if there exists a set system $\mF$ on ground set $V$ such that $E = \{uv : \mF_{uv} = \emptyset\}$ and $R = \{uv : 0 < |\mF_{uv}| \leq \varepsilon |\mF|\}$.

\begin{lemma}\label{lem:VCdim-disj-trigraph}
Every disjointness-trigraph with sensitivity $\varepsilon$ has VC-dimension at most $\frac{1}{\varepsilon}$.
\end{lemma}

\begin{proof}
Let $T=(V,E,R)$ be a disjointness-trigraph with sensitivity $\varepsilon$ coming from a set system $\mF$.
Let $X \subseteq V$ be a shattered set of $T$: for every $Y \subseteq X$ there exists $v_Y \in V$ such that $N[v_Y] \cap X = (N[v_Y] \cup R(v_Y)) \cap X = Y$.
Let $\mY = \{v_Y: Y \subseteq X\}$, and consider the bipartite graph $G$ on vertex set $X \cup \mY$ where $xv_Y$ is an edge if $x \in W(v_Y)$ (or equivalently if $x \notin N[v_Y]$).
Since $\mY$ shatters $X$, the $2^{|X|}$ vertices in $\mY$ have all possible adjacencies on $X$.
Thus, for any subset $X' \subseteq X$, the number of common neighbors of the vertices of $X'$ is $2^{|X| - |X'|}$, and therefore $G$ has no complete bipartite subgraph with more than $|X'|2^{|X| - |X'|}$ edges, which is at most $2^{|X|-1}$ edges.

We show that $G$ has a complete bipartite subgraph of size $\varepsilon \cdot |X| \cdot 2^{|X|-1}$.
Consider the following experiment : take a uniformly random set $S \in \mF$ and let $Z$ be the random variable which counts the number of pairs $(x, v_Y) \in X \times \mY$ such that $S \in \mF_{xv_Y}$ (i.e. such that $\{x, v_Y\} \subseteq S$).
For every edge $xv_Y$ in $G$, we have that $x \in W(v_Y)$, which means that $|\mF_{xv_Y}| \geq \varepsilon |\mF|$, so the probability that $z \in \mF_{xv_Y}$ is at least $\varepsilon$.
Therefore, $\mathbb{E}[Z] \geq \varepsilon \cdot e(X, \mY) = \varepsilon \cdot |X| \cdot 2^{|X|-1}$, so there exists $S\in \mF$ for which $Z \geq \varepsilon \cdot |X| \cdot 2^{|X|-1}$.
Then, let $X'  = S \cap X$ and $\mY' = S \cap \mY$.
For every $x \in X'$ and $v_Y \in \mY'$ we have $S \in \mF_{xv_Y}$, so $xv_Y \notin E(T)$, i.e. $xv_Y \in E(G)$.
Since $|X'| \cdot |\mY'| = Z \geq \varepsilon \cdot |X| \cdot 2^{|X|-1}$ then $G$ has a complete bipartite subgraph with at least $\varepsilon \cdot |X| \cdot 2^{|X|-1}$ edges.

Putting together the upper bound and the lower bound, we get $\varepsilon \cdot |X| \cdot 2^{|X|-1} \leq 2^{|X|-1}$ so $|X| \leq \frac{1}{\varepsilon}$.
\end{proof}

Combining \cref{lem:VCdim-disj-trigraph} with \cref{lem:VCdim-intersection-tri-hyper} immediately gives the following result.

\begin{restatable}{lemma}{VCdimsetinter}\label{lem:VCdim-set-intersection}
If $T$ is a disjointness-trigraph with sensitivity $\varepsilon$ then $T \cap T$ has VC-dimension $O\left(\frac{1}{\varepsilon}\right)$.
\end{restatable}

The previous result can be generalized further, with essentially the same proof as \cref{lem:VCdim-disj-trigraph}, see \cref{sec:proofs-VCdim-sets}. 
We first need a definition.
Consider a set system $\mF$ on ground set $V$, and let $\mT \subseteq V^t$ be a set of $t$-tuples of elements of $V$.
For $v \in V$, we set $B(v) = \{(w^1, \ldots, w^t) \in \mT : \forall i \in [t], \mF_{vw^i} = \emptyset\}$ and $R(v) = \{(w^1, \ldots, w^t) \in \mT : \forall i \in [t], |\mF_{vw^i}| \leq \varepsilon |\mF| \} \setminus B(v)$.
The tri-hypergraph $\mI = (\mT, \mE)$ with $\mE = \{(B(v), R(v), W(v)): v \in V\}$ is the \emph{$\varepsilon$-disjointness-tri-hypergraph} of $V$ on $\mT$.

\begin{restatable}{lemma}{VCdimsetintertuple}\label{lem:VCdim-set-intersection-tuples}
Let $\mF$ be a set system on ground set $V$ and $\mT \subseteq V^t$ be a set of $t$-tuples of elements of $V$.
Let $\mI$ be the $\varepsilon$-disjointness-tri-hypergraph of $V$ on $\mS$.
Then, $\mI \cap \mI$ has VC-dimension at most $\frac{1}{\varepsilon}$.
\end{restatable}

The previous two results involved tri-hypergraphs arising from the intersection of disjointness-trigraphs. We now investigate tri-hypergraphs arising from the difference of disjointness-trigraphs.
It follows immediately from \cref{lem:VCdim-difference-tri-hyper,lem:VCdim-disj-trigraph} that the difference tri-hypergraph of a disjointness-trigraph has bounded VC-dimension.
However, to prove the clustering result for disjointness-trigraphs (\cref{lem:DNLS-cluster-intro}), we need a stronger result, analogous to \cref{lem:VCdim-diff}. We again need some definitions.

Let $T = (V, E, R)$ be a disjointess-trigraph with sensitivity $\varepsilon$, and consider a set system $\mF$ on ground set $V$ such that $E = \{uv: \mF_{uv} = \emptyset\}$ and $R = \{uv: 0 < |\mF_{uv}| \leq \varepsilon |\mF|\}$.
We refine the trigraph $T$ by classifying the red edges into two types.
Formally, let $R_1 = \{vw \in R : |\mF_{vw}| \leq \varepsilon |\mF|/2\}$ and $R_2 = \{vw \in R : |\mF_{vw}| > \varepsilon |\mF|/ 2\}$.
We consider two refinements of $T$, which are $T_1 = (V, E, R_1)$ and $T_2 = (V, E \cup R_1, R_2)$.
The tri-hypergraph $\mD = T_1 \setminus T_2$ is the \emph{refined tri-hypergraph of differences} of $T$.

To prove that $\mD$ has bounded VC-dimension, we would like to apply \cref{lem:VCdim-difference-tri-hyper,lem:VCdim-disj-trigraph}. However, even though $T_1$ is a disjointness-trigraph, $T_2$ is not, and we could only prove a bound of $O(1/\varepsilon^2)$ on the VC-dimension of $T_2$. We therefore take a different approach.

\begin{restatable}{lemma}{VCdimsetdiff}\label{lem:VCdim-set-diff}
Let $T$ be a disjointness-trigraph with sensitivity $\varepsilon$ and let $\mD$ be its refined tri-hypergraph of differences.
Then, $\mD$ has VC-dimension $O\left(\frac{1}{\varepsilon} \cdot \log\left(\frac{1}{\varepsilon}\right)\right)$.
\end{restatable}

\begin{proof}
Write $\mathcal{D} = (V, \mE)$, with $\mE = \{(B(vw), R(vw), W(vw)) : v, w \in V\}$ and consider a set system $\mF$ on ground set $V$ with $N \coloneqq |\mF|$ sets such that $E = \{uv: \mF_{uv} = \emptyset\}$ and $R = \{uv: 0 < |\mF_{uv}| \leq \varepsilon N\}$. 
By definition of $\mathcal{D}$, for every $v, w \in V$, if $x \in B(vw)$ then $\mF_{xv} = \emptyset$ and $|\mF_{xw}| \geq \varepsilon N/2$, and if $x \in W(vw)$ then either $|\mF_{xv}| \geq \varepsilon N/2$ or $\mF_{xw} = \emptyset$.

Let $X \subseteq V$ be a set which is shattered by $\mathcal{D}$. For every subset $Y \subseteq X$, there exists a pair $p_{Y} = (v_{Y}, w_{Y}) \in V^2$ such that $B(v_{Y}w_{Y}) \cap X = (B(v_{Y}w_{Y}) \cup R(v_{Y}w_{Y})) \cap X = Y$.
Let $U$ be a random multi-subset of $\{p_{Y} : Y \subseteq X\}$ of size $|X|$ obtained by picking with repetition uniformly random elements of $\{p_{Y} : Y \subseteq X\}$.
Consider the bipartite graph $\Gamma$ on vertex set $X \cup U$ where $xu$ is an edge if $x \in B(u)$ (or equivalently if $x \notin W(u)$).
Then, $\Gamma$ is a random bipartite graph where each edge between $X$ and $U$ is present independently of all others with probability $1/2$.

We say that a sub-set system $\mS$ of $\mF$ \emph{represents} $\Gamma$ if for every $u = vw \in U$ and $x \in X$, $ux$ is an edge of $\Gamma$ if and only if $\mS_{vx} = \emptyset$ and $\mS_{wx} \neq \emptyset$.
Let $E_0$ be the event that there exists a sub-set system $\mS$ of $\mF$ of size $\ell \coloneqq 2\ln(9|X|^2)/\varepsilon$ which represents $\Gamma$.

We start by showing that there always exists such a sub-set system $\mS$, independently of $\Gamma$.
Let $V' \subseteq V$ be the set of all vertices which either belong to $X$ or belong to a pair $v_Yw_Y \in U$. Note that $|V'| \leq |X| + 2 \cdot |U| = 3 \cdot |X|$.
Observe that if $\mS \subseteq \mF$ satisfies $\mS_{vv'} \neq \emptyset$ whenever $v, v' \in V'$ satisfy $|\mF_{vv'}| \geq \varepsilon N/2$ then $\mS$ represents $\Gamma$.
Consider the set system $\mH = \{\mF_{vv'} : v, v' \in V' \text{ and } |\mF_{vv'}| \geq \varepsilon N/2\}$ on ground set $\mF$.
Then, $\mH$ consists of at most $9|X|^2$ sets, each of size at least $\varepsilon N/2$, from some universe of size $N$. 
Therefore, there exists a transversal $\mS \subseteq \mF$ of $\mH$ of size at most $2\ln(9|X|^2)/\varepsilon$ (which can be found by a greedy algorithm).
Such an $\mS$ represents $\Gamma$.

We now give an upper bound on $\mathbb{P}[E_0]$. 
For any sub-set system $\mS$ of $\mF$ of size $\ell$, the probability that $\mS$ represents $\Gamma$ is at most $2^{-|X|^2}$ since $\Gamma$ is a random bipartite graph with $|X|^2$ edges. 
There are at most $2^{\ell|X|}$ different set systems of size $\ell$ on $X$, therefore by a union bound, the probability that one of them represents $\Gamma$ is at most $2^{\ell|X| - |X|^2}$.

Then, $1 = \mathbb{P}[E_0] \leq 2^{\ell|X| - |X|^2}$ so $|X| \leq \ell$.
Thus, $|X| \leq 2\ln(9|X|^2)/\varepsilon = \frac{4}{\varepsilon}\ln(3|X|)$.
By \cref{lem:tech}, we get $|X| \leq \frac{8}{\varepsilon} \cdot \ln\left(\frac{12}{\varepsilon}\right)$.
\end{proof}

\section{Dense Neighborhood Lemmas}\label{sec:DNL}

Putting together the results of the previous two sections, we can finally state and prove several variants of the Dense Neighborhood Lemma. In \cref{subsec:dom-trigraphs}, we prove that dense disjointness-trigraphs and dense metric-trigraphs have small dominating sets. In \cref{subsec:nets-higher-order}, we prove that the corresponding intersection tri-hypergraphs have small $\delta$-nets. In \cref{subsec:clustering}, we show that disjointness-trigraphs and metric-trigraphs admit clusterings with a bounded number of parts. Finally, in \cref{subsec:ultrastrong}, we analyze more deeply the interactions between the clusters to obtain a Regularity-Lemma-type partition into an exponential number of parts with 0/1 densities between the parts.

\subsection{Dominating dense trigraphs}\label{subsec:dom-trigraphs}

We start by deriving \cref{lem:DNL-set-anti,lem:DNL-Ham-antipodal} from the results of \cref{sec:vcdimtri} and from \cref{thm:triHW}. 

We first recall the relevant definitions for \cref{lem:DNL-set-anti}.
If $\mF$ is a set system on ground set $V$, $\mF_{x}$ is the set of hyperedges containing $x$, and $\mF_{xy}=\mF_x\cap \mF_y$. We set $D_\varepsilon(x)=\{y \in V : |\mF_{xy}|\leq \varepsilon  |\mF|\}$ and $D(x)\coloneqq D_0(x)$. 
A set $X$ is an $\varepsilon$-covering if $D_\varepsilon(X)=V$.
The disjointness-ratio of $\mF$ is the minimum of $|D(x)|/|V|$ over all $x\in V$. 

\DNLSanti*

\begin{proof}
Let $\mF$ be a set system on ground set $V$ with disjointness-ratio $\delta$.
Consider the trigraph $T = (V, E, R)$ where $E = \{uv: \mF_{uv} = \emptyset\}$ and $R = \{uv: 0 < |\mF_{uv}| \leq \varepsilon |\mF|\}$.
Then, $T$ is a disjointness-trigraph with sensitivity $\varepsilon$, so by \cref{lem:VCdim-disj-trigraph} $T$ has VC-dimension $O\left(\frac{1}{\varepsilon}\right)$. Since $N(v) = D(v)$ for every $v \in V$, we have $|N(v)| \geq \delta |V|$ by assumption. Hence, by \cref{thm:triHW}, $T$ has a dominating set  $X \subseteq V$ of size $O\left(\frac{1}{\varepsilon\delta}\log(\frac{1}{\delta})\right)$.
We argue that $X$ is an $\varepsilon$-covering of $\mF$. For every $v\in V$, there is some vertex $x \in X$ such that $x \in N(v) \cup R(v) = D_{\varepsilon}(v)$, hence $v\in  D_{\varepsilon}(X)$.
\end{proof}

We now move to Hamming-trigraphs. The next result is stated in antipodal setting, but it also holds in the local setting, with a similar proof.

\DNLHamanti*

\begin{proof}
Consider the trigraph $T = (V, E, R)$ with $E = \{uv : d_H(u, v) \geq c \cdot N\}$ and $R = \{uv : (c - \varepsilon) \cdot N \leq d(u, v) < c \cdot N\}$.
Note that the complement $\overline{T}$ of $T$ is a Hamming-trigraph with sensitivity $\varepsilon$. By \cref{thm:VCdim-Hamming} the VC-dimension of $\overline{T}$ is $O(\varepsilon^{-2})$.
Since the VC-dimension is invariant under complementation, $T$ also has VC-dimension $O(\varepsilon^{-2})$.
By \cref{thm:triHW}, $T$ has a dominating set $X$ of size $|X|= \text{poly}(\varepsilon^{-1}, \delta^{-1})$.

Then, for any $v \in V$ there exists $x \in X$ such that $x \in (B(v) \cup R(v))$, which means that $d_H(x, v) \geq (c - \varepsilon) \cdot N$.
\end{proof}

For future reference, we also state the analogous result for spherical-trigraphs in the antipodal setting. Once again, it also holds in the local setting. The proofs of these results are very similar to the proof of \cref{lem:DNL-Ham-antipodal}.

\begin{restatable}{lemma}{DNLSantipodal}\label{lem:DNL-spherical-antipodal}
Let $V \subseteq \mathbb{S}^N$ be such that for every $v\in V$ there are at least $\delta |V|$ points $u$ such that $\langle u,v\rangle \leq c$. Then for every $\varepsilon>0$, there exists a set $X\subseteq V$ of size $\textup{poly}(1/\varepsilon)$ such that for every $v\in V$ there exists $x\in X$ with $\langle x,v\rangle \leq c+\varepsilon$.
\end{restatable}

\subsection{\texorpdfstring{$\delta$}{delta}-nets for intersection tri-hypergraphs}\label{subsec:nets-higher-order}

The take-away message from \cref{subsec:dom-trigraphs} is ``Disjointness-trigraphs and metric-trigraphs have small $\delta$-nets''. With that in mind, the results of this section should be understood as ``Intersection tri-hypergraphs of disjointess-trigraphs and metric-trigraphs have small $\delta$-nets''. The proofs are again straightforward, simply combining results from \cref{sec:vcdimtri} with \cref{thm:HW-trihyper}.

\begin{lemma}\label{lem:DNLS-intersection}
For every set system $\mF$ on ground set $V$, there exists $X \subseteq V$ of size $O\left(\frac{1}{\varepsilon\delta}\log\left(\frac{1}{\delta}\right)\right)$ such that whenever $u, v \in V$ satisfy $|D(u) \cap D(v)| \geq \delta|V|$, there exists $x \in X$ such that $x \in D_{\varepsilon}(u) \cap D_{\varepsilon}(v)$.
\end{lemma}

\begin{proof}
Consider the trigraph $T = (V, E, R)$ with $E = \{uv: \mF_{uv} = \emptyset\}$ and $R = \{uv: 0 < |\mF_{uv}| \leq \varepsilon |\mF|\}$.
By \cref{lem:VCdim-set-intersection}, the tri-hypergraph $T \cap T$ has VC-dimension $O\left(\frac{1}{\varepsilon}\right)$.

By \cref{thm:HW-trihyper}, there exists a $\delta$-net $X \subseteq V$ of $T~\cap~T$ of size $O\left(\frac{1}{\varepsilon\delta}\log(\frac{1}{\delta})\right)$.
Consider now $u, v \in V$ such that $|D(u) \cap D(v)| \geq \delta|\mF|$.
Then, in $T \cap T$ the tri-edge $(B(uv), R(uv), W(uv))$ satisfies $|B(uv)| = |D(u) \cap D(v)| \geq \delta|\mF|$ so there exists $x \in X$ such that $x \in B(uv) \cup R(uv)$.
This means that $|\mF_{ux}| \leq \varepsilon |\mF|$ and $|\mF_{vx}| \leq \varepsilon |\mF|$ so $x \in D_{\varepsilon}(u) \cap D_{\varepsilon}(v)$.
\end{proof}

This result can be generalized to the intersection of disjointness-tri-hypergraphs. However, the statement might seem somewhat obscure. To get some intuition on how it can be used, see the discussion after the statement of \cref{thm:homomorphism-kt-free}.

If $\mF$ is a set system on ground set $V$ and $\mT \subseteq V^t$ is a set of $t$-tuples of elements of $V$, for every $v \in V$, we set $D^{\mT}(v) = \{(u^1, \ldots, u^t) \in \mT : \forall i \in [t], \mF_{vu^i} = \emptyset\}$ and $D^{\mT}_{\varepsilon}(v) = \{(u^1, \ldots, u^t) \in \mT : \forall i \in [t], |\mF_{vu^i}| \leq \varepsilon |\mF|\}$.

\begin{lemma}\label{lem:DNLS-intersection-higher-order}
Let $\mF$ be a set system on ground set $V$ and let $\mT \subseteq V^t$ be a set of $t$-tuples of elements of $V$.
There exists $\mC \subseteq \mT$ of size $O\left(\frac{1}{\varepsilon\delta}\log\left(\frac{1}{\delta}\right)\right)$ such that for every $u, v \in V$ such that $|D^{\mT}(u) \cap D^{\mT}(v)| \geq \delta|\mT|$, there exists $C \in \mC$ such that $C \in D^{\mT}_{\varepsilon}(u) \cap D^{\mT}_{\varepsilon}(v)$.
\end{lemma}

The proof is exactly the same as the previous one, except that we use \cref{lem:VCdim-set-intersection-tuples} instead of \cref{lem:VCdim-set-intersection}.

Similar results also hold for metric-trigraphs. We state the analogue of \cref{lem:DNLS-intersection} in the antipodal Hamming setting as an illustration. Analogous results also hold in the local setting, as well as for spherical-trigraphs. The proofs again follow the proof of \cref{lem:DNLS-intersection}. An analogue of \cref{lem:DNLS-intersection-higher-order} can be proved using an adaptation of \cref{lem:VCdim-set-intersection-tuples} for metric-trigraphs.

\begin{lemma}
For every set $V \subseteq \{0, 1\}^N$, there exists $X \subseteq V$ of size $\textup{poly}(\varepsilon^{-1}, \delta^{-1})$ such that for every $u, v \in V$ such that there are at least $\delta |V|$ vectors $w \in V$ at distance at least $c \cdot N$ of both $u$ and $v$, there exists $x \in X$ at distance at least $(c - \varepsilon) \cdot N$ of both $u$ and $v$.
\end{lemma}

\subsection{Clustering variants}\label{subsec:clustering}

Several results on graphs of bounded VC-dimension rely on the fact that their vertex set can be partitioned into a bounded number of parts (often called \emph{clusters}), such that vertices in the same cluster almost have the same neighborhood, see e.g. \cite{FPS19,NSS24b}. This is usually done using the so-called Haussler Packing Lemma \cite{Haussler95}. For trigraphs, the standard proofs of the Packing Lemma break because the red edges cannot be used to classify the vertices. Without an analogue of the Packing Lemma for trigraphs, we could not get the clustering results in the most general form. 
Fortunately, we were able to recover these results for disjointness-trigraphs and metric-trigraphs, which is sufficient for our applications. However, it would be very interesting to determine whether similar results hold in full generality for trigraphs of bounded VC-dimension.

Note that in the statements of \cref{lem:DNLS-cluster-intro,lem:DNLH-cluster}, the bounds are of the form $2^{\text{poly}(1/\varepsilon, 1/\eta)}$, while the corresponding bounds for graphs of VC-dimension $d$ are of the form $(1/\eta)^{O(d)}$. In most of our applications we will have $\eta = \text{poly}(\varepsilon)$ and VC-dimension polynomial in $1/\varepsilon$, hence the order of our bound is actually not a bottleneck. Furthermore, \cref{lem:lower-bound-clustering} shows that this dependency on $\varepsilon$ is required, and not just an artifact of the proof.

We start with the clustering statement on disjointness-trigraphs. If $\mF$ is a set system on ground set $V$, recall that for $v \in V$, we have $D(v) = \{u \in V: \mF_{uv} = \emptyset\}$, and $D_{\varepsilon}(v) = \{u \in V: |\mF_{uv}| \leq \varepsilon |\mF|\}$. 
An \emph{$(\varepsilon, \eta)$-cluster} is a subset $X$ of $V$ such that $|D(u)\setminus D_\varepsilon (v)|\leq \eta |V|$ for all $u, v\in X$, and an \emph{$(\varepsilon, \eta)$-clustering} is a partition into $(\varepsilon, \eta)$-clusters. 
Its size is the number of clusters. 
Observe that the following result directly implies \cref{lem:DNLS-cluster-intro}.

\begin{lemma}\label{lem:DNLS-cluster}
Every set system has a $(\varepsilon, \eta)$-clustering of size $2^{\textup{poly}(\varepsilon^{-1}, \eta^{-1})}$.
\end{lemma}

\begin{proof}
Let $\mF$ be a set system on ground set $V$.
Consider the trigraph $T = (V, E, R)$ where $E = \{uv: \mF_{uv} = \emptyset\}$ and $R = \{uv: 0 < |\mF_{uv}| \leq \varepsilon |\mF|\}$.
Let $\mD$ be its refined tri-hypergraph of differences, whose tri-edges are the triples $\{(B(uv), R(uv), W(uv)) : u, v \in V\}$.
By \cref{lem:VCdim-set-diff}, $\mathcal{D}$ has VC-dimension $\textup{poly}(\varepsilon^{-1})$. 

By \cref{thm:HW-trihyper}, there exists an $\eta$-net $X$ of $\mD$ of size $\textup{poly}(\varepsilon^{-1}, \eta^{-1})$.
Note that by definition of $\mD$, if $x \in B(uv) \cup R(uv)$ then $|\mF_{ux}| \leq \varepsilon |\mF|/2$ and $|\mF_{vx}| > \varepsilon |\mF|/2$.
Let $\mathcal{C} = (C_1, \ldots, C_t)$ be the partition of $V$ such that $u, v \in V$ are in the same part if and only if for every $x \in X$, $|\mF_{ux}| \leq \varepsilon|\mF|/2$ if and only if $|\mF_{vx}| \leq \varepsilon|\mF|/2$.
Note that $t = 2^{\textup{poly}(\varepsilon^{-1}, \eta^{-1})}$.

Consider any $u, v \in V$ in the same cluster.
Then, there is no $x \in X$ such that $|\mF_{ux}| \leq \varepsilon|\mF|/2$ and $|\mF_{vx}| > \varepsilon|\mF|/2$. This means that no vertex of $X$ belongs to $B(uv) \cup R(uv)$. Therefore, $|B(uv)| \leq \eta \cdot |V|$ as $X$ is an $\eta$-net.
This proves that $|D(u) \setminus D_{\varepsilon}(v)| \leq \eta \cdot |V|$, so $\mC$ is a $(\varepsilon, \eta)$-clustering.
\end{proof}

The previous proof crucially relies on the fact that we considered the refined tri-hypergraph of differences of $T$, and not just the standard tri-hypergraph of differences of $T$. Indeed, in the standard hypergraph of differences of $T$, if $x \in B(uv) \cup R(uv)$ then it can be the case that $|\mF_{ux}| = |\mF_{vx}|$ or even that $|\mF_{ux}| > |\mF_{vx}|$. Then, defining the clusters according to the sizes of $\mF_{ux}$ for all $x \in X$ would not be sufficient to ensure that $|B(xy)| \leq \eta \cdot n$ whenever $u$ and $v$ are in the same cluster.

We next prove the corresponding statement for Hamming-trigraphs, which directly implies \cref{lem:DNLH-cluster-intro}. Again, using the refined tri-hypergraph of differences instead of the standard tri-hypergraph of differences is key in the proof. Note that a similar result also holds for a finite set of points in $\mathbb{S}^N$.

\begin{lemma}\label{lem:DNLH-cluster}
Let $V$ be an $n$-set of $\{0,1\}^N$ and $c, \varepsilon, \eta >0$. Then $V$ can be partitioned into $2^{\textup{poly}(\eta^{-1},\varepsilon^{-1})}$ clusters such that if $u, v$ are in the same cluster, there are at most $\eta \cdot n$ points $w$ of $V$ such that $d_H(u,w)\geq c\cdot N$ and $d_H(v,w)\leq (c-\varepsilon)\cdot N$.
\end{lemma} 

\begin{proof}
Consider the trigraph $T = (V, E, R)$ where $E = \{uv : d_H(u, v) \geq c \cdot N\}$ and $R = \{uv : (c - \varepsilon) \cdot N \leq d_H(u, v) < c \cdot N\}$.
Let $\mD$ be its refined tri-hypergraph of differences.
Write $\mathcal{D} = (V, \mE)$ with $\mE = \{(B(uv), R(uv), W(uv)) : u, v \in V\}$, see \cref{fig:diff}.
By \cref{lem:VCdim-diff}, $\mathcal{D}$ has VC-dimension $\textup{poly}(\varepsilon^{-1})$. 

By \cref{thm:HW-trihyper}, there exists an $\eta$-net $X$ of $\mD$ of size $\textup{poly}(\varepsilon^{-1}, \eta^{-1})$.
Note that if $x \in B(vu) \cup R(vu)$ then $d_H(v, x) \leq (c - \varepsilon/2) \cdot N$ and $d_H(u, x) > (c- \varepsilon/2) \cdot N$.
Let $\mathcal{C} = (C_1, \ldots, C_t)$ be the partition of $V$ into clusters such that $u, v \in V$ are in the same cluster if and only if for every $x \in X$, $d_H(u, x) \leq (c - \varepsilon/2) \cdot N$ if and only if $d_H(v, x) \leq (c - \varepsilon/2) \cdot N$.
Observe that the number of clusters is $2^{\textup{poly}(\varepsilon^{-1}, \eta^{-1})}$.

Consider any $u, v \in V$ which belong to the same cluster.
Then, there is no $x \in X$ such that $d_H(v, x) \leq (c - \varepsilon/2) \cdot N$ and $d_H(u, x) > (c - \varepsilon/2) \cdot N$. This means that no vertex of $X$ belongs to $B(vu) \cup R(vu)$. Therefore, $|B(vu)| \leq \eta \cdot n$.
This proves that there are at most $\eta \cdot n$ vectors $w \in V$ such that $w$ is at distance at least $c \cdot N$ of $u$ and at most $(c - \varepsilon/2) \cdot N$ of $v$.
\end{proof}

We now show that the dependency of the number of clusters on $\varepsilon$ cannot be significantly improved.

\begin{lemma}\label{lem:lower-bound-clustering}
There exists $\eta > 0$ such that for every large enough $N$, if $\{0, 1\}^N$ is partitioned into clusters such that whenever $u, v$ are in the same cluster, there are at most $\eta \cdot 2^N$ points $w \in \{0, 1\}^N$ such that $d_H(u, w) \geq N/2$ and $d_H(v, w) \leq N/2-\sqrt{N}$, then the number of clusters is at least $2^{\Omega(N)}$.
\end{lemma}

\begin{proof}
First, note that $\binom{N}{N/2 - 2\sqrt{N}} \sim \binom{N}{N/2} \cdot e^{-8} \sim \sqrt{\frac{2}{\pi \cdot N}} \cdot e^{-8} \cdot 2^N$ (see e.g. \cite{asymptopia}).
Set $\eta = \sqrt{\frac{2}{\pi}} \cdot \frac{e^{-8}}{4}$ and let $N$ be large enough so that $\binom{N}{N/2 - 2\sqrt{N}} \geq 2\eta \cdot 2^N/\sqrt{N}$.
Observe then that $$\left|\left\{x \in \{0, 1\}^N : N/2-2\sqrt{N} \leq \norm{x}_1 \leq N/2-\sqrt{N}\right\}\right| \geq 2\eta \cdot 2^N.$$

Suppose that $\{0, 1\}^N$ is partitioned into clusters which satisfy the desired property, and consider $u, v \in \{0, 1\}^N$ which belong to the same cluster.
We first prove that $d_H(u, v) \leq 6\sqrt{N}$.
Suppose by contradiction that $d_H(u, v) > 6\sqrt{N}$.
Without loss of generality, we can assume that $u = 0^N$.
Let $r \in \left[N/2-2\sqrt{N}, N/2-\sqrt{N}\right]$, and pick $x \in \{0, 1\}^N$ uniformly at random such that $d_H(u, x) = r$, or equivalently $\norm{x}_1 = r$.
Let $K$ be the number of 1-coordinates shared by $x$ and $v$.
Then, $K$ follows a hypergeometric distribution, so by standard tail bounds (see \cite{skala13} for instance), $\mathbb{P}\left[K - \mathbb{E}[K] \geq \sqrt{N}\right] \leq \exp\left(-2 \cdot \left(\sqrt{N}/r\right)^2 \cdot r\right) = \exp\left(-2N/r\right) \leq \exp(-4)$ since $r \leq N/2$.

Observe that $\mathbb{E}[K] = r \cdot \norm{v}_1 /N= r \cdot d_H(u, v)/N \leq d_H(u, v)/2$ and $d_H(x, v) = r + d_H(u, v) - K$.
Suppose that $d_H(x, v) \leq N/2$.
Then, $K \geq r + d_H(u, v) - N/2$ so $K - \mathbb{E}[K] \geq r + d_H(u, v)/2 - N/2 \geq \sqrt{N}$.
Thus, $\mathbb{P}[d_H(x, v) \leq N/2] \leq \mathbb{P}\left[K - \mathbb{E}[K] \geq \sqrt{N}\right] \leq \exp(-4) < 1/2$.

This proves that $$\left|B\left(u, N/2-\sqrt{N}\right) \setminus B\left(v, N/2\right)\right| > 1/2 \cdot \left|\left\{x \in \{0, 1\}^N : N/2-2\sqrt{N} \leq \norm{x}_1 \leq N/2-\sqrt{N}\right\}\right| \geq \eta \cdot 2^N.$$
This contradicts the fact that $u$ and $v$ belong to the same cluster.

Therefore, the number of points in a cluster is at most $$\sum_{i = 0}^{6\sqrt{N}}\binom{N}{i} \leq \left(\frac{eN}{6\sqrt{N}}\right)^{6\sqrt{N}} \leq 2^{6\sqrt{N} \cdot \log\left(e\sqrt{N}/6\right)}.$$
This in turn implies that the number of clusters is at least $2^{N - 6\sqrt{N} \cdot \log\left(e\sqrt{N}/6\right)} = 2^{\Omega(N)}$.
\end{proof}

\subsection{Ultra-Strong Regularity Lemma}\label{subsec:ultrastrong}

Given a graph $G$ and two vertex subsets $V_1$ and $V_2$ of $G$, the \emph{edge density} between $V_1$ and $V_2$ is $\frac{|E(V_1, V_2))|}{|V_1| \cdot |V_2|}$, where $E(V_1, V_2)$ is the set of edges of $G$ with one endpoint in $V_1$ and the other in $V_2$.
A pair of subsets $(V_1, V_2)$ is \emph{$\varepsilon$-homogeneous} if the edge density between $V_1$ and $V_2$ is at most $\varepsilon$ or at least $1 - \varepsilon$.

Lov\'asz and Szegedy \cite{LS10}  proved that the Regularity Lemma can be strengthened for graphs of bounded VC-dimension. More precisely they proved that for every integer $d$ and every $\varepsilon > 0$, there exists an integer $M = M(\varepsilon)$ such that for every graph $G$ of VC-dimension at most $d$, there exists a partition of $V(G)$ into at most $M$ parts such that all but at most an $\varepsilon$-fraction of the pairs are $\varepsilon$-homogeneous.
They obtained a bound of the form $M(\varepsilon) \leq (1 / \varepsilon)^{O(d^2)}$.
Alon, Fischer and Newman \cite{AFN07} later improved the bound to $(d/\varepsilon)^{O(d)}$ for bipartite graphs. Fox, Pach and Suk \cite{FPS19} generalized these results to hypergraphs, and obtained the bound $M(\varepsilon) \leq O_d\left(1/\varepsilon^{2d+1}\right)$ for graphs.

A straightforward adaptation of their proof yields the same result for trigraphs of bounded VC-dimension.
Given a trigraph $T = (V, E, R)$ and two vertex subsets $V_1$ and $V_2$ of $V$, the \emph{edge density} between $V_1$ and $V_2$ is $\frac{|E(V_1, V_2))|}{|V_1| \cdot |V_2|}$.
The \emph{non-edge density} between $V_1$ and $V_2$ is $\frac{|W(V_1, V_2))|}{|V_1| \cdot |V_2|}$, where $W(V_1, V_2)$ is the set of non-edges of $T$ with one endpoint in $V_1$ and the other in $V_2$.
A pair of subsets $(V_1, V_2)$ is \emph{$\varepsilon$-homogeneous} if either the edge density between $V_1$ and $V_2$ is at most $\varepsilon$, or the non-edge density between $V_1$ and $V_2$ is at most $\varepsilon$.

A key tool in the proof of Fox, Pach and Suk is Haussler's Packing Lemma, which they use to do clustering. We will mimick this by using \cref{lem:DNLH-cluster,,lem:DNLS-cluster}. In a bipartite trigraph $B=(V_1\cup V_2, E ,R)$, a \emph{$V_1$-conflict} is a triple $(v_1, v_2, v_2') \in V_1 \times V_2 \times V_2$ such that one of $v_1v_2, v_1v_2'$ is in $E$, while the other is not in $E \cup R$, i.e. one is an edge and the other is a non-edge. A \emph{conflict} is a $V_1$-conflict or a $V_2$-conflict.
The other key Lemma of the proof of Fox, Pach and Suk can be adapted as follows.

\begin{lemma}\label{lem:technical-usrl}
Let $0 < \eta < 1/4$ and $B = (V_1 \cup V_2, E, R)$ be a bipartite trigraph such that $|V_1| = |V_2| = m$.
If $(V_1, V_2)$ is not $\eta$-homogeneous and has at most $\eta m^2/2$ red edges then there are at least $\eta m^3/4$ conflicts.
\end{lemma}

\begin{proof}
Let $\eta_1$ be the fraction of triples $(v_1, v_2, v_2') \in V_1 \times V_2 \times V_2$ which are $V_1$-conflicts, and define $\eta_2$ similarly.
It suffices to show that $\eta_1 + \eta_2 \geq \eta/2$.
Pick $v_1, v'_1 \in V_1$ and $v_2, v'_2 \in V_2$ uniformly at random with repetition.
Let $e_0 = \{v_1, v_2\}$, $e_1 = \{v_1, v'_2\}$ and $e_2 = \{v'_1, v'_2\}$.
Let $X$ be the event that one of $e_0, e_2$ is in $E$ while the other is not in $E \cup R$.
Observe that $e_0$ and $e_2$ are independent uniformly random pairs of vertices from $V_1 \times V_2$.

Let $p = \mathbb{P}[e_0 \in E]$. Then, $\mathbb{P}[e_0 \notin E \cup R] \geq 1 - p - \eta/2$ since there are at most $\eta m^2/2$ red edges between $V_1$ and $V_2$.
By symmetry between $e_0$ and $e_2$ we have $\mathbb{P}[e_2 \in E] = p$ and $\mathbb{P}[e_2 \notin E \cup R] \geq 1 - p - \eta/2$.
By independence between $e_0$ and $e_2$ we get $\mathbb{P}[X] \geq 2p(1-p-\eta/2)$.
Since $(V_1, V_2)$ is not $\eta$-homogeneous then $p \geq \eta$ and $1 - p - \eta/2 \geq \eta$ so $p \leq 1-3\eta/2$.
The value $2p(1-p-\eta/2)$ is minimized either for $p = \eta$ or for $p = 1-3\eta/2$. 
In both cases we have $2p(1-p-\eta/2) = 2\eta(1-3\eta/2)$ so $\mathbb{P}[X] \geq 2\eta(1 - 3\eta/2) \geq \eta$ since $\eta \leq 1/4$.

Let $X_1$ be the event that $(v_1, v_2, v'_2)$ is a $V_1$-conflict and $X_2$ be the event that $(v'_2, v_1, v'_1)$ is a $V_2$-conflict. Let $Y$ be the event that at least one of $X_1, X_2$ occurs. Then, $\mathbb{P}[Y] \leq \mathbb{P}[X_1] + \mathbb{P}[X_2] = \eta_1 + \eta_2$.
On the other hand, if $X$ occurs then either $Y$ occurs or $e_1 \in R$, so $\eta_1 + \eta_2 \geq \mathbb{P}[Y] \geq \mathbb{P}[X] - \mathbb{P}[e_1 \in R] \geq \eta - \eta/2 = \eta/2$.
\end{proof}

We can now prove the Ultra-Strong Regularity Lemma for disjointness-trigraphs.
Replacing the application of \cref{lem:DNLS-cluster} by \cref{lem:DNLH-cluster} would show that the same result also holds for Hamming-trigraphs.

\begin{theorem}
Let $\varepsilon \in (0, 1)$ and $\eta \in (0, 1/4)$. Let $T = (V, E,R)$ be an $n$-vertex disjointness-trigraph with sensitivity $\varepsilon$ such that $|R| \leq \eta^2 \cdot n^2/8$.
Then, there is an equitable partition $(V_1, \ldots, V_K)$ of $V$ with $8/\eta \leq K \leq 2^{\textup{poly}(\varepsilon^{-1}, \eta^{-1})}$ parts such that all but an $\eta$-fraction of the pairs of parts are $\eta$-homogeneous.
\end{theorem}

\begin{proof}
Set $\delta = \frac{\eta^2}{16}$.
By \cref{lem:DNLS-cluster}, there exists a partition $\mathcal{Q} = (U_1, \ldots, U_{t})$ of $V$ with $t = 2^{\text{poly}(\varepsilon^{-1}, \delta^{-1})} = 2^{\text{poly}(\varepsilon^{-1}, \eta^{-1})}$ such that whenever $x,y$ belong to the same part, $|(N(x) \cap W(y)) \cup (W(x) \cap N(y))| \leq 2\delta \cdot n$.
Set $K = 4t/\eta$, and partition arbitrarily each $U_i$ into parts of size $n/K$ and possibly one additional part of size less than $n/K$.
Put these additional parts together and partition the resulting set into parts of size $n/K$ to obtain an equitable partition $\mathcal{P} = (V_1, \ldots, V_K)$ into $K$ parts.
By construction, there are at most $t$ parts $V_i$ which are not contained in a part of $\mathcal{Q}$.
Thus, the fraction of unordered pairs $(V_{i}, V_{j})$ such that at least one of them is not contained in a part of $\mathcal{Q}$ is at most $t/K = \eta/4$.
Since there are at most $\eta^2 \cdot n^2/8$ red edges, the number of unordered pairs $(V_i, V_j)$ with at least $\eta \cdot (n/K)^2/2$ red edges between them is at most $\eta K^2/4$.
Let $X$ denote the set of unordered pairs $(V_i, V_j)$ such that both $V_i$ and $V_j$ are contained in a part of $\mathcal{Q}$ (not necessarily the same), which are not $\eta$-homogeneous and such that there are at most $\eta \cdot (n/K)^2/2$ red edges between them.

Let $C$ be the set of all conflicts $(u, v, w)$ appearing in the pairs $(V_i, V_j) \in X$.
For every $(V_i, V_j) \in X$, we have that $(V_i, V_j)$ is not $\eta$-homogeneous and that there are at most $\eta \cdot (n/K)^2/2$ red edges between $V_i$ and $V_j$ so by \cref{lem:technical-usrl}, $(V_i, V_j)$ yields at least $\eta/4 \cdot (n/K)^3$ conflicts $(u, v, w)$ in $C$.
Therefore, $|C| \geq |X| \cdot \eta/4 \cdot (n/K)^3$.
On the other hand, given $V_i$ and $b,b'\in V_i$, consider all the $V_j$-conflicts $(a, b, b')\in C$ where $a \in V_j$ and $(V_i, V_j) \in X$. Since $(V_i, V_j) \in X$ then $V_i$ is contained in a part of $\mathcal{Q}$, so $|(N(b) \cap W(b')) \cup (W(b) \cap N(b'))| \leq 2\delta \cdot n$.
Therefore, given $b$ and $b'$, there are at most $2\delta \cdot n$ corresponding conflicts $(a, b, b') \in C$.
Thus, $|C| \leq K \cdot \left(\frac{n}{K}\right)^2 \cdot 2\delta \cdot n$.

Putting together the previous two bounds, we get $|X| \leq \frac{\eta}{2} \cdot K^2$.
Therefore, the number of unordered pairs $(V_i, V_j)$ which are not $\eta$-homogeneous is at most $\eta K^2$, which is an $\eta$-fraction of all pairs $(V_i, V_j)$.
\end{proof}

\section{Chromatic Threshold of \texorpdfstring{$K_t$}{Kt}-Free Graphs}\label{sec:TFG}

The main reason why DNL is particularly suited to triangle-free graphs is that adjacent vertices are pairwise far apart. Indeed whenever $uv$ is an edge, the corresponding vectors $V_u$ and $V_v$ of the adjacency matrix of $G$ are at Hamming distance $\deg(u)+\deg(v)$. From a different angle, in the neighborhood set system $\mF$, if $uv$ is an edge then $\mF_{uv} = \emptyset$. The main difference when using DNL, compared with the existing proofs of chromatic thresholds, is that we do not need to assume that our graphs are maximal triangle-free (i.e. that every pair of vertices is at distance at most 2). Indeed, the trigraph point of view offers the opportunity to consider as red edges the non-edges which are close to being edges (for instance which could be replaced by an edge without creating a triangle). 

In general, DNL does not provide better bounds than other methods based on variants of VC-dimension, and we merely reprove already known results. However, it can yield strikingly short and simple proofs, as illustrated in the introduction. Furthermore, in the case of $K_t$-free dense \emph{regular} graphs, assuming edge maximality is not possible, as the edge-closure would violate the regularity condition.  Our main contribution here is to give a full characterization of chromatic thresholds of $K_t$-free graphs in the regular case.

All the results discussed here produce very simple randomized algorithms. Our proofs only rely on uniform sampling, sometimes with a round of classification. The main interest of this section is to offer a (hopefully) simple explanation for why uniform sampling works so well in dense triangle-free graphs.

\subsection{A spherical viewpoint and Borsuk-Hajnal extremal examples}

We revisit the proof given in the introduction stating that a triangle-free graph with minimum degree at least $(1/3+\varepsilon)n$ has bounded chromatic number. The question was first raised in 1973 by Erd\H{o}s and Simonovits~\cite{Erd73}, who conjectured that minimum degree $n/3$ would imply 3-colorable. This was disproved almost 10 years later by Häggkvist in \cite{Hag82}, with a 10-regular, 4-chromatic graph on 29 vertices, based on the Grötzsch graph. In 1997, Chen, Jin and Koh~\cite{Che97} showed that minimum degree greater than $10n/29$ implies 3-colorable. In 2004, Brandt and Thomassé \cite{BT2004} finally proved that 4 is the correct bound. See Łuczak, Polcyn and Reiher \cite{Luc24} for the most recent results on this subject. Our first use of DNL for the $(1/3+\varepsilon)n$ bound was inspired by vector coloring, see~\cite{Kar98} for a general introduction. We used the following antipodal spherical version of DNL.

\DNLSantipodal*

The idea proposed in~\cite{CT07} is to associate to every vertex $v_i$ a unit vector in such a way that adjacent vertices form an angle of more than $2\pi/3$. The remarkable fact is that this representation certifies by itself that $G$ is triangle-free. 

Specifically, when the vertex set of $G$ is $\{v_1,\dots ,v_n\}$, we select for each $v_i$ a set $Z_i$ of $(1/3+\varepsilon)n$ of its neighbors. Associate to $v_i$ the $n$-vector $V_i$ whose $j$-th coordinate is $a\coloneqq2/3-\varepsilon$ if $v_j\in Z_i$, and is $a-1=-1/3-\varepsilon$ otherwise. Note that $\langle V_i, V_i\rangle=a^2(1-a)n+a(a-1)^2n=a(1-a)n$ for all $v_i$, and that $\langle V_i, V_j\rangle=2(1-a)(a-1)an+(1-2(1-a))(a-1)^2n=-(1-a)^2n$ for every edge $v_iv_j$. When $v_iv_j$ is an edge, the angle $V_i0V_j$ is $\arccos\left((a-1)/a\right)$, where $(a-1)/a=-(1+3\varepsilon)/(2-3\varepsilon)<-1/2$. Therefore edges form angles greater than $2\pi/3+c$ for some $c>0$. 
It was asked in~\cite{CT07} if this representation could be used to obtain a bound on the chromatic number. This is indeed the case, and this was our main motivation for the introduction of DNL. Observe that the proof in the spherical setting does not achieve a $O(1/\varepsilon)$ bound but rather a $O\left(1/\varepsilon ^2\right)$ one. This is because the VC-dimension is quadratic in $1/\varepsilon$, and explains why we use set systems rather than metric spaces.

\thomassen*

\begin{proof} To get a $O\left(\frac{1}{\varepsilon^2}\right)$ bound, 
associate to every vertex $v_i$ the $n$-vector $V_i$ as above. Since $G$ has linear degree,  \cref{lem:DNL-spherical-antipodal} provides a bounded size subset $X$ of vectors such that every $V_i$ forms an angle at least $2\pi/3$ with some vector $V_j\in X$. As the set of vectors forming an angle at least $2\pi/3$ with some fixed vector $V_j$ is a stable set, $G$ has chromatic number at most $|X|$.
\end{proof}

This spherical point of view is particularly attractive since the extremal examples with degree less than $n/3$ can also be expressed as points on the sphere, and the large chromatic number follows from the Borsuk-Ulam Theorem. Here is the construction of \emph{Borsuk-Hajnal graphs} $BH_\varepsilon$ with arbitrarily large chromatic number and minimum degree $(1/3-\varepsilon)n$. Consider the $d$-dimensional unit sphere $\mathbb{S}^d$ and pick two finite well-distributed sets $X$ and $Y$ in $\mathbb{S}^d$ with $|X|<\!\!<|Y|$. Consider also an (abstract) set $Z$ of size $|Y|/2$, unrelated to $\mathbb{S}^d$. 
The vertex set of $BH_\varepsilon$ is $X\cup Y\cup Z$. 
Add all edges between $Y$ and $Z$, all edges inside $X$ between vertices at spherical distance at least $\pi-\varepsilon$ in $\mathbb{S}^d$, and all edges between vertices of $X$ and $Y$ at spherical distance at most $\pi/2-\varepsilon$. The resulting graph is triangle-free, its minimum degree is arbitrarily close to $|Y|/2$ (hence can be fixed to a $(1/3-\varepsilon)$-fraction of the number of vertices) and its chromatic number is $d+2$ by applying the Borsuk-Ulam Theorem to the subgraph induced on $X$. 

Adding balanced complete multipartite graphs completely joined to $BH_\varepsilon$ results in $K_t$-free graphs with minimum degree arbitrary close to $\frac{2t-5}{2t-3}n$. These bounds are sharp and discussed in \cref{subsec:homthreshold}. In \cref{subsec:regthreshold}, we present analogous constructions, proposed by O'Rourke, for regular $K_t$-free graphs with degree almost $\frac{3t-8}{3t-5}n$.

\subsection{The vicinity of the \texorpdfstring{$n/3$}{n/3} chromatic threshold}\label{subsec:digging}

The $n/3$ degree threshold is an intriguing value with respect to triangles in graphs. Notably, it also appears as a critical value for the Triangle Removal Lemma. More precisely, the dependence between the parameters is linear above that threshold, and super-polynomial below that threshold, see \cite{https://doi.org/10.1002/jgt.22891}.  
We explore in this section the chromatic number of triangle-free graphs whose minimum degree is slightly less than $n/3$. In particular, we show that minimum degree $n/3-n^{1-\varepsilon}$ suffices to bound the chromatic number of a triangle-free graph. This is essentially best possible.

It was already known that $n/3$ is not only a soft threshold but also a hard one: every triangle-free graph with minimum degree more than $n/3$ has bounded chromatic number. The first proof was given in~\cite{BT2004} where the full structure of these graphs was provided, showing that their chromatic number is at most 4. A much simpler proof (giving the bound 1665 instead of 4) was given in~\cite{Luczak2010coloringdensegraphsvcdimension}. 
Both proofs critically use that every \emph{maximal} triangle-free graph with minimum degree more than $n/3$ does not induce a cube (this was proved by Brandt in~\cite{BRANDT200233}). 
Therefore, they have bounded VC-dimension, hence bounded domination since the degree is linear, and finally bounded chromatic number (this gives the 1665 bound of~\cite{Luczak2010coloringdensegraphsvcdimension}).
Brandt's cube-free Lemma is a very nice and finely crafted graph theoretical argument. 
It asserts that if a maximal triangle-free graph with minimum degree more than $n/3$ contains an induced cube, then it must contain the graph in \cref{fig:Brandt's-graph-2}, which is a contradiction since Brandt's graph has no independent set of size 5 and by double counting some vertex must be adjacent to more than 12/3 of its vertices.

\begin{figure}[ht]
    \centering
\begin{tikzpicture}[scale=2, every node/.style={draw, circle, minimum size=8mm, inner sep=1pt},
    bend angle=20]

\node (u1) at (0,2) {$u_1$};
\node (u2) at (1.5,2) {$u_2$};
\node (u3) at (3,2) {$u_3$};
\node (u4) at (4.5,2) {$u_4$};

\node (x) at (0,1) {$x$};
\node (v2) at (1.5,1) {$v_2$};
\node (v3) at (3,1) {$v_3$};
\node (y) at (4.5,1) {$y$};

\node (w1) at (0,0) {$w_1$};
\node (w2) at (1.5,0) {$w_2$};
\node (w3) at (3,0) {$w_3$};
\node (w4) at (4.5,0) {$w_4$};

\draw (u1) -- (w4);
\draw (u2) -- (w3);
\draw (u3) -- (w2);
\draw (u4) -- (w1);

\draw (u1) to (w2);
\draw[bend right=20] (u1) to (w3);
\draw (u2) to (w1);
\draw[bend left=20] (u2) to (w4);
\draw[bend right=20] (u3) to (w1);
\draw (u3) to (w4);
\draw[bend left=20] (u4) to (w2);
\draw (u4) to (w3);

\draw (v2) -- (u2);
\draw (v2) -- (w2);
\draw (v2) -- (x);
\draw[bend right=20] (v2) to (y); 

\draw (v3) -- (u3);
\draw (v3) -- (w3);
\draw (v3) -- (y);
\draw[bend right=20] (v3) to (x); 

\draw (x) -- (u1);
\draw (x) -- (w1);

\draw (y) -- (u4);
\draw (y) -- (w4);

\end{tikzpicture}
    \caption{Brandt's graph.}
    \label{fig:Brandt's-graph-2}
\end{figure}
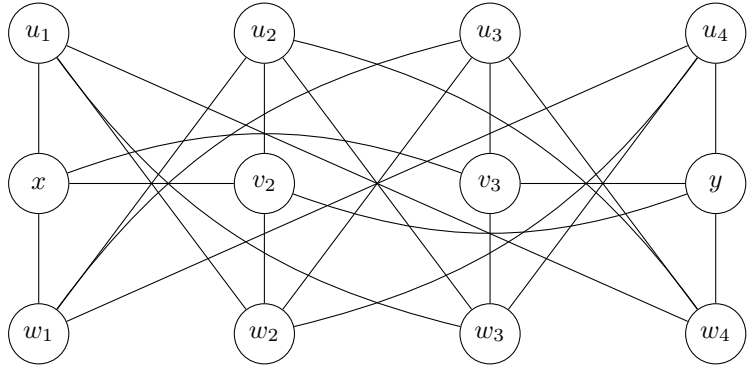

Brandt's graph is very useful to characterize the full structure, but is too razor-edged to dig very far below the threshold. 
The proof of 4-colorability needs the $n/3$ bound, and the VC-dimension proof allows to consider minimum degree $n/3$ minus a constant, to the price of increasing the 1665 bound. 
Our goal here is to provide a ``cube-free free'' proof of the $n/3$ threshold in a much smoother way, which is reminiscent of stability techniques in Regularity Lemma. 
In fact, the whole argument of this subsection could have been done using the Regularity Lemma, albeit with less ease of use and significantly worse bounds. 
The idea is to first find a large approximate structure, here a large almost-complete bipartite graph of size close to $(2n/3,n/3)$, and then discuss the structure of the remaining vertices.

This allows to shave off a function from $n/3$ rather than a constant. The function itself is a bit surprising and involves the \emph{odd girth} parameter (length of a shortest odd cycle).  
Let $og_c(n)$ be the minimum integer $k$ such that every $n$-vertex graph with odd girth at least $k$ has chromatic number at most $c$. 
We show that every triangle-free graph with minimum degree at least $n/3-n/8 og_c(n)$ has chromatic number bounded in terms of $c$. 
Answering a question of Erd\H{o}s~\cite{Erdos79}, Kierstead, Szemerédi and Trotter~\cite{DBLP:journals/combinatorica/KiersteadST84} showed that $og_c(n)=O\left(n^{\frac{1}{c-1}}\right)$. This implies that triangle-free graphs with minimum degree  $n/3-n^{1-\varepsilon}$ have bounded chromatic number, when $\varepsilon>0$ is fixed.

The following result explores via DNL the landscape of dense triangle-free graphs with minimum degree between $n/4$ and $n/3$. This is inspired by the argument of O'Rourke in the regular case. Here, $\alpha (G)$ denotes the maximum size of an independent set of $G$. The proof is identical to the one of \cref{thm:regulartrianglefree}.

\begin{theorem}\label{thm:alphachro}
Let $\varepsilon >0$ and $\delta \geq 1/4+\varepsilon$. If $G$ is a triangle-free graph with $\delta(G)=\delta n$ and such that  $\alpha (G)<(2\delta-\varepsilon)n$, then $G$ has bounded chromatic number. 
\end{theorem}

\begin{proof}
Let $P$ be a $\varepsilon/2$-clustering of $\mF=\{N(v):v\in V\}$ as in \cref{lem:DNLS-cluster-intro}. Assume for contradiction that some cluster contains $x$ and $y$ such that $xy$ is an edge. Note that if $z\in N(x)$ then $z\in D(x)$. Then, $X'=N(x)\setminus D_{\varepsilon/2}(y)$ has size at most $\varepsilon n/2$ by definition of $\varepsilon/2$-cluster.
Define $Y'$ analogously and note that $(N(x)\setminus X')\cup (N(y)\setminus Y')$ has size at least $2\delta n -2\varepsilon n/2>\alpha(G)$, so contains an edge $zt$.
As $N(x), N(y), N(z), N(t)$  all pairwise intersect on at most $\varepsilon n/2$ elements their union has size at least $4(n/4+\varepsilon n)-\binom{4}{2}\varepsilon n/2>n$, a contradiction. Hence each cluster is an independent set, thus $\chi(G)\leq |P|$. 
\end{proof}

We can now state the main result of this section, which holds for any small $\eta > 0$.

\begin{theorem}\label{thm:hardthresh2}
Every triangle-free graph $G$ with $\delta(G)\geq n/3-n^{1-\eta}$ has bounded chromatic number. 
\end{theorem} 

\begin{proof}
We show in this proof that $G$ has chromatic number bounded in terms of $c$ if
$\delta(G)\geq n/3-n/8og_c(n)$. Since $og_c(n)=O\left(n^{\frac{1}{c-1}}\right)$, we then choose $c$ such that $1/(c-1)<\eta$ and consider a large enough value of $n$ to conclude.

Denote $n/8og_c(n)$ by $g$. Pick some small $\varepsilon, \varepsilon' > 0$ with  $1/g <\!\!<\varepsilon <\!\!< \varepsilon'$.
By \cref{thm:alphachro}, we can conclude if $\alpha(G)<(2/3-\varepsilon)n$, so we can assume that $G$ contains a maximum independent set $I$ of size at least $(2/3-\varepsilon)n$. 
Let $A$ be the set of vertices with at least $(1-\varepsilon ')|I|$ neighbors in $I$, and $R=V\setminus (A\cup I)$. Note that $|I| \cdot (n/3-g) \leq e(I, V\setminus I) \leq |R| \cdot (1 - \varepsilon') |I| + (n - |R| - |I|) \cdot |I|$, hence $n/3-g\leq -|R| \varepsilon' + n - |I|$, and finally $|R| \leq (2n/3 - |I|+g)/\varepsilon'\leq \varepsilon n/\varepsilon' +g/\varepsilon '$. Thus $R$ is a very small set. Moreover, vertices in $I$ are adjacent to all but at most $\varepsilon n+g$ vertices outside of $I$, so they see most of $A$.

Assume for contradiction that a vertex $r$ of $R$ has a neighbor $x$ in $A$. Since $x$ is adjacent to at least $(1 - \varepsilon')|I|$ elements of $I$, the vertex $r$ is adjacent to at most $\varepsilon'|I|$ elements of $|I|$. 
Therefore $r$ is adjacent to at least $(n/3-g)-\varepsilon'|I|-\varepsilon n/\varepsilon'-g/\varepsilon '$ vertices outside of $I\cup R$, hence to most vertices in $A$. In particular, $r$ cannot be adjacent to a vertex of $I$ since it would then form a triangle with some vertex in $A$. This would contradict the maximality of $I$. 

Since there are no edges between $R$ and $A$, every vertex of $R$ is adjacent to at least nearly half of the vertices of $I$. Thus, there is no $C_5$ (cycle of length 5) in $G[R]$, as some vertex of $I$ would have three neighbors on this cycle, and form a triangle. 
This means that the second neighborhood (in $G[R]$) of every vertex of $R$ forms an independent set.
Let $S$ be the set of vertices in $R$ with at most $|I|/2-2g$ neighbors in $I$. 
Note that $G[R \setminus S]$ does not contain an odd cycle $C$ of length $2k+1 \leq og_c(n)$. Indeed, the total number of edges between $C$ and $I$ would be at least $(2k+1)|I|/2-og_c(n)\cdot 2g= k|I|+|I|/2-n/4>k|I|$, so some vertex of $I$ would have at least $k+1$ neighbors in $C$, thereby forming a triangle. 
So the odd girth of $G[R \setminus S]$ is at least $og_c(n)$, therefore $G[R \setminus S]$ has chromatic number at most $c$.

Now, every vertex $s \in S$ has at least $n/3 - g -|I|/2+2g$ neighbors in $R$ and since $|R| \leq (2n/3 - |I|+g)/\varepsilon '$, the size of the neighborhood of $s$ in $R$ divided by the size of $R$ is at least $\frac{n/3 - |I|/2+g}{2n/3 - |I|+g} \cdot \varepsilon '\geq \varepsilon '/2$. Hence $S$ consists of vertices with linear degree in $G[R]$ while $G[R]$ does not induce a $C_5$. 
Consider a maximum packing $P$ of vertices of $S$ at pairwise distance at least 3 in $G[R]$. 
Observe that $|P| \leq 2/\varepsilon'$. 
The union of the first and the second neighborhoods in $G[R]$ of the vertices of $P$ (which are independent sets for every $v\in P$) covers $S$ and has chromatic number at most $2|P| \leq 4/\varepsilon'$. 

Finally, G has bounded chromatic number if $\alpha(G)<(2/3-\varepsilon)n$, and when $\alpha(G)\geq (2/3-\varepsilon)n$, we have $\chi(G)\leq 2+c+4/\varepsilon'$ by bounding the chromatic number of $G[I],G[A],G[R\setminus S]$ and $G[S]$.
\end{proof}

This bound is essentially best possible as having minimum degree $n/3-n^{1-o(1)}$ is not sufficient to bound the chromatic number. 
Recall for this the construction of the Schrijver graph $S(\ell,k)$ (see~\cite{Sch78}): its vertices are the independent sets of size $\ell$ of the cycle $C$ of length $2\ell +k$ and its edges are the pairs $X,Y$ such that $X\cap Y=\emptyset$.
The chromatic number of $S(\ell,k)$ is $k+2$, identical to that of its supergraph, the Kneser graph (see~\cite{Lov78}). Assume that the cycle $C$ is on vertex set $[2\ell+k]$ with $\ell >\!\!> k$.
Select some large constant $K>\!\!>k$ and form the following \emph{Schrijver-Hajnal} graph $SH(\ell,k)$: add to $S(\ell,k)$ an independent set $A$ of size $(2\ell+k)\ell^K$ divided into subsets $A_1,\dots,A_{2\ell+k}$ of size $\ell ^K$. 
Connect every vertex $X$ of $S(\ell,k)$, seen as a subset of $[2\ell+k]$, to every vertex of $A_i$ for every $i\in X$. 
Finally, add an independent set $B$ of size $|A|/2$ completely connected to $A$.
By construction, $SH(\ell,k)$ is triangle-free, has chromatic number at least $k+2$ and has roughly $n\coloneqq3|A|/2$ vertices (we neglect the size $O((2\ell +k)^k)$ of the copy of $S(\ell,k)$ in this argument).
Let us compute the difference $d$ between $n/3$ and the minimum degree of $SH(\ell,k)$. 
The vertices of $B$ have degree close to $2n/3$ and the vertices of $A$ have degree more than $n/3$, so we just have to consider the vertices of $S(\ell,k)$. 
Then, $d=|A|/2-d_A(X)=|A| \cdot (1/2-\ell/(2\ell+k))= |A|\cdot k/(4\ell+2k)=k\ell^{K}/2$. 
Thus $d$ has order of magnitude $\ell^K$ while the order of magnitude of $n=3(2\ell+k)\ell^{K}/2$ is $\ell^{K+1}$. 
In particular, we can reach chromatic number $k$ with the minimum degree being $n^{1-\frac{1}{K+1}}$-close to $n/3$. 
This shows that shaving $n^{1-o(1)}$ from $n/3$ cannot guarantee bounded chromatic number.

\subsection{Homomorphism threshold of \texorpdfstring{$K_t$}{Kt}-free graphs}\label{subsec:homthreshold}

Answering a question of Thomassen, Łuczak~\cite{Luc06} showed that there exists a finite family ${\mathcal G}_{\varepsilon}$ of triangle-free graphs such that every triangle-free graph with minimum degree $(1/3+\varepsilon)n$ has a homomorphism to some graph in ${\mathcal G}_{\varepsilon}$. Brandt and Thomassé \cite{BT2004} completely described triangle-free graphs with minimum degree $(1/3 + \varepsilon)n$, and proved that they are all homomorphic to one of two graphs of size $O(1/\varepsilon)$, called Vega graph and Andr\'asfai graph. Goddard and Lyle~\cite{God11}, and independently Nikiforov \cite{nikiforov10} extended the result of Łuczak to arbitrary cliques. Oberkampf and Schacht~\cite{Obe20} gave a bound on the size of ${\mathcal G}_{\varepsilon}$ doubly exponential in $1/\varepsilon$, and showed that this can be achieved by a simple sampling, followed by two rounds of classification. This bound was recently reduced to singly exponential by Liu, Shangguan, Skokan and Xu \cite{LSSZ24}. Their ingenious proof relies on VC-dimension theory and classical tools from discrete geometry. 
Even more recently, Huang, Liu, Rong and Xu \cite{HLRX25} proved that for every $t \geq s \geq 3$, dense $K_s$-free graphs with bounded VC-dimension have a homomorphic $K_t$-free image of bounded size, where the density threshold depends on $s$ and $t$. Interestingly, for $t=s$ this is the usual homomorphism threshold, and when $t \to \infty$ this corresponds to the usual chromatic threshold (for $K_s$-free graphs of bounded VC-dimension).

The main goal of this section is to give an alternative, arguably more natural, proof of a bound singly exponential in $\text{poly}(1/\varepsilon)$, using DNL. The proof is direct: take a uniformly random sample of vertices, of size polynomial in $1/\varepsilon$, and classify the vertices according to the size of their common neighborhood with the sampled vertices. Then, the classes are independent sets and the graph obtained after contraction is $K_t$-free. 

To highlight the simplicity of the argument, we start by presenting it on triangle-free graphs.

\begin{theorem}\label{thm:oberkampf}
For every $\varepsilon >0$, there exists a family ${\mathcal G}_{\varepsilon}$ of triangle-free graphs, each of size exponential in $\text{poly}(1/\varepsilon)$, such that every $n$-vertex triangle-free graph with minimum degree $(1/3+\varepsilon)n$ has a homomorphism to some graph in ${\mathcal G}_{\varepsilon}$.
\end{theorem} 

\begin{proof}
Let $G = (V, E)$ be an $n$-vertex triangle-free graph with minimum degree $(1/3+\varepsilon)n$.
Consider the set system $\mF = \{N(v) : v \in V\}$.
By \cref{lem:DNLS-intersection}, there exists $X \subseteq V$ of size $\text{poly}(1/\varepsilon)$ such that for every $u, v \in V$ such that $|D(u) \cap D(v)| \geq \varepsilon|V|$, there exists $x \in X$ such that $x \in D_{\varepsilon}(u) \cap D_{\varepsilon}(v)$.
Let $\mathcal{P} = (P_1, \ldots, P_t)$ be the partition of $V$ such that $u, v \in V$ are in the same part if and only if for every $x \in X$, $x \in D_{\varepsilon}(u)$ if and only if $x \in D_{\varepsilon}(v)$.
Note first that $t \leq 2^{\text{poly}(1/\varepsilon)}$. Note also that by definition, if $x \in D_{\varepsilon}(u)$ then $|N(u) \cap N(x)| \leq \varepsilon n$.

Observe that for every $v \in V$ we have $N(v) \subseteq D(v)$ so $|D(v)| \geq \varepsilon n$. Thus, by definition of $X$ (by taking $u = v$), for every $v \in V$ there exists $x \in X \cap D_{\varepsilon}(v)$.
Now, consider any two vertices $u, v$ in the same part $P_i$ and suppose that $uv \in E(G)$. 
Consider $x \in X \cap D_{\varepsilon}(v)$. Since $u$ and $v$ are in the same part then $x \in D_{\varepsilon}(u)$.
Thus, $|N(u) \cup N(v) \cup N(x)| \geq 3 \cdot (1/3 + \varepsilon)n - 2\varepsilon n > n$, a contradiction. This proves that each $P_i$ is an independent set.

Thus, $G \to G/\mathcal{P}$ is a homomorphism. To conclude, it suffices to show that $G/\mathcal{P}$ is triangle-free. To do so, we argue that whenever there is an edge between $P_i$ and $P_j$, then $|N(v_i) \cap N(v_j)| < \varepsilon n$ for every $v_i \in P_i$ and $v_j \in P_j$. This implies that $G/\mathcal{P}$ is triangle-free, otherwise we would have three vertices $u, v, w$ such that $|N(u) \cup N(v) \cup N(w)| > 3 \cdot (1/3 + \varepsilon)n - 3\varepsilon n = n$, a contradiction.

Consider two parts $P_i, P_j$ such that there is an edge $p_ip_j \in E(G)$ between $P_i$ and $P_j$, and suppose by contradiction that there exist $v_i \in P_i$ and $v_j \in P_j$ such that $|N(v_i) \cap N(v_j)| \geq \varepsilon n$.
Then, $|D(v_i) \cap D(v_j)| \geq \varepsilon n$ so by definition of $X$, there exists $x \in X$ such that $x \in D_{\varepsilon}(v_i) \cap D_{\varepsilon}(v_j)$.
Since $v_i$ and $p_i$ are in the same part and $x \in X$ then $x \in D_{\varepsilon}(p_i)$, and similarly $x \in D_{\varepsilon}(p_j)$.
Then, $|N(p_i) \cup N(p_j) \cup N(x)| \geq 3 \cdot (1/3 + \varepsilon)n - 2\varepsilon n > n$, a contradiction.
\end{proof}

Observe that in the previous proof, a uniformly random sample of $\text{poly}(1/\varepsilon)$ vertices satisfies the same properties as $X$ with constant probability. This gives a simple polynomial-time randomized algorithm for constructing a homomorphism from $G$ to a graph in $\mathcal{G}_{\varepsilon}$.

We now state and prove the statement for abitrary cliques.

\begin{theorem}\label{thm:homomorphism-kt-free}
For every $\varepsilon>0$, every $K_t$-free $n$-vertex graph with minimum degree $\left(\frac{2t-5}{2t-3} + \varepsilon\right)n$ has a $K_t$-free homomorphic image of size $2^{\textup{poly}(1/\varepsilon)}$.
\end{theorem} 

The proof of this result is exactly the same as the proof of \cref{thm:oberkampf}, with a key subtlety, which is that in the classification step, instead of considering the neighborhoods we consider the \emph{$(t-2)$-clique neighborhoods}.
Formally, for every vertex $v$, we denote by $N_{k}(v)$ the set of cliques $K$ of size $k$ such that $\{v\}\cup K$ is a clique of size $k+1$. Observe that $N_1$ is the usual neighborhood of a vertex.
The key object is the \emph{clique incidence matrix} $A_{k}=(a_{K,v})$ of $G$, whose columns are indexed by the vertices $v$, and whose rows are indexed by the cliques $K$ of size $k$ of $G$, such that $a_{K,v}=1$ if $K\in N_{k}(v)$ and $a_{K,v}=0$ otherwise.
Note that $A_1$ is the adjacency matrix of $G$.
For every clique $K$ of size $k$, let $E(K) = \{v \in V : K \in N_k(v)\}$ be the set of vertices which extend $K$ to a clique of size $k+1$.
When $G$ is $K_t$-free, the matrix $A_{t-2}$ is the incidence matrix of the set system $\mF = \{E(K) : K \text{ is a clique of size }t-2\}$, which has the property that every set is an independent set, i.e. that $\mF_{uv} = \emptyset$ whenever $uv \in E(G)$.

To give some intuition, we say that two vertices $u, v$ are distant if $\mF_{uv}$ is small.
Once more, if $u$ and $v$ are distant, we think of them as being adjacent, and otherwise we think of them as being non-adjacent.
Using that $G$ is $K_t$-free with large minimum degree, we show that there cannot be $t$ pairwise distant vertices.
Using DNL, we then find a family $\mathcal{C}$ of cliques of size $t-2$ in $G$, of size $\textup{poly}(1/\varepsilon)$, such that whenever two vertices are distant to many cliques of size $t-2$, at least one of these cliques is in $\mathcal{C}$.
We then classify the vertices according to their distances to the cliques in $\mathcal{C}$.
Using the properties of $\mathcal{C}$ and the fact that there are no $t$ pairwise distant vertices, it follows that each partition class is an independent set and that the quotient graph is $K_t$-free.

Once again, such a family $\mC$ can actually be obtained with constant probability by sampling uniformly random vertices. Hence, this proof also gives a simple polynomial-time randomized algorithm for constructing the desired homomorphism.

We start with a Lemma stating that in a graph with large minimum degree, every large subset of vertices contains many large cliques. We will use multiple times in the course of the proofs of \cref{thm:homomorphism-kt-free,thm:regularkt}.

\begin{lemma}\label{lem:find-many-cliques}
If $G$ is an $n$-vertex graph with minimum degree $(1-c)n$ and $X$ is a subset of vertices of size at least $(sc + \varepsilon)n$, there are at least $\frac{\varepsilon \cdot c^s \cdot n^{s+1}}{(s+1)!}$ cliques of size $s+1$ in $G[X]$.
\end{lemma}

\begin{proof}
Take a random tuple $S$ of $s+1$ vertices of $G$ by sampling with repetition uniformly random vertices.
Let $E_0$ be the event that $S \subseteq X$ and induces a clique in $G$.
Write $S = (x_1, \ldots, x_{s+1})$.
Since $|X| \geq (sc+\varepsilon)n$ and since every vertex has at most $cn$ non-neighbors, for every $i \geq 1$ we have $\mathbb{P}[x_i \in X \cap N(x_1, \ldots, x_{i-1})] \geq sc + \varepsilon - (i-1)c = (s-i+1)c + \varepsilon$.
By the chain rule, $\mathbb{P}[E_0] \geq \prod_{i = 1}^{s+1} \left((s-i+1)c + \varepsilon\right) \geq c^s \cdot \varepsilon$.
Since every clique of size $s+1$ is created exactly $(s+1)!$ times in this process, the number of cliques of size $s+1$ in $G[X]$ is at least $\frac{\varepsilon \cdot c^s \cdot n^{s+1}}{(s+1)!}$.
\end{proof}

We now move to the proof of \cref{thm:homomorphism-kt-free}.

\begin{proof}[Proof of \cref{thm:homomorphism-kt-free}.]
Let $G$ be such a graph and set $\gamma = \frac{\left(2/2t-3\right)^{t-3}}{(t-2)!}$.
We can assume that $n \geq \frac{t}{(t-1) \varepsilon}$, otherwise $G$ has size $2^{\textup{poly}(1/\varepsilon)}$.
Let $\mK$ be the set of all cliques of size $t-2$ in $G$.

Consider the set system $\mF = \{E(K) : K \in \mK\}$.
Note that $|\mF| = |\mK| < n^{t-2}$ (assuming $t \geq 4$).
For $u, v \in V$ and $\varepsilon' > 0$, say that $u$ and $v$ are $\varepsilon'$-disjoint if $|\mF_{uv}| \leq \varepsilon' \cdot |\mF|$.
Note that if $uv \in E(G)$ then $u$ and $v$ are 0-disjoint. 

\begin{claim}
There cannot be $t$ vertices which are pairwise $\gamma\varepsilon$-disjoint.
\end{claim}

\begin{proof}
By contradiction, suppose that there are $t$ vertices $x_1, \ldots, x_t$ which are pairwise $\gamma \varepsilon$-disjoint.
Take such a set $X = \{x_1, \ldots, x_t\}$ which induces as many edges as possible.
Since $G$ is $K_t$-free, there is a non-edge in $G[X]$.
Furthermore, by maximality of $e(G[X])$, no vertex is adjacent to all vertices of $X$.

Let $Y$ be the set of vertices of $G$ which have $t-1$ neighbors in $X$, and let $|Y| = \beta \cdot n$.
On the one hand, $\sum_{x \in X}d(x) \geq \left(\frac{2t-5}{2t-3} + \varepsilon\right) \cdot n \cdot t$.
On the other hand, $\sum_{x \in X}d(x) = \sum_{v \in V}d_X(v) \leq (t-1) \cdot \beta \cdot n + (t-2) \cdot (1 - \beta) \cdot n$.
This implies $\beta \geq \frac{2t-6}{2t-3} + \varepsilon \cdot t$.
Since $n \geq \frac{t}{(t-1)\varepsilon}$ then $|Y - X| \geq |Y| - t \geq \left(\frac{2t-6}{2t-3} + \varepsilon\right)n$.

Since there exists a vertex outside of $X$ with degree $t-1$ on $X$, the maximality of $e(G[X])$ implies that every vertex of $X$ has degree $t-1$ or $t-2$ on $X$.
Furthermore, if a vertex outside of $X$ has degree $t-1$ on $X$, its non-neighbor cannot be a vertex of degree $t-2$, again by maximality of $e(G[X])$.
Consider $x \in X$ of degree $t-2$ in $X$ (which exists since there is a non-edge in $G[X]$), and let $x' \in X$ be the non-neighbor of $x$ in $X$.
Then, $x$ and $x'$ are adjacent to all the vertices in $Y-X$.
Since $|Y - X| \geq \left(\frac{2t-6}{2t-3} + \varepsilon\right)n$, \cref{lem:find-many-cliques} applied with $c = \frac{2}{2t-3}$ and $s = t-3$ implies that there are at least $\gamma \cdot \varepsilon \cdot n^{t-2}$ many cliques of size $t-2$ in $G[Y-X]$.
This means that $|\mF_{uv}| \geq \gamma \cdot \varepsilon \cdot n^{t-2} > \gamma \cdot \varepsilon \cdot |\mF|$ so $u$ and $v$ are not $\gamma \varepsilon$-disjoint, a contradiction.
\end{proof}

For $\varepsilon' > 0, u \in V$ and $C= \{v_1, \ldots, v_{t-2}\} \in \mK$, say that $u$ is $\varepsilon'$-disjoint of $C$ if $u$ is $\varepsilon'$-disjoint of every $v_i$.
We denote by $D^{\mK}(v)$ the set of all $C \in \mK$ which are $0$-disjoint of $v$ and by $D^{\mK}_{\varepsilon'}(v)$ the set of all $C \in \mK$ which are $\varepsilon'$-disjoint of $v$.
By \cref{lem:DNLS-intersection-higher-order} there exists $\mC \subseteq \mK$, a family of $(t-2)$-cliques of size $\text{poly}(1/\varepsilon)$ such that for every $u, v \in V$ such that $|D^{\mK}(u) \cap D^{\mK}(v)| \geq \gamma\varepsilon|\mK|$, there exists a clique $C \in \mC$ such that $C \in D^{\mK}_{\gamma\varepsilon}(u) \cap D^{\mK}_{\gamma\varepsilon}(v)$.  

Write $\mathcal{C} = \{C_1, \ldots, C_s\}$ with $s = |\mathcal{C}| = \textup{poly}(1/\varepsilon)$.
Consider the partition of $V$ into clusters, such that $u, v \in V$ belong to the same cluster if and only if for every $C \in \mC$ we have $C \in D^{\mK}_{\gamma\varepsilon}(u)$ if and only if $C \in D^{\mK}_{\gamma\varepsilon}(v)$. Observe that the number of clusters is $2^{\textup{poly}(1/\varepsilon)}$.

\begin{claim}\label{cl:cluster-all-far-or-all-close}
If $u, v \in V$ belong to the same cluster and $w \in V$ is $\gamma\varepsilon$-disjoint of $u$ then $w$ is $\gamma\varepsilon$-disjoint of $v$.
\end{claim}

\begin{proof}
Suppose by contradiction that $w \in V$ is $\gamma\varepsilon$-disjoint of $u$ and not $\gamma\varepsilon$-disjoint of $v$.
Then, $|\mF_{vw}| \geq \gamma \cdot \varepsilon \cdot |\mF| = \gamma \cdot \varepsilon \cdot |\mK|$.
This means that $v$ and $w$ extend at least $\gamma \cdot \varepsilon \cdot |\mK|$ common cliques. 
However, if $v$ and $w$ extend a clique $C \in \mK$ then $u$ and $w$ are 0-disjoint of $C$.
Thus, $|D^{\mK}(v) \cap D^{\mK}(w)| \geq \gamma \cdot \varepsilon \cdot |\mK|$ so there exists $C_i \in \mathcal{C}$ which is $\gamma\varepsilon$-disjoint of both of them.
Since $u$ and $v$ are in the same cluster, $u$ is also $\gamma\varepsilon$-disjoint of $C_i$, and $u, w, V(C_i)$ give $t$ points pairwise $\gamma\varepsilon$-disjoint, a contradiction.
\end{proof}

\begin{claim}
Every cluster is an independent set.
\end{claim}

\begin{proof}
By contradiction, suppose that $u, v$ belong to the same cluster and $uv \in E(G)$.
Then, $u$ and $v$ are $\gamma\varepsilon$-disjoint. By \cref{cl:cluster-all-far-or-all-close} applied with $w = v$, $v$ is $\gamma \varepsilon$-disjoint of $v$. However, by \cref{lem:find-many-cliques} applied with $c = \frac{2}{2t-3}$ and $s = t-3$, there are at least $\gamma \cdot \varepsilon \cdot n^{t-2} > \gamma \cdot \varepsilon \cdot |\mF|$ cliques in $G[N(v)]$, contradicting that $v$ is $\gamma\varepsilon$-disjoint of itself.
\end{proof}

The clusters define a partition $\mathcal{P}$ of $V(G)$ into independent sets, let $G/\mathcal{P}$ denote the corresponding quotient graph.

\begin{claim}
The graph $G/\mathcal{P}$ is $K_t$-free.
\end{claim}

\begin{proof}
First, \cref{cl:cluster-all-far-or-all-close} implies that if $P_i, P_j$ are two different parts of the partition (i.e. two different clusters) then either every vertex in $P_i$ is $\gamma \varepsilon$-disjoint of every vertex in $P_j$, or every vertex in $P_i$ is not $\gamma \varepsilon$-disjoint of every vertex in $P_j$.

If there exists a $K_t$ in $G/\mathcal{P}$, there are $t$ parts $P_1, \ldots, P_t$ with an edge between $P_i$ and $P_j$ for every $i \neq j$. Thus, for every $i \neq j$, some point of $P_i$ is $\gamma \varepsilon$-disjoint of some point of $P_j$, so all points of $P_i$ are $\gamma \varepsilon$-disjoint of all points of $P_j$. Therefore, there would be $t$ vertices pairwise $\gamma \varepsilon$-disjoint, a contradiction.
\end{proof}
\end{proof}

\subsection{Chromatic Threshold of Regular \texorpdfstring{$K_t$}{Kt}-free Graphs}\label{subsec:regthreshold}

The goal of this section is to extend \cref{thm:regulartrianglefree}, which shows that the chromatic threshold for regular triangle-free graphs is 1/4, to arbitrary cliques. 
This was investigated in the Master's Thesis of O'Rourke \cite{O'R14}, in which he generalized the constructions of $(1/4-\varepsilon)n$-regular triangle-free graphs $G_{\varepsilon}$ with arbitrarily large chromatic number to the $K_t$-free case. 
As a first step, adding an independent set of size $3n/4+\varepsilon n$ completely joined to $G_\varepsilon$ yields an $n$-regular $K_4$-free graph with $7n/4+\varepsilon n$ vertices. 
This construction shows that the regular chromatic threshold of $K_4$-free graphs is at least $4/7$.
O'Rourke investigated whether $4/7$ is indeed the threshold for $K_4$, and more generally whether $r_t \coloneqq \frac{3t-8}{3t-5}$ could be the regular chromatic threshold for $K_t$-free graphs.
These bounds are matched by constructions obtained from $G_{\varepsilon}$ by adding a balanced complete multipartite graph, see \cref{fig:Regular-Borsuk-Hajnal}. 

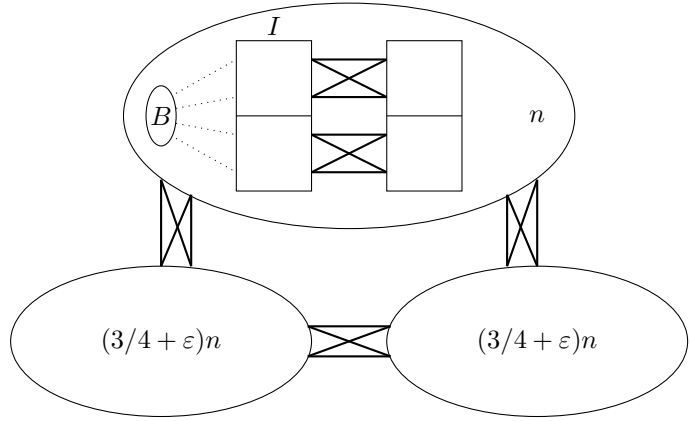
\begin{figure}[ht]
    \centering
    \begin{tikzpicture}[every node/.style={font=\sffamily}]
    \draw (1,0) ellipse (0.2 and 0.4);
    \node at (1,0) {$B$};
    \node at (6,0) {$n$};
    \draw (2,-1) rectangle (3,1);
    \node at (2.5, 1.2) {$I$};
    \draw (2,0) -- (3,0);
    \draw (4,-1) rectangle (5,1);
    \draw (4,0) -- (5,0);
    \draw[dotted] (1.2,0.3) -- (2,0.75);
    \draw[dotted] (1.2,0.1) -- (2,0.25);
    \draw[dotted] (1.2,-0.1) -- (2,-0.25);
    \draw[dotted] (1.2,-0.3) -- (2,-0.75);
    \draw[thick] (3,0.75) -- (4,0.75);
    \draw[thick] (3,0.25) -- (4,0.25);
    \draw[thick] (3,0.75) -- (4,0.25);
    \draw[thick] (3,0.25) -- (4,0.75);
    \draw[thick] (3,-0.25) -- (4,-0.25);
    \draw[thick] (3,-0.75) -- (4,-0.75);
    \draw[thick] (3,-0.25) -- (4,-0.75);
    \draw[thick] (3,-0.75) -- (4,-0.25);
    \draw (3.5,0) ellipse (3 and 1.5);
    \draw (1,-3) ellipse (2 and 1);
    \node at (1,-3) {$(3/4 + \varepsilon)n$};
    \draw (6,-3) ellipse (2 and 1);
    \node at (6,-3) {$(3/4 + \varepsilon)n$};
    \coordinate (E121) at (1,-0.85);
    \coordinate (E122) at (1.4,-1.05);
    \coordinate (E131) at (5.6,-1.05);
    \coordinate (E132) at (6,-0.85);
    \coordinate (E211) at (1,-2);
    \coordinate (E212) at (1.4,-2);
    \coordinate (E231) at (2.95,-2.8);
    \coordinate (E232) at (2.95,-3.2);
    \coordinate (E311) at (5.6,-2);
    \coordinate (E312) at (6,-2);
    \coordinate (E321) at (4.05,-2.8);
    \coordinate (E322) at (4.05,-3.2);
    \draw[thick] (E121) -- (E211);
    \draw[thick] (E121) -- (E212);
    \draw[thick] (E122) -- (E211);
    \draw[thick] (E122) -- (E212);
    \draw[thick] (E131) -- (E311);
    \draw[thick] (E131) -- (E312);
    \draw[thick] (E132) -- (E311);
    \draw[thick] (E132) -- (E312);
    \draw[thick] (E231) -- (E321);
    \draw[thick] (E231) -- (E322);
    \draw[thick] (E232) -- (E321);
    \draw[thick] (E232) -- (E322);
    \end{tikzpicture}

    \caption{O'Rourke's construction of a $(7/10 - \varepsilon)N$-regular $K_5$-free $N$-vertex graph with large chromatic number. The graph $B$ is a Borsuk graph, triangle-free with large chromatic number. The neighborhoods of the vertices of $I$ in $B$ give a $2 + \varepsilon'$-fractional coloring of $B$.}
    \label{fig:Regular-Borsuk-Hajnal}
\end{figure}

We confirm that O'Rourke's constructions are optimal:

\begin{theorem}\label{thm:regularkt}
For every $\varepsilon>0$, every $(r_t+\varepsilon)n$-regular $n$-vertex $K_t$-free graph has bounded chromatic number.
\end{theorem} 

Similarly to the proof of \cref{thm:homomorphism-kt-free}, the key object is the \emph{clique incidence matrix} $A_{t-2}$. 
We recall the relevant definitions.
For every vertex $v$, let $N_{k}(v)$ be the set of cliques $K$ of size $k$ such that $\{v\} \cup K$ is a clique of size $k+1$. 
For every clique $K$ of size $k$, let $E(K) = \{v \in V : K \in N_k(v)\}$ be the set of vertices which extend $K$ to a clique of size $k+1$.
An important observation is that if $G$ is $K_{t}$-free then the set system $\mF = \{E(K) : K \text{ is a clique of size }t-2\}$ has the property that every set is an independent set, i.e. that $\mF_{uv} = \emptyset$ whenever $uv$ is an edge of $G$. 

Given $\varepsilon'>0$ and a $K_t$-free graph $G$, we say that two vertices $u,v \in V(G)$ are \emph{$\varepsilon'$-disjoint} if $|\mF_{uv}| \leq \varepsilon'|\mF|$. Note that adjacent vertices are $0$-disjoint since they cannot extend the same $K_{t-2}$. As usual for DNL, almost-disjoint vertices will emulate the edges of the graph. 
A key tool in the proof of~\cref{thm:regulartrianglefree} was that $n$-vertex triangle-free $(1/4+\varepsilon)n$-regular graphs do not have four vertices which are pairwise $\varepsilon /2$-disjoint. This is indeed a general feature:

\begin{lemma}\label{lem:farvertices}
For every $\varepsilon>0$, there exists $\varepsilon '>0$ such that every $n$-vertex $(r_t+\varepsilon)n $-regular $K_t$-free graph does not contain $t+1$ vertices which are pairwise $\varepsilon'$-disjoint.
\end{lemma} 

\begin{proof}
Consider for contradiction a set $F=\{x_1,\ldots ,x_{t+1}\}$ of vertices which are pairwise $\varepsilon'$-disjoint (with $\varepsilon'$ to be defined later), and assume that $F$ spans as many edges as possible under this constraint. 
Since $G$ is $K_t$-free, some vertices in $F$ are non-adjacent.
If some vertex $u$ were adjacent to all vertices in $F$, we could exchange it with any vertex in $F$ incident to a non-edge in $G[F]$, and this would increase the number of edges of $G[F]$.
Hence there is no vertex adjacent to all $t+1$ vertices of $F$. 
Denote by $X$ the set of vertices with $t$ neighbors in $F$, by $Y$ the set of vertices with $t-1$ neighbors in $F$, and by $Z$ the other vertices.

On the one hand, $\sum_{u \in F}d(u) = (t+1)\cdot \left(\frac{3t-8}{3t-5}+\varepsilon\right) \cdot n = \left(t-1+\frac{3t-13}{3t-5}+(t+1)\cdot\varepsilon\right) \cdot n$. On the other hand, $\sum_{u \in F}d(u) = \sum_{v \in V}d_F(v) \leq t \cdot |X| + (t-1) \cdot (n - |X|)$, hence the size of $X$ is at least $\left(\frac{3t-13}{3t-5}+(t+1)\cdot\varepsilon\right) \cdot n=n-\frac{8n}{3t-5}+(t+1)\cdot\varepsilon \cdot n$. 
Note that this does not give anything for $t \leq 4$, but we will not need it.

Using the condition on the minimum degree and the large size of $X$, we will show that there are many cliques of size $t-2$ which are extended by two vertices in $F$, which will contradict that all vertices in $F$ are pairwise almost disjoint.

Assume first that $|Z|\leq \frac{n}{3t-5}$.
By \cref{lem:find-many-cliques}, applied with $c = 3/(3t-5)$ and $s = t-5$ there exists $\gamma_1 > 0$ such that $G[X]$ has at least $\gamma_1 \cdot n^{t-4}$ cliques of size $t-4$.
Consider any clique $K$ of size $t-4$ in $G[X]$.
Then, the common neighborhood of the vertices of $K$ has size at least $n - (t-4) \cdot \frac{3- \varepsilon}{3t-5} \cdot n \geq \frac{7}{3t-5} \cdot n$. Thus, since  $|Z| \leq \frac{n}{3t-5}$, they have at least $\frac{6}{3t-5} \cdot n$ common neighbors outside of $Z$.
By \cref{lem:find-many-cliques} applied with $c = \frac{3-\varepsilon}{3t-5}$ and $s = 2$, there exists $\gamma_2 > 0$ such that there are at least $\varepsilon \cdot \gamma_2 \cdot n^3$ cliques of size 3 inside this common neighborhood outside of $Z$.
Therefore, any clique $K$ of size $t-4$ in $X$ can be extended to a clique of size $t-1$ avoiding $Z$ in at least $\varepsilon \cdot \gamma_2 \cdot n^3$ different ways.
Overall, there exists a constant $\gamma_3$ such that there are at least $\varepsilon \cdot \gamma_3 \cdot n^{t-1}$ cliques of size $t-1$ which avoid $Z$ and have at least $t-4$ vertices in $X$.
Note that all this previous reasoning also applies for $t=4$.

Consider any such clique $K'$ of size $t-1$, and fix a subset $K$ of $K' \cap X$ of size $t-4$.
Consider $F' \subseteq F$ of size 5 which is complete to $K$. Note that each of the three vertices of $K'\setminus K$ has at least three neighbors in $F'$, hence there are two vertices $x, y$ in $F'$ which are adjacent to two vertices $u, v$ of $K'\setminus K$. In particular, both $x$ and $y$ extend the clique $K \cup \{u, v\}$ of size $t-2$, i.e. $E(K\cup\{u, v\}) \in \mF_{xy}$. 

Iterating this process over all such cliques $K'$ of size $t-1$, we obtain that there exist two vertices $x, y \in F$ for which there are at least $\frac{\varepsilon \cdot \gamma_3 \cdot n^{t-1}}{\binom{t+1}{2}}$ cliques $K'$ for which there exists a subclique $K''$ of size $t-2$ such that $E(K'') \in \mF_{xy}$. Note that two different cliques $K'$ of size $t-1$ can yield the same clique $K''$ of size $t-2$, but at most $n$ of them can yield it. Therefore, $|\mF_{xy}| \geq \frac{\varepsilon \cdot \gamma_3}{\binom{t+1}{2}} \cdot n^{t-2} > \frac{\varepsilon \cdot \gamma_3}{\binom{t+1}{2}} \cdot |\mF|$, which is a contradiction if $\varepsilon' \leq \frac{\varepsilon \cdot \gamma_3}{\binom{t+1}{2}}$.

Assume now that $\frac{n}{3t-5}<|Z|\leq \frac{4n}{3t-5}$. Since there are at least $\frac{n}{3t-5}$ vertices with degree at most $t-2$ on $F$, the lower bound on the size of $X$ is increased by at least $\frac{n}{3t-5}$. 
Therefore, $|X| \geq \left(\frac{3t-12}{3t-5} + \varepsilon\right)n$. 
As before, using \cref{lem:find-many-cliques} twice, we get that there exists a constant $\gamma_4$ such that there are at least $\varepsilon^2 \cdot \gamma_4 \cdot n^{t-2}$ cliques $K'$ of size $t-2$ with at least $t-3$ vertices in $X$ and the last vertex outside of $Z$. 
Now observe that there are two vertices $x, y$ in $F$ which are complete to $K'$. We conclude as previously.

If $|Z|> \frac{4n}{3t-5}$, again by \cref{lem:find-many-cliques}, there exists a constant $\gamma_5 > 0$ such that there are at least $\varepsilon \cdot \gamma_5 \cdot n^{t-2}$ cliques $K'$ of size $t-2$ in $X$. Again, there are two vertices $x, y$ in $F$ which are complete to $K'$ and we can conclude as previously. 
\end{proof}

Before proving \cref{thm:regularkt}, we prove a technical Lemma whose only positive side (other than being useful in the proof) is that it relies on a very nice fact about the minimum fractional vertex cover of a graph. 
Recall that a minimum \emph{vertex cover} of a graph $G$ is a minimum size subset of vertices which intersects all edges. 
Its fractional relaxation, a minimum \emph{fractional vertex cover}, is a non-negative weight function $\omega$ on the vertices of $G$ such that $\omega(x)+\omega(y)\geq 1$ for every edge $xy$ of $G$. 
Nemhauser and Trotter \cite{NT74} proved that the polytope of minimum fractional vertex covers is half-integral and moreover has the following property: 
Every graph has a minimum vertex cover $C$ and a minimum fractional vertex cover $\omega$ with values $0, 1/2$ and $1$ such that every vertex $v$ with $\omega(v)=1$ is in $C$, and every vertex $v$ with $\omega(v)=0$ is not in $C$.

\begin{lemma}\label{lem:vertexcoverKt}
Let $0\leq b\leq a\leq 1$ be some values. Let $G$ be a graph on $(2+2a+2\varepsilon)n$ vertices and $m$ edges with minimum vertex cover value $bn$. Then the minimum value $vc^*(G)$ of its fractional vertex cover satisfies $$D(G)\coloneqq(3-a) \cdot vc^*(G) \cdot n-m+3a^2n^2/2 + 3a\varepsilon n^2 \geq 0$$
\end{lemma} 

\begin{proof}
We prove the result by induction on $m$. This is clearly true if $m=0$.
Let $C$ be a minimum vertex cover of size $bn$ and $\omega$ be a minimum fractional vertex cover such that every vertex $v$ with $\omega(v)=1$ is in $C$, and every vertex $v$ with $\omega(v)=0$ is not in $C$. We denote by $C_1$ the vertices of $C$ with weight 1, and by $C_{1/2}$ the vertices of $C$ with weight 1/2.
If some vertex $v\notin C$ has weight 1/2, consider the graph $G'$ obtained by removing all its incident edges (no more than $bn$ of them as $v$ can only be adjacent to vertices of $C$, and at least one otherwise $v$ would have weight $0$). 
Note that $vc^*(G') \leq vc^*(G) - 1/2$ since we can simply set $\omega(v)=0$. 
Then, $D(G)\geq D(G')-bn+(3-a)n/2$. 
By induction, $G'$ satisfies $D(G')\geq 0$. 
Observe that $(3-a)/2-b=(3-a-2b)/2\geq 0$ since $a,b$ are at most 1, hence $D(G)\geq 0$.

Otherwise, all vertices outside of $C$ have weight 0. In particular, every vertex $v$ with $\omega(v)=1/2$ is only adjacent to vertices of $C$. In that case, we can apply the same argument as before. Finally, we can assume that $C=C_1$ and $\omega$ is a 0,1 function. In the worst case, all vertices of $C$ are adjacent to all vertices of $G$. This gives:
\belowdisplayskip=-12pt\begin{align*}
    D(G)&=(3-a)\cdot bn \cdot n-bn \cdot (2n+2an+2\varepsilon n-bn)-\binom{bn}{2}+3a^2n^2/2 + 3a\varepsilon n^2\\
    &\geq n^2(b-3ab-2\varepsilon b+b^2-b^2/2+3a^2/2 +3a\varepsilon) \\
    &\geq n^2(b-3ab+b^2/2+3a^2/2 + \varepsilon(3a-2b))\\
    &=n^2(3(a-b)^2/2+b-b^2 + \varepsilon(3a-2b)) \\
    &\geq 0
\end{align*}
\end{proof}

We now have all the tools for the proof of \cref{thm:regularkt}:

\begin{proof}[Proof of \cref{thm:regularkt}.]
To simplify the calculations, we slightly change the value of $\varepsilon$ and assume that $G$ is $K_t$-free and $\frac{3t-8+\varepsilon}{3t-5} \cdot n$-regular with $n$ vertices. 
We denote by $\varepsilon'$ the constant of \cref{lem:farvertices} and consider a very small  $\eta<\!\!<\varepsilon '$.

Let $\mK$ denote the set of all cliques of size $t-2$ in $G$. Consider the set system $\mF = \{E(K) : K \in \mK\}$. By \cref{lem:DNLS-cluster}, $\mF$ has a $(\varepsilon', \eta)$-clustering of size $2^{\textup{poly}(1/\varepsilon', 1/\eta)}$.

Let us now consider two vertices $u,v$ in the same cluster, and assume for contradiction that they are adjacent.
We denote by $I$ the intersection of their neighborhoods, by $U$ their union, by $D$ the set $U\setminus I$ and finally by $R$ the set $V\setminus U$. 
By definition of a $(\varepsilon', \eta)$-clustering, apart from a subset $Q$ of at most $2\eta n$ vertices, all vertices of $U$ are $\varepsilon'$-disjoint of both $u$ and $v$. 
Therefore by \cref{lem:farvertices}, there is no clique of size $t-1$ inside $U\setminus Q$. 
In the rest of the proof, we will disregard the set $Q$ in all computations (one can think of it as deleted from the graph). 
Since the contradiction involves values much larger than $\eta n$, this will not change the final conclusion, and avoids writing corrective terms $2\eta n$ in all (already tedious) future equations.
From now on, we simply assume that $U$ is $K_{t-1}$-free. Note that since $G$ is $K_{t}$-free, $I$ is also $K_{t-2}$-free. 

Assume for contradiction that $|I| \leq \frac{3t-10}{3t-5} \cdot n$, thus $|U| \geq |N(u)|+|N(v)|-\frac{3t-10}{3t-5} \cdot n=\frac{3t-6+2\varepsilon}{3t-5} \cdot n$. 
By \cref{lem:find-many-cliques} applied with $c = (3-\varepsilon)/(3t-5)$ and $s = t-2$, there is a clique of size $t-1$ in $G[U]$, a contradiction. 
Assume for contradiction that $|I| \geq \frac{3t-9}{3t-5} \cdot n$. In this case, by \cref{lem:find-many-cliques}, there is a clique of size $t-2$ in $G[I]$, again a contradiction.
We can then suppose that $|I| =  \frac{3t-9-a}{3t-5} \cdot n$ for some $0\leq a \leq 1$ (in fact $0<a<1$ as $\varepsilon$ offers some slack in the previous contradictions).
Therefore $|D| = 2 \cdot \frac{3t-8+\varepsilon}{3t-5} \cdot n -2|I| =\frac{2+2a+2\varepsilon}{3t-5} \cdot n$ and $|R| = \frac{2-a-2\varepsilon}{3t-5} \cdot n$.

Let us denote by $x$ the number of edges between $I$ and $D$, by $y$ the number of edges between $I$ and $R$ and by $z$ the number of edges between $D$ and $R$, see \cref{fig:setup-regular-threshold}.

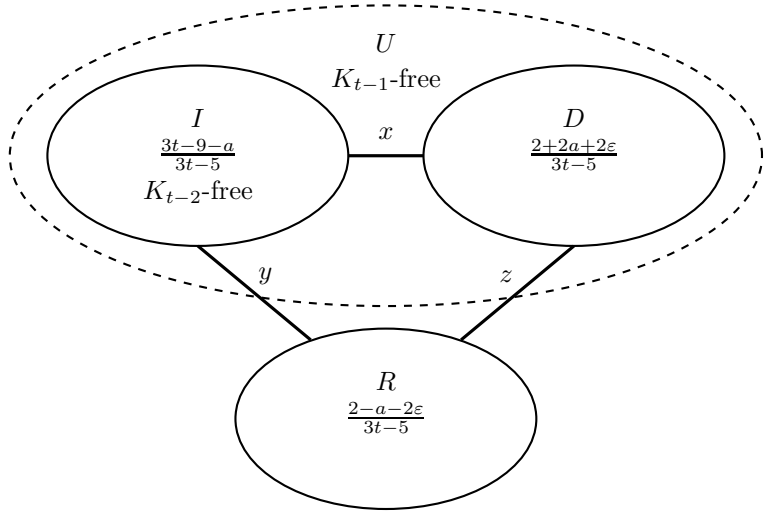
\begin{figure}[ht]
    \centering
    \begin{tikzpicture}

    \draw[thick, dashed] (0,0.5) ellipse [x radius=5cm, y radius=2cm];
    \node at (0,2) {$U$};
    \node at (0, 1.5) {$K_{t-1}$-free};
    
    \draw[thick] (-2.5,0.5) ellipse [x radius=2cm, y radius=1.2cm];
    \node at (-2.5,1) {$I$};
    \node at (-2.5, 0.5) {$\frac{3t - 9 - a}{3t-5}$};
    \node at (-2.5, 0) {$K_{t-2}$-free};
    
    \draw[thick] (2.5,0.5) ellipse [x radius=2cm, y radius=1.2cm];
    \node at (2.5,1) {$D$};
    \node at (2.5, 0.5) {$\frac{2 + 2a + 2\varepsilon}{3t-5}$};
    
    \draw[thick] (0,-3) ellipse [x radius=2cm, y radius=1.2cm];
    \node at (0,-2.5) {$R$};
    \node at (0, -3) {$\frac{2 - a - 2\varepsilon}{3t-5}$};
    
    \draw[very thick] (-0.5, 0.5) -- (0.5,0.5);
    \node at (0, .8) {$x$};
    \draw[very thick] (-2.5, -0.7)--(-1,-1.95);
    \node at (-1.6, -1.1) {$y$};
    \draw[very thick] (2.5, -0.7)--(1,-1.95);
    \node at (1.6, -1.1) {$z$};
    
    \end{tikzpicture}
    \caption{Setup for the end of the proof of \cref{thm:regularkt}.}
    \label{fig:setup-regular-threshold}
\end{figure}

The total number of edges leaving $R$ is simply upper bounded by the product of the size of $R$ and the degree of the vertices:

$$y+z\leq \frac{(2-a-2\varepsilon)n}{3t-5}\cdot \frac{(3t-8+\varepsilon)n}{3t-5}$$

Which we restate as:

\begin{equation}\label{eq:yz}
 (2-a-2\varepsilon)(3t-8+\varepsilon) \geq \frac{(3t-5)^2}{n^2} \cdot (y+z)
\end{equation}

By Tur\'an's Theorem, since $I$ is $K_{t-2}$-free, the maximum number of edges inside $I$ is achieved by a balanced $(t-3)$-partite graph. So the number of edges leaving $I$ is at least the sum of the degrees of the vertices in $I$ minus the total degree of a balanced $(t-3)$-partite graph inside $I$. In total this gives:

\begin{align*}    \frac{(3t-5)^2}{n^2} \cdot (x+y)&\geq (3t-9-a)(3t-8+\varepsilon)-\left(1-\frac{1}{t-3}\right)(3t-9-a)^2\\
&\geq (3t-9-a)\left(1+a+\varepsilon+\frac{3t-9-a}{t-3}\right)
\end{align*}

Which we restate as:
\begin{equation}\label{eq:xy}
\frac{(3t-5)^2}{n^2}(x+y)\geq (3t-9-a)\left(4+a\cdot\frac{t-4}{t-3}+\varepsilon\right)
\end{equation}

We now turn to a lower bound on the edges leaving $D$. Denoting by $\frac{en^2}{(3t-5)^2}$ the number of edges inside $D$, we then have:

\begin{equation}\label{eq:xz}
\frac{(3t-5)^2}{n^2}(x+z)= (2+2a+2\varepsilon)(3t-8+\varepsilon)-2e
\end{equation}

Adding the inequalities (\ref{eq:yz}), (\ref{eq:xy}) and (\ref{eq:xz}), we get:
\begin{align*}    \frac{(3t-5)^2}{n^2}2x&\geq
 (3t-9-a)\left(4+a\cdot\frac{t-4}{t-3}+\varepsilon\right)+(2+2a+2\varepsilon)(3t-8+\varepsilon)-2e-(2-a-2\varepsilon)(3t-8+\varepsilon) \\
 &\geq (3t-9-a)\left(4+a\cdot\frac{t-4}{t-3}+\varepsilon\right) + (3t-8+\varepsilon)(3a+4\varepsilon) -2e \\
&= 12t-36-4a+(3t-9-a) \cdot a \cdot \frac{t-4}{t-3}+3a(3t-8)-2e+\varepsilon(3t-9-a+12t-32+4\varepsilon + 3a) \\
&\geq 12t-36-28a+9at+(3t-9-a)\cdot a \cdot\frac{t-4}{t-3}-2e+\varepsilon(15t-41+2a) \\
&= 12t-36-40a+12at-a^2+\frac{a^2}{t-3}-2e+\varepsilon(15t-41+2a) \\
\end{align*}

Thus
\begin{equation}\label{eq:xmino}
\frac{(3t-5)^2}{n^2}x\geq 6t-18-20a+6at-a^2/2-e +\varepsilon(15t/2-41/2+a)
\end{equation}

We now turn to an upper bound on $x$, which is here simply expressed as the number of pairs of vertices $(i, d) \in I \times D$ minus the number of non-edges between these two sets, which we denote by $\frac{e'n^2}{(3t-5)^2}$. Hence:

$$\frac{(3t-9-a)n}{3t-5}\cdot \frac{(2+2a+2\varepsilon)n}{3t-5}-\frac{e'n^2}{(3t-5)^2}\geq x$$

Which we restate as:
$$
(3t-9-a)(2+2a+2\varepsilon)-e'\geq \frac{(3t-5)^2}{n^2}x
$$

Finally giving the equation:

\begin{equation}\label{eq:xmajo}
6t-18+6at-20a-2a^2-e'+\varepsilon(6t-18-2a)\geq \frac{(3t-5)^2}{n^2}x 
\end{equation}

We then sum the equations (\ref{eq:xmino}) and (\ref{eq:xmajo}) to get:
\begin{equation}\label{eq:finalcontra}
e'-e+3a^2/2+\varepsilon(3t/2-5/2+3a)\leq 0 
\end{equation} 

Since $t\geq 3$, we have $\varepsilon(3t/2 - 5/2) > 0$. Thus, to reach a contradiction, it suffices to show that the number of edges $\frac{en^2}{(3t-5)^2}$ inside $D$ cannot exceed the number of non-edges $\frac{e'n^2}{(3t-5)^2}$ between $D$ and $I$ plus $\frac{3a^2n^2}{2(3t-5)^2} + \frac{3a\varepsilon n^2}{(3t-5)^2}$. Let us analyse $e'$ in more details. 

Let $x$ be a vertex of $D$. Denote by $\omega(x)$ the number of non-neighbors of $x$ in $I$. If $xy \in E(G[D])$ is such that $\omega(x)+\omega(y)< \frac{(3-a)n}{3t-5}$, the vertices $x$ and $y$ have more than $\frac{3t-9-a-(3-a)}{3t-5} \cdot n= \frac{3t-12}{3t-5} \cdot n$ common neighbors in $I$. By \cref{lem:find-many-cliques}, this implies that $xy$ can be extended with a clique of size $t-3$ in $I$ to form a clique of size $t-1$ in $U$, which is impossible. 

So we can assume that $\omega(x)+\omega(y)\geq \frac{(3-a)n}{3t-5}$ for every edge $xy \in E(G[D])$. In other words, the function $\omega_f=\frac{(3t-5)\omega}{(3-a)n}$ is a fractional vertex cover of the induced subgraph $G[D]$. Denoting by $c_f$ the total sum of $\omega _f$, we get $\frac{e'n^2}{(3t-5)^2} = \frac{c_f(3-a)n}{3t-5}$.

Let us now consider a clique $K$ in $I$ of size $t-3$ (which exists by \cref{lem:find-many-cliques}). 
The intersection $I_K$ of the neighborhoods of the vertices of $K$ has size at least $n-\frac{3(t-3)}{3t-5} \cdot n= \frac{4n}{3t-5}$. 
Note that $I_K$ is disjoint from $I$ since $I$ is $K_{t-2}$-free. 
In particular $I_K$ contains all but at most $\frac{an}{3t-5}$ vertices of $D$ (since $|V \setminus I|  =|D \cup R| = \frac{4+a}{3t-5} \cdot n$). Observe that $D\cap I_K$ is an independent set since otherwise $K$ could be extended to a clique of size $t-1$ in $U$. 
Denoting by $\frac{bn}{3t-5}$ the size of a minimum vertex cover of $G[D]$, we get $\frac{bn}{3t-5}\leq\frac{an}{3t-5}$.  

Finally, $G[D]$ has $\frac{2+2a+2\varepsilon}{3t-5} \cdot n$ vertices, has $\frac{en^2}{(3t-5)^2}$ edges, fractional vertex cover $vc^*(D)$ and vertex cover size $\frac{bn}{3t-5}$ with $0 \leq b\leq a \leq 1$. So we can apply~\cref{lem:vertexcoverKt} to obtain that 

$$D(G[D])=(3-a) \cdot vc^*(D) \cdot \frac{n}{3t-5}-\frac{en^2}{(3t-5)^2}+\frac{3a^2n^2}{2(3t-5)^2}+ \frac{3a\varepsilon n^2}{(3t-5)^2}\geq 0$$

Since $\frac{3-a}{3t-5} \cdot vc^*(D) \leq \frac{3-a}{3t-5} \cdot c_f = \frac{e'n}{(3t-5)^2}$, we finally get $e'-e+3a^2/2+3a\varepsilon\geq 0$ which contradicts (\ref{eq:finalcontra}).
\end{proof}

\section{Domination versus fractional chromatic number}\label{sec:domtour}

We start by giving the proof of the central result on tri-tournaments.

\HWtournaments*

\begin{proof}
From a tournament $T$, form the tri-hypergraph $H_T$ on vertex set $V$ by adding for every $v\in V$ the hyperedge $(B(v),R(v),W(v))$ such that $B(v)=N_A^-[v]$ and $R(v)=N_R^-(v)$.
Observe that $H_T$ has VC-dimension $d$, and fractional transversal value at most 2 by \cref{thm:fisherryan}. By \cref{thm:bounded-integrality-gap}, $H_T$ has a transversal $X$ (thus intersecting $B(v)\cup R(v)$ for every $v \in V$) of size $O(d)$. In particular, $X$ is a dominating set of $T$.
\end{proof}

We now turn to the main result of this section, whose main point is to illustrate how easily the randomness transversal argument comes into play. Since the proof repeatedly uses a density increase argument, we do not try to get the best estimates in all computations. The proof gives a bound of the form $\gamma^+ \leq 2^{O\left(\chi^a_f \cdot \log \chi^a_f\right)}$.

\domfracchi*

\begin{proof}
We show that there exists a function $h$ such that every tournament $T=(V,A)$ satisfying $\chi_f^a(T)\leq x$ has a dominating set of size at most $h(x)$. We can set $h(x)= 1$ when $x< 3/2$ since every tournament with $\chi_f^a(T)< 3/2$ does not contain a circuit of length 3, and hence is transitive. 
We now assume that $h\left(\chi_f^a(T)-1/2\right)$ exists, and we show that we can bound $\gamma ^+(T)$ in terms of $h\left(\chi_f^a(T)-1/2\right)$. We write $c\coloneqq1/\chi_f^a(T)$, where $c$ is a rational number since it is a solution of a linear program with integer coefficients.

Let ${\mathcal F}$ be a family of transitive sub-tournaments $T_1,\dots ,T_t$ such that every vertex belongs to $ct$ of them. 
Let $s = \frac{4}{c^2}$.
If $T$ has no shattered set of size $s$ then its VC-dimension is bounded and therefore its domination as well.

If $T$ contains shattered sets of size $s$, our strategy is then to find an arc transversal of all such shatters, in order to reduce to a tri-tournament $T'$ with bounded VC-dimension. 
To do so, consider a shattered set $X$ of size $s$, and let $S$ be a set of $2^{|X|}$ vertices with all possible adjacencies on $X$.
Observe that in a $T_i$ which contains $k$ vertices of $X$, there are at most $k+1$ possible adjacency types on $X$, so $T_i$ contains at most $\frac{k+1}{2^k} \cdot 2^s$ vertices of $S$.
Therefore, $T_i$ contains at most $\frac{k(k+1)}{2^k} \cdot 2^s \leq 2^{s+1}$ arcs between $X$ and $S$.

Hence there is an arc of $T$ (in any direction) between $X$ and $S$ which is contained in at most $\frac{2^{s+1}}{s2^s} \cdot t = \frac{2}{s} \cdot t = c^2t/2$ of the $T_i$'s. 

To form the tri-tournament $T'=(V,A,R)$, add to $R$ all the arcs $yx$ such that $xy$ is an arc of $T$ contained in at most $c^2t/2$ of the $T_i$'s. By construction, $T'$ has VC-dimension at most $s$.

Therefore, by \cref{thm:HWtournaments}, $T'$ has a bounded size dominating set $X$. Thus for every vertex $y \notin X$, there exists $x\in X$ such that $xy$ is an arc of $T$, or $xy$ is a red arc of $T'$.
Hence, the only vertices which are left to dominate are in the red out-neighborhood of $X$. To conclude, we just have to dominate (in $T$) each set $R^+(x)$, for $x\in X$. Let ${\mathcal F}_x$ be the subfamily of all $T_i$'s which contain $x$. By construction of $R$, given $x\in X$, every $y\in R^+(x)$ appears in at most $c^2t/2$ elements of ${\mathcal F}_x$.  In particular, ${\mathcal F}_y \setminus {\mathcal F}_x$
has size at least $(c-c^2/2)t$. Thus $y$ appears at least in a  $\frac{c-c^2/2}{1-c}$-fraction of ${\mathcal F}\setminus {\mathcal F}_x$. Since $\frac{1-c}{c-c^2/2}<\frac{1}{c}-\frac{1}{2}$, the family  ${\mathcal F}\setminus {\mathcal F}_x$ is a $\frac{1}{c}-\frac{1}{2}$-fractional acyclic coloring of $R^+(x)$.

Finally, we obtain that $h\left(\chi_f^a(T)\right)\leq |X|+|X| \cdot h\left(\chi_f^a(T)-1/2\right)$.
\end{proof}

An obvious open problem is to investigate whether the function in \cref{thm:domfracchi} can be reduced. Is it true that for tournaments the domination is polynomially bounded in terms of the fractional acyclic chromatic number?

This result has a direct corollary in the field of ``local to global'' properties. We can rephrase our discussion on chromatic thresholds of triangle-free graphs as: Under some minimum degree condition, if the neighborhood of every vertex is an independent set (hence has chromatic number 1), then the whole graph has bounded chromatic number.
The direct generalization of this result to tournaments (dropping any requirement on degree, since tournaments are already dense) was conjectured in \cite{BCCFLSST13}: If a tournament $T$ satisfies that for every vertex $v$, the out-neighborhood $N^+(v)$ satisfies $\chi^a \left(T[N^+(v)]\right)\leq k$ (say that $T$ is \emph{locally $k$-bounded}), then  $\chi^a(T)$ is bounded by a function of $k$. This was positively answered in \cite{HLTW19}. \cref{thm:domfracchi} sheds some light on this question, and shows that the fractional acyclic chromatic number is a versatile tool for tournaments:

\begin{theorem}\label{thm:locfracchi}
Every locally $k$-bounded tournament $T$ satisfies $\chi_f^a(T)\leq 2k$. 
\end{theorem}

\begin{proof}
Let $p$ be the winning strategy of $T$. For every vertex $v$, assign weight $2p(v)$ to each of the transitive tournaments $T_v^1,\dots,T_v^k$ which partition  $\{v\}\cup N^+(v)$. The total weight is $2k$, and since $p(v\cup N^-(v))\geq 1/2$, every vertex is covered with weight at least 1.
\end{proof}

\begin{corollary}
Every locally $k$-bounded tournament has bounded acyclic chromatic number.
\end{corollary}

\begin{proof}
Apply \cref{thm:locfracchi}, and then \cref{thm:domfracchi} to find a bounded size dominating set $X$. Since every out-neighborhood has acyclic chromatic number $k$, we have $\chi^a(T)\leq k|X|$. 
\end{proof}

The following was proposed to generalize locally $k$-bounded tournaments.
Given an arc $uv$ of $T$, we denote by $DT(uv)$ the set of vertices $w$ forming a directed triangle $uvw$. We say that $T$ is \emph{arc-locally $k$-bounded} if $\chi^a(DT(uv))\leq k$  for all arcs $uv$. 
It was independently proved by Klingelh{\"{o}}fer and Newman \cite{Kli24}, and by Nguyen, Scott and Seymour \cite{NSS24} that arc-locally $k$-bounded tournaments have bounded chromatic number. 
The proof of Klingelh{\"{o}}fer and Newman relies on ideas from \cite{KN23} to conclude in the case where the domination is bounded.
We show how to combine these ideas with \cref{thm:domfracchi} to simplify the proof. The following proof was suggested to us by Alantha Newman.

\begin{theorem}\label{thm:k-arc-bdd}
Every arc-locally $k$-bounded tournament has bounded acyclic chromatic number.
\end{theorem}

We will use the following result from \cite{Kli24}.

\begin{lemma}\label{lem:chi-KN23}
If $T$ is an arc-locally $k$-bounded tournament then $\chi^a(T) \leq 5k \cdot \gamma^-(T) \cdot \gamma^+(T)$.
\end{lemma}

To bound $\gamma^+(T)$ and $\gamma^-(T)$, by \cref{thm:domfracchi}, it suffices to prove that arc-locally $k$-bounded tournaments have bounded fractional acyclic chromatic number.

\begin{lemma}\label{lem:k-arc:bdd-frac}
Every arc-locally $k$-bounded tournament $T$ satisfies $\chi^a_f(T)\leq 20k$.
\end{lemma}

\begin{proof}
Since the fractional acyclic chromatic number of a tournament is the maximum of that of its strongly connected components, we can assume that $T$ is strongly connected. Consider a winning strategy $p$ (so that $2p$ is fractionally dominating) and a losing strategy $q$ (so that $2q$ is fractionally absorbing) for $T$. 
For each pair of (non-necessarily distinct) vertices $(s,t)$ of $T$, Lemma 2.4 from \cite{Kli24} implies that $N^-[s] \cap N^+[t]$ can be colored with at most $5k$ colors (see also Lemma 2.6 in \cite{KN23}).
Assign fractional weight $4q(s)p(t)$ to each of the corresponding transitive subtournaments.
Notice that some transitive subtournaments might be encountered more than once; in this case, their total weight is the sum of the weights they receive.

We claim that each vertex is covered by at least 1 with respect to the total weight function $w$.  For a vertex $v$, the weight of pairs $(s,t)$ for which $v \in N^-[s] \cap N^+[t]$ is $$\sum_{\substack{s \in N^+[v]\\t \in N^-[v]}} 4q(s)p(t) = \sum_{s \in   N^+[v]} 2q(s) \sum_{t \in N^-[v]} 2p(t) \geq 1.$$ It remains to compute the total weight over all transitive subtournaments encountered.  We have $$\chi^a_f(T) \leq \sum_{s \in V} \sum_{t \in V} 2q(s) 2p(t) \cdot 5k \leq 20k.\vspace{-10mm}$$
\end{proof}

\begin{proof}[\textbf{Proof of \cref{thm:k-arc-bdd}}.]
If $T$ is arc-locally $k$-bounded then \cref{lem:k-arc:bdd-frac} implies that $\chi^a_f(T)$ is bounded, and \cref{thm:domfracchi} implies in turn that $\gamma^+(T)$ is bounded.
Let $T'$ be the tournament obtained by reversing the orientation of all arcs of $T$. Then, $\chi^a_f(T') = \chi^a_f(T)$ and any dominating set of $T'$ is an absorbing set of $T$. Therefore, applying \cref{thm:domfracchi} to $T'$ gives that $\gamma^-(T)$ is bounded.
We conclude using \cref{lem:chi-KN23}.
\end{proof}

\section{Dominating majority digraphs}\label{sec:majo}

Our goal in this section is to show that every $(1/2-\varepsilon)$-majority digraph has bounded domination number. Let us first introduce some definitions.
Let $\leq_1, \ldots, \leq_m$ be total orders on the same ground set $V$.
For $c \in [0, 1]$, the \emph{$c$-majority digraph} $D_c$ of $(\leq_1, \ldots, \leq_m)$ is the digraph on vertex set $V$ where $xy$ is an arc if $|\{i \in [m] : x <_i y\}| \geq c \cdot m$. In other words, if $xy$ is an arc of $D_c$, the vertex $x$ comes before $y$ in at least a $c$-fraction of the orders $\leq_i$. Given $\varepsilon>0$, observe that the $(c-\varepsilon)$-majority directed graph  $D_{c-\varepsilon}$ on vertex set $V$ based on the same orders $\leq_1, \ldots, \leq_m$ is a supergraph of $D_c$. For more about majority digraphs, see~\cite{Alon02}.

The following result was proved independently by Charikar, Ramakrishnan and Wang \cite{CRW25}.

\majosloppy*

To prove \cref{thm:majosloppy}, we need some preliminary results. Given a directed graph $D = (V, A)$, an arc $uv$ is \emph{consistent} with a linear order $\preceq$ on $V$ if $u \prec v$.

For a directed graph $D = (V, A)$ and a linear order $\preceq$ on $V$, we denote by $h(D, \preceq)$ the number of arcs of $D$ which are consistent with $\prec$, and by $h(D)$ the maximum over all linear orders $\preceq$ on $V$ of $h(D, \preceq)$.
For an integer $n$, we denote by $f(n)$ the minimum value over all $n$-vertex tournaments $T$ of $h(T)$.
Erd\H{o}s and Moon \cite{EM65} asked about the asymptotics of the function $f(n)$.
The following result, giving an upper bound on $f(n)$, is due to De La Vega \cite{DLV83}, improving on a result of Spencer \cite{Spencer80}.

\begin{lemma}\label{lem:delavega}
If $T_n$ is a uniformly random $n$-vertex tournament, $\mathbb{P}\left[h(T_n) \leq \frac{1}{2} \cdot \binom{n}{2} + 1.73n^{3/2}\right] \to 1$ when $n \to \infty$.
\end{lemma}

\begin{proof}[\textbf{Proof of \cref{thm:majosloppy}}.]
Let $\leq_1, \ldots, \leq_m$ be a set of total orders on a vertex set $V$ of size $n$.
Let $D_{1/2}$ and $D_{1/2-\varepsilon}$ be the corresponding majority digraphs.
Note that when $m$ is even, $D_{1/2}$ is a semi-complete digraph (i.e. a tournament with possibly some additional arcs forming circuits of length 2 in case of a draw in the vote), but we will think of it as a tournament $T=(V,A)$ for simplicity (the argument is unchanged). We form the tri-tournament $T'=(V,A,R)$ where $R$ consists of all the arcs of $D_{1/2-\varepsilon}$ which are not arcs of $T$. By \cref{thm:HWtournaments}, we just have to show that the VC-dimension of $T'$ is at most $O\left(\frac{1}{\varepsilon^2}\right)$.

To do so, consider a shattered set $X \subseteq V$ of $T'$ of size $k$. There exists a set $S$ of $2^k$ vertices such that there is no red arc between $X$ and $S$ in $T'$ (in any direction), and such that for every $Y \subseteq X$, there exists a unique  $y \in S$ such that $Y = N_T^-[y] \cap X$ (the closed in-neighborhood of $y$ in $X$ in $T$).
Consider a multiset $Z$ of size $|X|$ obtained by picking uniformly random vertices in $S$, with repetition. Then, the edges between $X$ and $Z$ are oriented independently uniformly at random.
Consider now the tournament $R$ obtained by adding an arc between every pair of vertices inside of $X$, and between every pair of vertices inside of $Z$, each oriented uniformly at random independently of all others.
Then, $R$ is a uniformly random tournament on $2k$ vertices.

For a linear order $\preceq$ on $X \cup Z$, let $C_{\preceq}$ be the number of arcs of $R$ which are consistent with $\preceq$.
Let $C_{\preceq}^1$ be the number of all such arcs between $X$ and $Z$, and $C_{\preceq}^2$ be the set of all such arcs inside of $X$ and inside of $Z$.
Observe that $C_{\preceq} = C_{\preceq}^1 + C_{\preceq}^2$.
Let $s = \sum_{z \in Z}\left| N^-_{T}[z] \cap X\right|$. By definition of ${S}$ and since each element $z$ of $Z$ is chosen uniformly at random, for every $x\in X$, we have $x \in N^-_{T}[z]$ with probability $1/2$, and all these events are independent. 
Therefore, $\mathbb{P}\left[s \leq \frac{|X| \cdot |Z|}{2} \right] \geq 1/2$.
We show that if $s \leq \frac{|X| \cdot |Z|}{2}$ then there exists an order $\preceq$ such that $C_{\preceq}^1 \geq \left(\frac{1 + \varepsilon}{2} \right) \cdot |X| \cdot |Z|$.
Pick $i \in [m]$ uniformly at random and consider the order $\preceq \ \coloneqq \ \leq_i$.
Since every $z\in Z$ is in ${S}$, we have $X \cap N^-_{D_{1/2}}[z] = X \cap N^-_{D_{1/2-\varepsilon}}[z]$.
Thus, if $x \in N^-_{D_{1/2}}[z]$ then $\mathbb{P}[x \preceq z] \geq 1/2$ by definition of $D_{1/2}$.
On the other hand, if $x \notin N^-_{D_{1/2}}[z]$ then $x \notin N^-_{D_{1/2-\varepsilon}}[z]$ so $\mathbb{P}[x \preceq z] \leq 1/2-\varepsilon$ by definition of $D_{1/2-\varepsilon}$.
Thus, if $x \in N^+_{D_{1/2}}[z]$ then $\mathbb{P}[z \preceq x] \geq 1/2 + \varepsilon$.
Summing over all $(x, z) \in X \times Z$, we have 
\begin{align*}
    \mathbb{E}\left[C^1_{\preceq}\right] &\geq s/2 + (1/2+\varepsilon) \cdot (|X| \cdot |Z| - s) \\ &=(1/2+\varepsilon) \cdot |X| \cdot |Z|-s\varepsilon \\
    &\geq \left(\frac{1 + \varepsilon}{2}\right) \cdot |X| \cdot |Z|.
\end{align*}

Therefore there exists some $i \in [m]$ such that $C^1_{\leq_i} \geq (\frac{1 + \varepsilon}{2}) \cdot |X| \cdot |Z|$. 

There are $2 \cdot \binom{k}{2}$ arcs which are either inside of $X$ or inside of $Z$, and each of them is consistent with $\leq_i$ with probability $1/2$, independently of all others. Therefore, $\mathbb{P}\left[C_{\leq_i}^2 \geq \binom{k}{2}\right] \geq 1/2$.
Furthermore, the events that $C_{\leq_i}^2 \geq \binom{k}{2}$ and that $s \leq \frac{|X| \cdot |Z|}{2}$ are independent, with probability at least $1/2$ each.
Thus, both happen conjointly with probability at least $1/4$, in which case we have $C_{\leq_i} = C_{\leq_i}^1 + C_{\leq_i}^2 \geq \left(\frac{1 + \varepsilon}{2}\right) \cdot k^2 + \binom{k}{2}$.
We just showed that with probability at least $1/4$, there exists a linear order $\preceq$ on $X \cup Z$ such that $C_{\preceq} \geq \left(\frac{1 + \varepsilon}{2}\right) \cdot k^2 + \binom{k}{2}$.

By \cref{lem:delavega}, there exists an integer $N$ such that for every $n \geq N$, if $T_n$ is a uniformly random $n$-vertex tournament, $\mathbb{P}\left[h(T_n) \leq \frac{1}{2} \cdot \binom{n}{2} + 1.73n^{3/2}\right] < 1/4$.
Since $R$ is a uniformly random tournament on $2k$ vertices, either $2k \leq N$ or $\left(\frac{1 + \varepsilon}{2}\right) \cdot k^2 + \binom{k}{2} \leq \frac{1}{2} \cdot \binom{2k}{2} + 1.73 \cdot (2k)^{3/2}$.

In the second case, we obtain \begin{align*}
\frac{1+\varepsilon}{2} \cdot k^2 + \binom{k}{2} \leq \frac{1}{2} \cdot \binom{2k}{2} + 1.73 \cdot (2k)^{3/2} &\iff \frac{1+\varepsilon}{2} \cdot k^2 + \frac{k^2}{2} - \frac{k}{2} \leq \frac{2k(2k-1)}{4} +  1.73 \cdot (2k)^{3/2} \\
&\iff k^2 - \frac{k}{2} + \frac{\varepsilon k^2}{2} \leq k^2 - \frac{k}{2} + 1.73 \cdot (2k)^{3/2} \\
&\iff \sqrt{k} \leq \frac{1.73 \cdot 4\sqrt{2}}{\varepsilon}
\end{align*}

This concludes the proof that $k = O\left(\frac{1}{\varepsilon^2}\right)$.
\end{proof}

We do not know if this quadratic upper bound is optimal.

\begin{remark} We could have proved that the tri-tournament $T'$ has bounded VC-dimension in a more geometric fashion, using that $T'$ is ``gap-representable'' on the torus.
To see it, consider the $m$-dimensional torus $T^m = S^1 \times \ldots \times S^1$ and fix an embedding $\phi : V \to T^m$ such that for every coordinate $i \in [m]$, all the $\phi(v)_i$ belong to an interval of size $\varepsilon/2$, and are ordered along this interval according to $\leq_i$.
Then, if $u$ comes before $v$ in at least a $(1/2)$-fraction of the orders, the oriented distance between $\phi(u)$ and $\phi(v)$ is at most $1/2 + \varepsilon/4$, whereas if $u$ comes before $v$ in at most a $(1/2-\varepsilon)$-fraction of the orders, the oriented distance between $\phi(u)$ and $\phi(v)$ is at least $(1/2 + \varepsilon) \cdot (1 - \varepsilon/2) \geq 1/2 + \varepsilon/2$ if $\varepsilon$ is small enough.
Then, we could have used arguments similar to the proof of \cref{thm:VCdim-spherical} to show that tri-tournaments gap-representable on the torus have bounded VC-dimension.
\end{remark}

A closely related question was raised by Elkind, Lang and Saffidine \cite{ELS11,ELS15}: say that a set $X$ of applicants is \emph{$\alpha$-undominated} is for every applicant $y \notin X$, less than an $\alpha$-fraction of the referees rank $y$ before all $x \in X$. They gave examples where any $1/2$-undominated set has size at least $3$, and proved that if there are $n$ applicants, there is always a $1/2$-undominated set of size $\log(n)$. Recently, Charikar, Lassota, Ramakrishnan, Vetta and Wang \cite{CLRVW24} managed to reduce this bound to only $6$. 
\section*{Acknowledgements}
Many thanks: to Raphael Steiner for pointing out the Kierstead-Szemerédi-Trotter theorem  \cite{DBLP:journals/combinatorica/KiersteadST84}, to Yuval Widgerson for linking tri-hypergraphs and partial concept classes, to Noga Alon for suggesting the use of \cref{lem:delavega} (allowing us to shave a $\log$ factor from \cref{thm:majosloppy}), and for informing us of the independent proof of this result by Moses Charikar, Prasanna Ramakrishnan and Kangning Wang. We are very grateful to Alantha Newman for allowing us to include her proof of \cref{thm:k-arc-bdd}.
\bibliographystyle{plain}

\begin{thebibliography}{10}

\bibitem{Abo24}
Pierre Aboulker, Guillaume Aubian, and Pierre Charbit.
\newblock Heroes in oriented complete multipartite graphs.
\newblock {\em Journal of Graph Theory}, 105(4):652--669, 2024.

\bibitem{ABGKM13}
Peter Allen, Julia Böttcher, Simon Griffiths, Yoshiharu Kohayakawa, and Robert
  Morris.
\newblock The chromatic thresholds of graphs.
\newblock {\em Advances in Mathematics}, 235:261--295, 2013.

\bibitem{Alon02}
Noga Alon.
\newblock Voting paradoxes and digraphs realizations.
\newblock {\em Advances in Applied Mathematics}, 29(1):126--135, 2002.

\bibitem{ALON2006374}
Noga Alon, Graham Brightwell, Henry~A. Kierstead, Alexandr~V. Kostochka, and
  Peter Winkler.
\newblock Dominating sets in $k$-majority tournaments.
\newblock {\em Journal of Combinatorial Theory, Series B}, 96(3):374--387,
  2006.

\bibitem{AFN07}
Noga Alon, Eldar Fischer, and Ilan Newman.
\newblock Efficient testing of bipartite graphs for forbidden induced
  subgraphs.
\newblock {\em SIAM Journal on Computing}, 37(3):959--976, 2007.

\bibitem{AHHM21}
Noga Alon, Steve Hanneke, Ron Holzman, and Shay Moran.
\newblock A theory of PAC learnability of partial concept classes.
\newblock {\em2021 IEEE 62nd Annual Symposium on Foundations of Computer Science (FOCS)}, 2021.

\bibitem{Alo01}
Noga Alon, J{\'{a}}nos Pach, and J{\'{o}}zsef Solymosi.
\newblock Ramsey-type theorems with forbidden subgraphs.
\newblock {\em Combinatorica}, 21(2):155--170, 2001.

\bibitem{BCCFLSST13}
Eli Berger, Krzysztof Choromanski, Maria Chudnovsky, Jacob Fox, Martin Loebl,
  Alex Scott, Paul Seymour, and Stéphan Thomassé.
\newblock Tournaments and colouring.
\newblock {\em Journal of Combinatorial Theory, Series B}, 103(1):1--20, 2013.

\bibitem{BKTW21}
\'{E}douard Bonnet, Eun~Jung Kim, St\'{e}phan Thomass\'{e}, and R\'{e}mi
  Watrigant.
\newblock Twin-width {I}: tractable {FO} model checking.
\newblock {\em Journal of the ACM}, 69(1), November 2021.

\bibitem{BRANDT200233}
Stephan Brandt.
\newblock A 4-colour problem for dense triangle-free graphs.
\newblock {\em Discrete Mathematics}, 251(1):33--46, 2002.
\newblock Cycles and Colourings.

\bibitem{BT2004}
Stephan Brandt and St{\'{e}}phan Thomass{\'{e}}.
\newblock Dense triangle-free graphs are four-colorable: a solution to the
  {E}rd{\H{o}}s-{S}imonovits problem, 2004.

\bibitem{CT07}
Pierre Charbit and St{\'{e}}phan Thomass{\'{e}}.
\newblock Graphs with large girth not embeddable in the sphere.
\newblock {\em Combinatorics, Probability and Computing}, 16(6):829--832, 2007.

\bibitem{CLRVW24}
Moses Charikar, Alexandra Lassota, Prasanna Ramakrishnan, Adrian Vetta, and Kangning Wang.
\newblock Six candidates suffice to win a voter majority, 2024.

\bibitem{CRW25}
Moses Charikar, Prasanna Ramakrishnan, and Kangning Wang.
\newblock Approximately dominating sets in elections, 2025.

\bibitem{Che97}
Chuan Chong Chen, Guo Ping Jin, and Khee Meng Koh.
\newblock Triangle-free graphs with large degree.
\newblock {\em Combinatorics, Probability and Computing}, 6(4):381–396, 1997.
  

\bibitem{Chud06}
Maria Chudnovsky.
\newblock Berge trigraphs.
\newblock {\em Journal of Graph Theory}, 53(1):1--55, 2006.

\bibitem{Chu18}
Maria Chudnovsky, Ringi Kim, Chun-Hung Liu, Paul Seymour, and Stéphan
  Thomassé.
\newblock Domination in tournaments.
\newblock {\em Journal of Combinatorial Theory, Series B}, 130:98--113, 2018.

\bibitem{DLV83}
Wenceslas Fernandez de la Vega.
\newblock On the maximum cardinality of a consistent set of arcs in a random tournament.
\newblock {\em Journal of Combinatorial Theory, Series B}, 35(3):328--332, 1983.

\bibitem{ELS11}
Edith Elkind, Jérôme Lang, and Abdallah Saffidine.
\newblock Choosing collectively optimal sets of alternatives based on the Condorcet criterion.
\newblock In Twenty-Second International Joint Conference on Artificial Intelligence. Citeseer, 2011.

\bibitem{ELS15}
Edith Elkind, Jérôme Lang, and Abdallah Saffidine.
\newblock Condorcet winning sets.
\newblock {\em Social Choice and Welfare}, 44(3):493--517, 2015.

\bibitem{Erdos79}
Paul Erd{\H{o}}s.
\newblock Problems and results in graph theory and combinatorial analysis.
\newblock In {\em Graph Theory and Related Topics, Academic Press, New York},
  pages 153–--163, 1979.
  
\bibitem{EM65}
Paul Erd{\H{o}}s and John W. Moon.
\newblock On sets of consistent arcs in a tournament.
\newblock {\em Canadian Mathematical Bulletin}, 8(3):269--271, 1965.
  
\bibitem{Erd73}
Paul Erd{\H{o}}s and Mikl{\'o}s Simonovits.
\newblock On a valence problem in extremal graph theory.
\newblock {\em Discrete Mathematics}, 5(4):323-334, 1973.

\bibitem{Fish92}
David~C. Fisher and Jennifer Ryan.
\newblock Optimal strategies for a generalized “scissors, paper, and stone”
  game.
\newblock {\em The American Mathematical Monthly}, 99(10):935–942, December 1992.

\bibitem{Fox21}
Jacob Fox, Matthew Kwan, and Benny Sudakov.
\newblock Acyclic subgraphs of tournaments with high chromatic number.
\newblock {\em Bulletin of the London Mathematical Society}, 53(2):619--630,
  2021.

\bibitem{FPS19}
Jacob Fox, J{\'a}nos Pach, and Andrew Suk.
\newblock Erd{\H{o}}s--{H}ajnal conjecture for graphs with bounded
  {VC}-dimension.
\newblock {\em Discrete \& Computational Geometry}, 61(4):809--829, June 2019.

\bibitem{https://doi.org/10.1002/jgt.22891}
Jacob Fox and Yuval Wigderson.
\newblock Minimum degree and the graph removal lemma.
\newblock {\em Journal of Graph Theory}, 102(4):648--665, 2023.

\bibitem{Gir24}
António Girão, Kevin Hendrey, Freddie Illingworth, Florian Lehner, Lukas
  Michel, Michael Savery, and Raphael Steiner.
\newblock Chromatic number is not tournament-local.
\newblock {\em Journal of Combinatorial Theory, Series B}, 168:86--95, 2024.

\bibitem{God11}
Wayne Goddard and Jeremy Lyle.
\newblock Dense graphs with small clique number.
\newblock {\em Journal of Graph Theory}, 66(4):319--331, 2011.

\bibitem{Goe95}
Michel~X. Goemans and David~P. Williamson.
\newblock Improved approximation algorithms for maximum cut and satisfiability
  problems using semidefinite programming.
\newblock {\em Journal of the ACM}, 42(6):1115–1145, November 1995.

\bibitem{Hag82}
Roland Häggkvist.
\newblock Odd cycles of specified length in non-bipartite graphs.
\newblock {\em North-Holland Mathematics Studies}, Graph Theory, Béla Bollobás editor, 62:89-99, 1982.

\bibitem{HLTW19}
Ararat Harutyunyan, Tien-Nam Le, Stéphan Thomassé, and Hehui Wu.
\newblock Coloring tournaments: from local to global.
\newblock {\em Journal of Combinatorial Theory, Series B}, 138:166--171, 2019.

\bibitem{Haussler95}
David Haussler.
\newblock Sphere packing numbers for subsets of the boolean {$n$}-cube with
  bounded {V}apnik-{C}hervonenkis dimension.
\newblock {\em Journal of Combinatorial Theory, Series A}, 69(2):217--232,
  1995.

\bibitem{HW87}
David Haussler and Emo Welzl.
\newblock $\varepsilon$-nets and simplex range queries.
\newblock {\em Discrete \& Computational Geometry}, 2(2):127--151, June 1987.

\bibitem{HLRX25}
Xinqi Huang, Hong Liu, Mingyuan Rong, and Zixiang Xu.
\newblock Interpolating chromatic and homomorphism thresholds, 2025.

\bibitem{Kar98}
David Karger, Rajeev Motwani, and Madhu Sudan.
\newblock Approximate graph coloring by semidefinite programming.
\newblock {\em Journal of the ACM}, 45(2):246–265, March 1998.

\bibitem{DBLP:journals/combinatorica/KiersteadST84}
Henry~A. Kierstead, Endre Szemer{\'{e}}di, and William~T. Trotter.
\newblock On coloring graphs with locally small chromatic number.
\newblock {\em Combinatorica}, 4(2):183--185, 1984.

\bibitem{KN23}
Felix Klingelhoefer and Alantha Newman.
\newblock Coloring tournaments with few colors: algorithms and complexity.
\newblock {\em SIAM Journal on Discrete Mathematics}, 38(4):3111--3133, 2024.

\bibitem{Kli24}
Felix Klingelh{\"{o}}fer and Alantha Newman.
\newblock Bounding the chromatic number of dense digraphs by arc neighborhoods.
\newblock {\em Combinatorica}, 44(4):881--895, 2024.

\bibitem{LSSZ24}
Hong Liu, Chong Shangguan, Jozef Skokan, and Zixiang Xu.
\newblock {Beyond chromatic threshold via (p,q)-theorem, and blow-up
  phenomenon}.
\newblock In Wolfgang Mulzer and Jeff~M. Phillips, editors, {\em 40th
  International Symposium on Computational Geometry (SoCG 2024)}, volume 293 of
  {\em Leibniz International Proceedings in Informatics (LIPIcs)}, pages
  71:1--71:15, Dagstuhl, Germany, 2024. Schloss Dagstuhl -- Leibniz-Zentrum
  f{\"u}r Informatik.

\bibitem{Lov78}
L\'aszl\'o Lovász.
\newblock Kneser's conjecture, chromatic number, and homotopy.
\newblock {\em Journal of Combinatorial Theory, Series A}, 25(3):319--324,
  1978.

\bibitem{LS10}
László Lovász and Balázs Szegedy.
\newblock Regularity {partitions} and {the} {topology} of {graphons}.
\newblock In Imre Bárány, József Solymosi, and Gábor Sági, editors, {\em
  An {Irregular} {Mind}: {Szemerédi} is 70}, pages 415--446. Springer, Berlin,
  Heidelberg, 2010.

\bibitem{Luc06}
Tomasz Łuczak.
\newblock On the structure of triangle-free graphs of large minimum degree.
\newblock {\em Combinatorica}, 26(4):489--493, 2006.

\bibitem{Luczak2010coloringdensegraphsvcdimension}
Tomasz Łuczak and Stéphan Thomassé.
\newblock Coloring dense graphs via {VC}-dimension, 2010.

\bibitem{Luc24}
Tomasz Łuczak, Joanna Polcyn, and Christian Reiher.
\newblock Strong {Brandt-Thomassé} theorems, 2024

\bibitem{NT74}
George~L. Nemhauser and Leslie~E. Trotter.
\newblock Properties of vertex packing and independence system polyhedra.
\newblock {\em Mathematical Programming}, 6(1):48--61, December 1974.

\bibitem{NSS24b}
Tung Nguyen, Alex Scott, and Paul Seymour.
\newblock Induced subgraph density. {VI}. {B}ounded {VC}-dimension, 2024.

\bibitem{NSS24}
Tung Nguyen, Alex Scott, and Paul Seymour.
\newblock Some results and problems on tournament structure.
\newblock {\em Journal of Combinatorial Theory, Series B}, 173:146--183, 2025.

\bibitem{nikiforov10}
Vladimir Nikiforov.
\newblock Chromatic number and mimimum degree of {$K_r$}-free graphs, 2010.

\bibitem{Obe20}
Heiner Oberkampf and Mathias Schacht.
\newblock On the structure of dense graphs with bounded clique number.
\newblock {\em Combinatorics, Probability and Computing}, 29(5):641--649, 2020.

\bibitem{O'R14}
Jonathan~Lyons O'Rourke.
\newblock Chromatic thresholds of regular graphs with small cliques.
\newblock Master's thesis, University of Southern Mississippi, May 2014.
\newblock Available at \url{https://aquila.usm.edu/masters_theses/27/}.

\bibitem{Sauer72}
Norbert Sauer.
\newblock On the density of families of sets.
\newblock {\em Journal of Combinatorial Theory, Series A}, 13(1):145--147,
  1972.

\bibitem{Sch78}
Alexander Schrijver.
\newblock Vertex-critical subgraphs of {K}neser-graphs.
\newblock {\em Nieuw Archief voor Wiskunde}, 26(3):454--461, 1978.

\bibitem{Shelah72}
Saharon Shelah.
\newblock {A combinatorial problem; stability and order for models and theories
  in infinitary languages.}
\newblock {\em Pacific Journal of Mathematics}, 41(1):247 -- 261, 1972.

\bibitem{skala13}
Matthew Skala.
\newblock Hypergeometric tail inequalities: ending the insanity, 2013.

\bibitem{asymptopia}
Joel Spencer.
\newblock {\em Asymptopia}.
\newblock American Mathematical Society, 2014.

\bibitem{Spencer80}
Joel Spencer.
\newblock Optimally ranking unrankable tournaments.
\newblock {\em Periodica Mathematica Hungarica}, 11:131--144, 1980.

\bibitem{CT02}
Carsten Thomassen.
\newblock On the chromatic number of triangle-free graphs of large minimum
  degree.
\newblock {\em Combinatorica}, 22(4):591--596, 2002.

\end{thebibliography}

\appendix

\section{The Euclidean Setting}\label{sec:euclidean}

In this section, we prove the Euclidean variants of DNL, which we restate below.

\DNLeuclidean*

\DNLEclusterintro*

To prove \cref{lem:DNLE-cluster-intro}, we actually prove the following stronger statement.

\begin{restatable}{lemma}{DNLEcluster}\label{lem:DNLE-cluster}
Let $V$ be a finite subset of $\mathbb{R}^N$. There is a partition of $V$ into $2^{\textup{poly}(\varepsilon^{-1}, \eta^{-1})}$ clusters such that if $u, v$ are in the same cluster, there are at most $\eta \cdot |V|$ points $w$ of $V$ such that $w \in B(u, 1) \setminus B(v, 1+ \varepsilon)$.
\end{restatable}

We start by recalling some well-known properties of the normal distribution, the last item is just a Chernoff bound.

\begin{lemma}\label{lem:properties-normal}
If $X \sim \mN\left(0, \frac{\pi}{2n}\right)$ then: \begin{itemize}
    \item $\mathbb{E}[|X|] = \frac{1}{\sqrt{n}}$, 
    \item $\mathbb{E}[X^2] = \frac{\pi}{2n}$, and
    \item For every $\beta > 0$, $\mathbb{P}[|X| \geq \beta] \leq 2\exp\left(-\frac{\beta^2 \cdot n}{\pi}\right)$.
\end{itemize}
\end{lemma}

The following result is the key ingredient for the proofs of \cref{lem:realdensecor,lem:DNLE-cluster}. Intuitively, given a set $V$ of points in $\mathbb{R}^n$ with small norm, it provides an embedding $\phi$ of $V$ into $\{0, 1\}^N$ so that vertices which are at distance at most $1$ in $\mathbb{R}^n$ will be significantly closer in this embedding than vertices which are at distance at least $1 + \varepsilon$ in $\mathbb{R}^n$. We will then be able to apply all the DNL tools on the corresponding Hamming-trigraph.

The idea of the proof is the following: we pick many uniformly random vectors $U^1, \ldots, U^N$ from the standard normal distribution in $\mathbb{R}^N$, and random thresholds $t^1, \ldots, t^N$ from some interval $I$. We set the $i$-th coordinate of $\phi(v)$ to $0$ if $\langle U^i, v\rangle \leq t^i$ and to $1$ otherwise.
Then, the probability that $\phi(v)$ and $\phi(w)$ disagree on the $i$-th coordinate is the probability that $t^i$ lies in the interval between $\langle U^i, v\rangle$ and $\langle U^i, w\rangle$, whose expected length is $\mathbb{E}[|\langle U^i, v-w\rangle|]$, which is proportional to $d(v,w)$.
Thus, intuitively the proportion of coordinates on which $\phi(v)$ and $\phi(w)$ will differ will be proportional to $d(v, w)$. The proof is mainly technical calculations to show that this intuition is indeed correct.

However, we run into many technical complications in the proof, mainly because we need to handle the cases where some of the scalar products $\langle U^i, v\rangle$ do not lie in the interval $I$.

\begin{lemma}\label{lem:euclidean-to-hamming}
Let $V$ be a finite set of $\mathbb{R}^n$ such that $\norm{v} \leq 3$ for every $v \in V$.
For every small enough $\varepsilon > 0$, there exists $N \in \mathbb{N}$ and an embedding $\phi : V \to \{0, 1\}^N$ such that whenever $v_1, v_2, w_1, w_2 \in V$ satisfy $d(v_1, w_1) \leq 1$ and $d(v_2, w_2) \geq 1 + \varepsilon$ then $d_H(\phi(v_2), \phi(w_2)) - d_H(\phi(v_1), \phi(w_1)) \geq \varepsilon^2 \cdot N$.
\end{lemma}

\begin{proof}
Set $\alpha = \varepsilon^2/(4\pi)$ and $\beta = \sqrt{\pi \cdot \ln(2/\alpha)}$.
Let $N$ be large enough so that $\exp\left(-2 \cdot \varepsilon^4 \cdot N\right) < 1/|V|^2$.

Let $U^1, \ldots, U^N$ be uniformly random vectors from $\mathcal{N}\left(0, \frac{\pi}{2n} \cdot I_n\right)$ (where $I_n$ is the $n \times n$ identity matrix) and $t^1, \ldots, t^N$ be uniformly random values in $[-3\beta/\sqrt{n}, 3\beta/\sqrt{n}]$.
For every $i \in [N]$ and $v \in V$, set $\phi(v)_i = 0$ if $\langle U^i, v\rangle \leq t^i$ and $\phi(v)_i = 1$ if $\langle U^i, v\rangle > t^i$. This defines an embedding $\phi : V \to \{0, 1\}^N$. We now show that $\phi$ satisfies the desired properties with positive probability.

Let $v, w \in V$ and $i \in [N]$.
Then, $\phi(v)_i \neq \phi(w)_i$ if and only if $t^i$ lies in the interval between $\langle U^i, v\rangle$ and $\langle U^i, w\rangle$.
Let $A$ be the event that $|\langle U^i, w \rangle| \leq 3 \cdot \beta/\sqrt{n}$ and $|\langle U^i, w \rangle| \leq 3 \cdot \beta/\sqrt{n}$.
First, by rotational invariance of the normal distribution in $\mathbb{R}^n$, the variables $|\langle U^i, v \rangle|$ and $\norm{v} \cdot |(U^i)_1|$ have the same distribution.
Observe that since $U^i \sim \mN\left(0, \frac{\pi}{2n} \cdot I_n \right)$ then $(U^i)_1 \sim \mN\left(0, \frac{\pi}{2n}\right)$.
If $v = 0$ then $\mathbb{P}\left[|\langle U^i, v \rangle| \geq \frac{3 \cdot \beta}{\sqrt{n}}\right] = 0 \leq \alpha$, and if $v \neq 0$ then by \cref{lem:properties-normal}, using that $\norm{v} \leq 3$ and by definition of $\beta$, we have \begin{align*}
    \mathbb{P}\left[|\langle U^i, v \rangle| \geq \frac{3 \cdot \beta}{\sqrt{n}}\right] 
    &\leq \mathbb{P}\left[|\langle U^i, v \rangle| \geq \frac{\norm{v} \cdot \beta}{\sqrt{n}}\right] \\
    &= \mathbb{P}\left[\norm{v} \cdot |(U^i)_1| \geq \frac{\norm{v} \cdot \beta}{\sqrt{n}}\right] \\
    &= \mathbb{P}\left[|(U^i)_1| \geq \frac{\beta}{\sqrt{n}}\right] \\
    &\leq \alpha.
\end{align*} 
Thus, \begin{equation}\label{eq:probaA}
    \mathbb{P}[A] \geq 1 - 2\alpha.
\end{equation}

Consider $v, w \in V$ such that that $d(v, w) \leq 1$.
Let $Z$ be the random variable $|\langle U^i, v - w \rangle|$.
Then,
\begin{align*}
    \mathbb{P}[\phi(v)_i \neq \phi(w)_i] &= \mathbb{P}[A] \cdot \mathbb{P}[\phi(v)_i \neq \phi(w)_i~|~A] + \mathbb{P}[\overline{A}] \cdot \mathbb{P}[\phi(v)_i \neq \phi(w)_i \ | \ \overline{A}]\\
    &\leq \mathbb{P}[\phi(v)_i \neq \phi(w)_i \ | \ A] + \mathbb{P}[\overline{A}].
\end{align*}
However, since $t^i$ is a uniformly random vector in $[-3\beta/\sqrt{n}, 3\beta/\sqrt{n}]$, $$\mathbb{P}[\phi(v)_i \neq \phi(w)_i~|~A] = \frac{\mathbb{E}[|\langle U^i, v\rangle-\langle U^i, w\rangle|~|~A]|}{6\beta/\sqrt{n}} = \frac{\mathbb{E}[Z~|~A]}{6\beta/\sqrt{n}}.$$
Using \cref{eq:probaA}, we get $$\mathbb{E}[Z] = \mathbb{P}[A] \cdot \mathbb{E}[Z~|~A] + \mathbb{P}[\overline{A}] \cdot \mathbb{E}[Z~|~\overline{A}] \geq (1-2\alpha) \cdot \mathbb{E}[Z~|~A].$$
Thus, using \cref{lem:properties-normal} and that $\norm{v-w} \leq 1$, $$\mathbb{E}[Z~|~A] \leq \frac{\mathbb{E}[Z]}{1-2\alpha} = \frac{\norm{v-w} \cdot \mathbb{E}[|(U^i)_1|]}{1-2\alpha} \leq \frac{1}{(1-2\alpha)\sqrt{n}}.$$
Using \cref{eq:probaA}, we finally obtain 
\begin{equation}\label{eq:upperboundproba}
    \mathbb{P}[\phi(v)_i \neq \phi(w)_i] \leq \frac{1}{(1-2\alpha)6\beta} + 2\alpha.
\end{equation}

Consider now $v, w \in V$ such that that $d(v, w) \geq 1 + \varepsilon$.
Let $Z$ be the random variable $|\langle U^i, v - w \rangle|$.
Then,
\begin{align*}
    \mathbb{E}[Z] &= \mathbb{E}[Z \cdot \mathds{1}_A] + \mathbb{E}\left[Z \cdot \mathds{1}_{\overline{A}}\right] \\
    &= \mathbb{P}[A] \cdot \mathbb{E}[Z~|~A] + \mathbb{E}\left[Z \cdot \mathds{1}_{\overline{A}}\right] \\
    &\leq \mathbb{E}[Z~|~A] + \mathbb{E}\left[Z \cdot \mathds{1}_{\overline{A}}\right]
\end{align*}
Thus, $$\mathbb{E}[Z~|~A] \geq \mathbb{E}[Z] - \mathbb{E}\left[Z \cdot \mathds{1}_{\overline{A}}\right].$$
By the Cauchy-Schwarz inequality, using \cref{lem:properties-normal} and \cref{eq:probaA} we get
\begin{align*}
    \mathbb{E}\left[Z \cdot \mathds{1}_{\overline{A}}\right] &\leq \sqrt{\mathbb{E}[Z^2] \cdot \mathbb{E}\left[\left(\mathds{1}_{\overline{A}}\right)^2\right]} \\
    &= \norm{v-w}\sqrt{\mathbb{E}[|(U^i)_1|^2] \cdot \mathbb{P}\left[\overline{A}\right]} \\
    &\leq \norm{v-w} \cdot \sqrt{\frac{\pi \cdot 2\alpha}{2n}}.
\end{align*}
Thus, again using \cref{lem:properties-normal} and \cref{eq:probaA}, we have $$\mathbb{E}[Z~|~A] \geq \frac{\norm{v-w}}{\sqrt{n}} \cdot \left(1 - \sqrt{\pi\alpha}\right).$$
Note that 
\begin{align*}
    \mathbb{P}[\phi(v)_i \neq \phi(w)_i~|~A] &= \frac{\mathbb{E}[|\langle U^i, v\rangle-\langle U^i, w\rangle|~|~A]}{6\beta/\sqrt{n}} \\
    &= \frac{\mathbb{E}[Z~|~A]}{6\beta/\sqrt{n}} \\
    &\geq \frac{\norm{v-w}}{6\beta} \cdot \left(1 - \sqrt{\pi\alpha}\right).
\end{align*}
Then,
\begin{align*}
    \mathbb{P}[\phi(v)_i \neq \phi(w)_i] &= \mathbb{P}[A] \cdot \mathbb{P}[\phi(v)_i \neq \phi(w)_i~|~A] + \mathbb{P}[\overline{A}] \cdot \mathbb{P}[\phi(v)_i \neq \phi(w)_i \ | \ \overline{A}] \\
    &\geq \mathbb{P}[A] \cdot \mathbb{P}[\phi(v)_i \neq \phi(w)_i~|~A] \\
    &\geq (1-2\alpha) \cdot \frac{\norm{v-w}}{6\beta} \cdot \left(1 - \sqrt{\pi\alpha}\right).
\end{align*}

Using that $\norm{v-w} \geq 1 + \varepsilon$, we finally obtain \begin{equation}\label{eq:lowerboundproba}
    \mathbb{P}[\phi(v)_i \neq \phi(w)_i] \geq \frac{1+\varepsilon}{6\beta} \cdot \left(1 - \sqrt{\pi\alpha}\right) \cdot (1-2\alpha).
\end{equation}

We now show that \begin{equation}\label{eq:euclidean}
    \frac{1}{(1-2\alpha)6\beta} + 2\alpha + 3\varepsilon^2 \leq \frac{1+\varepsilon}{6\beta} \cdot \left(1 - \sqrt{\pi\alpha}\right) \cdot (1-2\alpha).
\end{equation}
Using the definition of $\alpha$, we get
\begin{align*}
    \frac{1+\varepsilon}{6\beta} \cdot \left(1 - \sqrt{\pi\alpha}\right) \cdot (1-2\alpha) &= \frac{(1+\varepsilon)(1-\varepsilon/2)}{6\beta} \cdot \left(1-\frac{\varepsilon^2}{2\pi}\right) \\
    &= \frac{(1 + \varepsilon/2-\varepsilon^2/2) \cdot (1-\varepsilon^2/(2\pi))}{6\beta} \\
    &= \frac{1 + \varepsilon/2 - \varepsilon^2/2 -\varepsilon^2/(2\pi) -\varepsilon^3/(4\pi) + \varepsilon^4/(4\pi)}{6\beta} \\
    &\geq \frac{1 + \varepsilon/2 - \varepsilon^2}{6\beta}.
\end{align*}
Using that $\frac{1}{1-x} \leq 1+2x$ for $0 \leq x \leq 1/2$, and the definition of $\alpha$, we also have $$\frac{1}{(1-2\alpha)6\beta} + 2\alpha \leq \frac{1+4\alpha}{6\beta} + 2\alpha \leq \frac{1+\varepsilon^2/\pi}{6\beta} + \frac{\varepsilon^2}{2\pi} \leq \frac{1+\varepsilon^2}{6\beta} + \varepsilon^2.$$
To prove \cref{eq:euclidean}, it then suffices to prove that $$\frac{1+\varepsilon^2}{6\beta} + \varepsilon^2 + 3\varepsilon^2 \leq \frac{1 + \varepsilon/2-\varepsilon^2}{6\beta}.$$
We have: \begin{align*}
    \frac{1+\varepsilon^2}{6\beta} + \varepsilon^2 + 3\varepsilon^2 \leq \frac{1 + \varepsilon/2-\varepsilon^2}{6\beta} &\iff 1 + \varepsilon^2 + 24\varepsilon^2\beta \leq 1 + \frac{\varepsilon}{2} - \varepsilon^2 \\
    &\iff 2\varepsilon + 24\varepsilon\beta \leq \frac{1}{2}.
\end{align*}

Since $\beta = \sqrt{\pi \cdot \ln(2/\alpha)} = \sqrt{\pi \cdot \ln(8\pi/\varepsilon^2)} = \sqrt{2\pi \cdot \ln(\sqrt{8\pi}/\varepsilon)}$ then $\varepsilon\beta \xrightarrow[\varepsilon \to 0]{} 0$ so for $\varepsilon$ small enough, \cref{eq:euclidean} indeed holds.

Let $v, w \in V$ and let $D$ be the random variable $d_H(\phi(v), \phi(w))$.
If $d(v, w) \leq 1$ then $D$ is a sum of $N$ independent Bernoulli random variables with success probability at most $\frac{1}{(1-2\alpha)6\beta} + 2\alpha$ by \cref{eq:upperboundproba}.
By Hoeffding's inequality, and using the definition of $N$, we get
\begin{align*}
    \mathbb{P}\left[D \geq \left(\frac{1}{(1-2\alpha)6\beta} + 2\alpha + \varepsilon^2\right) \cdot N\right]
    &\leq \mathbb{P}[D - \mathbb{E}[D] \geq \varepsilon^2 \cdot N] \\
    &\leq \exp\left(-\frac{2\varepsilon^4 \cdot N^2}{N}\right) \\
    &= \exp\left(-2\varepsilon^4 \cdot N\right) \\
    &< \frac{1}{|V|^2}.
\end{align*}

If $d(v, w) \geq 1+\varepsilon$ then $D$ is a sum of $N$ independent Bernoulli random variables with success probability at least $\frac{1+\varepsilon}{6\beta} \cdot \left(1 - \sqrt{\pi\alpha}\right) \cdot (1-2\alpha)$ by \cref{eq:lowerboundproba}.
By Hoeffding's inequality, and using the definition of $N$, we get 
\begin{align*}
    \mathbb{P}\left[D \leq \left(\frac{1+\varepsilon}{6\beta} \cdot \left(1 - \sqrt{\pi\alpha}\right) \cdot (1-2\alpha) - \varepsilon^2\right) \cdot N\right]
    &\leq \mathbb{P}[D - \mathbb{E}[D] \leq -\varepsilon^2 \cdot N] \\ &\leq \exp\left(-\frac{2\varepsilon^4 \cdot N^2}{N}\right) \\
    &= \exp\left(-2\varepsilon^4 \cdot N\right)\\
    &< \frac{1}{|V|^2}.
\end{align*}

By a union bound, the probability that $d_H(\phi(v), \phi(w)) \leq \left(\frac{1}{(1-2\alpha)6\beta} + 2\alpha + \varepsilon^2\right) \cdot N$ whenever $d(v, w) \leq 1$ and that $d_H(\phi(v), \phi(w)) \geq \left(\frac{1+\varepsilon}{6\beta} \cdot \left(1 - \sqrt{\pi\alpha}\right) \cdot (1-2\alpha) - \varepsilon^2\right) \cdot N$ whenever $d(v, w) \geq 1 + \varepsilon$, is positive so there exists $\phi$ which satisfies these conditions. 

Given such a $\phi$, consider $v_1, v_2, w_1, w_2 \in V$ such that $d(v_1, w_1) \leq 1$ and $d(v_2, w_2) \geq 1 + \varepsilon$. It then follows from \cref{eq:euclidean} that $d_H(\phi(v_2), \phi(w_2)) - d_H(\phi(v_1), \phi(w_1)) \geq \varepsilon^2$.
\end{proof}

We now have all the tools to prove \cref{lem:realdensecor,lem:DNLE-cluster}.

\begin{proof}[Proof of \cref{lem:realdensecor}.]
Let $\{z_1, \ldots, z_c\}$ be a maximal set of elements of $V$ so that $d(z_i, z_j) > 2$ for every $i \neq j \in [c]$.
Observe that the radius 1 balls around the $z_i$ are pairwise disjoint, and by assumption each of them contains at least $\delta|V|$ elements of $V$, which implies $c \leq 1/\delta$.
Fix $i \in [c]$ and let $V_i = \{v \in V : d(v, z_i) \leq 3\}$. Up to translating the elements of $V_i$, we can assume that $\norm{v} \leq 3$ for every $v \in V_i$. By \cref{lem:euclidean-to-hamming}, there exists $N \in \mathbb{N}$ and an embedding $\phi_i : V_i \to \{0, 1\}^N$ such that whenever $u_1, u_2, v_1, v_2 \in V_i$ satisfy $d(u_1, v_1) \leq 1$ and $d(u_2, v_2) \geq 1 + \varepsilon$ then $d_H(\phi_i(u_2), \phi_i(v_2)) - d_H(\phi_i(u_1), \phi_i(v_1)) \geq \varepsilon^2 \cdot N$.

Consider the trigraph $T_i = (V_i, E_i, R_i)$ with $E_i = \{uv : d(u, v) \leq 1\}$ and $R_i = \{uv : 1 < d(u, v) \leq 1 + \varepsilon\}$.
The embedding $\phi_i$ certifies that $T_i$ is a Hamming-trigraph with sensitivity $\varepsilon^2$.
By DNL, $T_i$ has a $\delta$-net $X_i$ of size $\text{poly}(\varepsilon^{-1}, \delta^{-1})$.

Let $X = X_1 \cup \ldots \cup X_c$. Then, $|X| = \text{poly}(\varepsilon^{-1}, \delta^{-1})$.
Consider any $v \in V$. By maximality of $\{z_1, \ldots, z_c\}$, there exists $i \in [c]$ such that $d(v, z_i) \leq 2$, and thus $B(v, 1)$, the ball of radius 1 centered in $v$, is entirely contained in $V_i$. Then, $B(v, 1) \subseteq N_{T_i}[v]$ so $\left|N_{T_i}[v]\right| \geq |B(v, 1)| \geq \delta|V| \geq \delta|V_i|$ so there exists $x \in X_i$ such that $x \in N_{T_i}[v] \cup R_{T_i}(v)$, i.e. $d(v, x) \leq 1 + \varepsilon$.
Then, for every $v \in V$, there exists $x \in X$ such that $d(v, x) \leq 1+\varepsilon$, so $V$ is covered by the $X$-balls with radius $1 + \varepsilon$.
\end{proof}

\begin{proof}[Proof of \cref{lem:DNLE-cluster}.]
Let $\{z_1, \ldots, z_c\}$ be a maximal set of elements of $V$ so that $d(z_i, z_j) > 2$ for every $i \neq j \in [c]$, and $|B(z_i, 1)| \geq \eta \cdot |V|$ for every $i \in [c]$.
Observe that the radius 1 balls around the $z_i$ are pairwise disjoint, and by assumption each of them contains at least $\eta \cdot |V|$ elements of $V$, which implies $c \leq 1/\eta$.
Fix $i \in [c]$ and let $V_i = \{v \in V : d(v, z_i) \leq 3\}$. Up to translating the elements of $V_i$, we can assume that $\norm{v} \leq 3$ for every $v \in V_i$. By \cref{lem:euclidean-to-hamming}, there exists $N_i \in \mathbb{N}$ and an embedding $\phi_i : V_i \to \{0, 1\}^{N_i}$ such that whenever $u_1, u_2, v_1, v_2 \in V_i$ satisfy $d(u_1, v_1) \leq 1$ and $d(u_2, v_2) \geq 1 + \varepsilon$ then $d_H(\phi_i(u_2), \phi_i(v_2)) - d_H(\phi_i(u_1), \phi_i(v_1)) \geq \varepsilon^2 \cdot N_i$.

Consider the trigraph $T_i = (V_i, E_i, R_i)$ with $E_i = \{uv : d(u, v) \leq 1\}$ and $R_i = \{uv : 1 < d(u, v) \leq 1 + \varepsilon\}$.
The embedding $\phi_i$ certifies that $T_i$ is a Hamming-trigraph with sensitivity $\varepsilon^2$, for some threshold $c_i - \varepsilon^2$.
By \cref{lem:DNLH-cluster}, $V_i$ can be partitioned into $2^{\textup{poly}(\varepsilon^{-1}, \eta^{-1})}$ pre-clusters such that if $u, v \in V_i$ are in the same pre-cluster, there are at most $\eta \cdot |V_i|$ points $w$ of $V_i$ such that $d_H(\phi_i(v), \phi_i(w)) \geq c_i \cdot N_i$ and $d_H(\phi_i(u), \phi_i(w)) \leq (c_i - \varepsilon^2) \cdot N_i$.

We then partition $V$ into clusters such that $u, v$ are in the same cluster if and only for every $i \in [c]$, either both $u$ and $v$ are not in $V_i$, or they are both in $V_i$ and they belong to the same pre-cluster for $T_i$.
Observe that the number of clusters is $\left(2^{\textup{poly}(\varepsilon^{-1}, \eta^{-1})}\right)^{c} = 2^{\textup{poly}(\varepsilon^{-1}, \eta^{-1})}$.

Consider now $u, v \in V$ which belong to the same cluster. If $|B(u, 1)| < \eta \cdot |V|$ then $|B(u, 1) \setminus B(v, 1 + \varepsilon)| \leq \eta \cdot |V|$ and we are done. Assume now that $|B(u, 1)| \geq \eta \cdot |V|$. By maximality of $\{z_1, \ldots, z_c\}$, there exists $i \in [c]$ such that $d(u, z_i) \leq 2$, and thus $B(u, 1)$, the ball of radius 1 centered in $u$, is entirely contained in $V_i$.
Since $u$ and $v$ are in the same cluster then $v \in V_i$ and $u$ and $v$ are in the same pre-cluster for $T_i$.
Consider $w \in V$ such that $w \in B(u, 1) \setminus B(v, 1 + \varepsilon)$.
Then, $w \in V_i$ and $uw \in E_i$ so $d_H(\phi_i(u), \phi_i(w)) \leq (c_i - \varepsilon^2) \cdot N_i$. 
Furthermore, $vw \notin E_i \cup R_i$ so $d_H(\phi_i(v), \phi_i(w)) \geq c_i \cdot N_i$.
By definition of the pre-clusters in $T_i$, there are at most $\eta \cdot |V_i| \leq \eta \cdot |V|$ such $w \in V$.
\end{proof}

\section{Proofs of \texorpdfstring{\cref{sec:VCdim-trigraphs}}{Section 4.1}} \label{sec:proofs-VCdim-trigraphs}

\VCdimspherical*

\begin{proof}
Let $T = (V, E, R)$ be a spherical-trigraph with sensitivity $\varepsilon$. 
There exists an integer $N$ so that $V \subseteq \mathbb{S}^{N-1} \subseteq \mathbb{R}^N$, and a threshold $\tau$ such that if $uv \in E$ then $d_S(u, v) \leq \tau$, and if $uv \notin E \cup R$ then $d_S(u, v) \geq \tau + \varepsilon$.

Let $m$ be such that $\exp\left(-\frac{\varepsilon^2 \cdot m}{8 \cdot \pi^2}\right) < \frac{1}{|V|^2}$.
Pick $m$ uniformly random points $p_1, \ldots, p_m \in \mathbb{S}^{N-1}$.
Define $\phi : V \to \{0, 1\}^m$ by setting $\phi(v)_i = 0$ if $\langle v, p_i\rangle \leq 0$ and $\phi(v)_i = 1$ if $\langle v, p_i\rangle > 0$.

For $u, v \in V$ and $i \in [m]$, $\phi(u)_i \neq \phi(v)_i$ if and only if the hyperplane orthogonal to $p_i$ separates $u$ and $v$. This happens with probability $d_S(u, v)/\pi$, and these events are pairwise independent.
Let $D$ be the random variable $d_H(\phi(u), \phi(v))$.
Thus, $D$ is the sum of $m$ independent Bernoulli random variables with success probability $d_S(u, v)/\pi$.

Suppose first $d_S(u, v) \leq \tau$.
Then, by Hoeffding's inequality, and using the definition of $m$, we get \begin{align*}
    \mathbb{P}\left[D \geq \frac{\tau + \varepsilon/4}{\pi} \cdot m\right] &\leq \mathbb{P}\left[D - \mathbb{E}[D] \geq \frac{\varepsilon}{4\pi} \cdot m\right] \\
    &\leq \exp\left(-\frac{2 \cdot \varepsilon^2 \cdot m^2}{16\pi^2 \cdot m}\right) \\
    & \exp\left(-\frac{\varepsilon^2 \cdot m}{8 \cdot \pi^2}\right) \\
    &< \frac{1}{|V|^2}.
\end{align*}
Suppose now that $d_S(u, v) \geq \tau + \varepsilon$.
Again, Hoeffding's inequality gives \begin{align*}
    \mathbb{P}\left[D \leq \frac{\tau + 3\varepsilon/4}{\pi} \cdot m\right] &\leq \mathbb{P}\left[D - \mathbb{E}[D] \leq -\frac{\varepsilon}{4\pi} \cdot m\right] \\
    &\leq \exp\left(-\frac{2 \cdot \varepsilon^2 \cdot m^2}{16\pi^2 \cdot m}\right) \\
    & \exp\left(-\frac{\varepsilon^2 \cdot m}{8 \cdot \pi^2}\right) \\
    &< \frac{1}{|V|^2}.
\end{align*}

By a union bound, with positive probability we have $d_H(\phi(u), \phi(v)) \leq \frac{\tau + \varepsilon/4}{\pi} \cdot m$ whenever $d(u, v) \leq~\tau$, and $d_H(\phi(u), \phi(v)) \geq \frac{\tau + 3\varepsilon/4}{\pi} \cdot m$ whenever $d(u, v) \geq \tau + \varepsilon$.
This proves that $G$ is a $\frac{\varepsilon}{2\pi}$-Hamming-trigraph.

The last assertion then follows immediately from \cref{thm:VCdim-Hamming}.
\end{proof}

\Hamissph*

\begin{proof}
Let $T = (V, E, R)$ be a Hamming-trigraph with sensitivity $\varepsilon$. Let $N$ be such that $V \subseteq \{0, 1\}^N$, and $\tau$ be such that if $uv \in E$ then $d_H(u, v) \leq \tau \cdot N$ and if $uv \notin E \cup R$ then $d_H(u, v) \geq (\tau + \varepsilon) \cdot N$.
For every $v \in V$, let $\phi(v) \in \mathbb{S}^{N-1}$ be defined by $\phi(v)_i = \frac{1}{\sqrt{N}}$ if $v_i = 1$ and $\phi(v)_i = \frac{-1}{\sqrt{N}}$ if $v_i = 0$.
For every $u, v \in V$, we have $\langle \phi(u), \phi(v)\rangle = d_H(u, v) \cdot \frac{-1}{N} + (N - d_H(u, v)) \cdot \frac{1}{N} = 1 - \frac{2 \cdot d_H(u, v)}{N}$, and thus $d_S(\phi(u), \phi(v)) = \arccos\left(1 - \frac{2 \cdot d_H(u, v)}{N}\right)$.
Let $u_1, u_2, v_1, v_2 \in V$ be such that $u_1v_1 \in E$ and $u_2v_2 \notin E \cup R$. Then, $d_H(u_1, v_1) \leq \tau \cdot N$ and $d_H(u_2, v_2) \geq (\tau + \varepsilon) \cdot N$.
Using that if $-1 \leq a \leq b \leq 1$ then $\arccos(a) - \arccos(b) \geq b-a$, we get 
\belowdisplayskip=-12pt\begin{align*}
    d_S(\phi(u_2), \phi(v_2)) - d_S(\phi(u_1), \phi(v_1)) &= \arccos\left(1 - \frac{2 \cdot d_H(u_2, v_2)}{N}\right) - \arccos\left(1 - \frac{2 \cdot d_H(u_1, v_1)}{N}\right) \\
    &\geq 1 - \frac{2 \cdot d_H(u_1, v_1)}{N} - \left(1 - \frac{2 \cdot d_H(u_2, v_2)}{N}\right) \\
    &\geq \frac{2 \cdot d_H(u_2, v_2) - 2 \cdot d_H(u_1, v_1)}{N} \\
    &\geq \frac{2 \cdot (\tau + \varepsilon) \cdot N - 2 \cdot \tau \cdot N}{N} \\
    &= 2\varepsilon. 
\end{align*}
\end{proof}

\section{Proofs of \texorpdfstring{\cref{sec:VCdim-sets}}{Section 4.4}} \label{sec:proofs-VCdim-sets}

\VCdimsetintertuple*

\begin{proof}
Write $\mI \cap \mI = (\mT, \mE)$ with $\mE = \{(B(vw), R(vw), W(vw)): v, w \in V\}$.
Observe that if $(x^1, \ldots, x^t) \in B(vw)$ then for every $i \in [t]$ we have $\mF_{vx^i} = \emptyset$ and $\mF_{wx^i} = \emptyset$. On the other hand, if $(x^1, \ldots, x^t)\in W(vw)$ then there exists $i \in [t]$ such that either $|\mF_{vx^i}| \geq \varepsilon |\mF|$ or $|\mF_{wx^i}| \geq \varepsilon |\mF|$.

Let $X \subseteq \mT$ be a shattered set: for every $Y \subseteq X$ there exists $p_Y = (v_Y, w_Y) \in V^2$ such that $B(v_Yw_Y) \cap X = (B(v_Yw_Y) \cup R(v_Yw_Y)) \cap X = Y$.
Let $P = \{p_Y: Y \subseteq X\}$, and consider the bipartite graph $G$ on vertex set $X \cup P$ where $xp_Y$ is an edge if $x \in W(p_Y)$ (or equivalently if $x \notin B(p_Y)$).
Since $P$ shatters $X$, the $2^{|X|}$ vertices in $P$ have all possible adjacencies with respect to $X$.
Thus, for any subset $X' \subseteq X$, the number of common neighbors in $P$ of the vertices of $X'$ is $2^{|X| - |X'|}$, and therefore there is no complete bipartite subgraph of $G$ with more than $2^{|X|-1}$ edges.

We show that $G$ has a complete bipartite subgraph with at least $\varepsilon \cdot |X| \cdot 2^{|X|-1}$ edges.
Consider the following experiment : take a uniformly random set $S \in \mF$ and let $Z$ be the random variable which counts the number of pairs $(x, p_Y) \in X \times P$ with $x = (x^1, \ldots, x^t)$ and $p_Y = (v_Y, w_Y)$ such that there exists $i \in [t]$ such that $S \in \mF_{v_Yx^i} \cup \mF_{w_Yx^i}$.
For every edge $xp_Y$ in $\Gamma$ with $x = (x^1, \ldots, x^t)$ and $p_Y = (v_Y, w_Y)$, we have that $x \in W(p_Y)$, which means that there exists $i \in [t]$ such that either $|\mF_{v_Yx^i}| \geq \varepsilon |\mF|$ or $|\mF_{w_Yx^i}| \geq \varepsilon |\mF|$, so the probability that $S \in \mF_{v_Yx^i} \cup \mF_{w_Yx^i}$ is at least $\varepsilon$.
Therefore, $\mathbb{E}[Z] \geq \varepsilon \cdot e(X, P) = \varepsilon \cdot |X| \cdot 2^{|X|-1}$, so there exists $S \in \mF$ such that $Z \geq \varepsilon \cdot |X| \cdot 2^{|X|-1}$.
Then, let $X'  = \{(x^1, \ldots, x^t) \in X : S \cap \{x^1, \ldots, x^t\} \neq \emptyset\}$ and $P' = \{(v_Y, w_Y) \in P : S \cap \{v_Y, w_Y\} \neq \emptyset\}$.
For every $x = (x^1, \ldots, x^t) \in X'$ and $p_Y = (v_Y, w_Y) \in P'$, there exists $i \in [t]$ such that $x^i \in S$, so $\mF_{v_Yx^i} \cup \mF_{w_Yx^i} \neq \emptyset$ so $x \notin B(p_Y)$, i.e. $xp_Y \in E(\Gamma)$.
Since $|X'| \cdot |P'| = Z \geq \varepsilon \cdot |X| \cdot 2^{|X| - 1}$ then $G$ has a complete bipartite subgraph with at least $\varepsilon \cdot |X| \cdot 2^{|X|-1}$ edges.

Putting together the upper bound and the lower bound, we get $\varepsilon \cdot |X| \cdot 2^{|X|-1} \leq 2^{|X|-1}$ so $|X| \leq \frac{1}{\varepsilon}$.
\end{proof}

\end{document}